\newif\ifcameraready
\newif\ifshowcomment

\showcommentfalse
\camerareadyfalse %

\ifcameraready
\documentclass[dvipsnames,sigconf,screen]{acmart}
\else
\documentclass[dvipsnames,authorversion,sigconf,screen]{acmart}
\subtitle{\normalsize\rmfamily The extended version. Updated on \today.}
\fi

\copyrightyear{2020}
\acmYear{2020} \setcopyright{acmcopyright}
\acmConference[CCS '20]{Proceedings of the 2020 ACMSIGSAC Conference on Computer and Communications Security}{November 9--13,2020}{Virtual Event, USA}
\acmBooktitle{Proceedings of the 2020 ACM SIGSAC Conference on Computer andCommunications Security (CCS '20), November 9--13, 2020, Virtual Event, USA}
\acmPrice{15.00}
\acmDOI{10.1145/3372297.3417239}\acmISBN{978-1-4503-7089-9/20/11}

\author{Fan Zhang}
\affiliation{Cornell Tech}
\authornote{Also affiliated with IC3, The Initiative for CryptoCurrencies \& Contracts.}
\author{Deepak Maram}
\affiliation{Cornell Tech}
\authornotemark[1]
\author{Harjasleen Malvai}
\affiliation{Cornell University}
\authornotemark[1]
\author{Steven Goldfeder}
\affiliation{Cornell Tech}
\authornotemark[1]
\author{Ari Juels}
\affiliation{Cornell Tech}
\authornotemark[1]

\usepackage{microtype}
\usepackage[ff,sets,adversary,advantage,keys,primitives,operators,lambda,logic,notions]{cryptocode}
\usepackage{boxedminipage}
\usepackage{comment}
\usepackage[utf8x]{inputenc}
\usepackage{nicefrac}
\usepackage{subcaption} %

\makeatletter
\DeclareRobustCommand*\cal{\@fontswitch\relax\mathcal}
\makeatother

\usepackage{bm} %
\usepackage{float}
\usepackage{color}
\usepackage{ucs}
\usepackage{xspace}
\usepackage{url}
\usepackage{tabularx}
\usepackage{booktabs}
\usepackage{cleveref}
\crefname{definition}{Def.}{Def.}
\Crefname{definition}{Definition}{Definitions}
\crefname{section}{Sec.}{Sec.}
\Crefname{section}{Section}{Sections}
\crefname{appendix}{App.}{App.}
\Crefname{appendix}{Appendix}{Appendices}

\definecolor{mygreen}{rgb}{0,0.6,0}
\definecolor{mygray}{rgb}{0.5,0.5,0.5}
\definecolor{mymauve}{rgb}{0.58,0,0.82}

\usepackage{listings}
\usepackage[nomessages]{fp}%

\usepackage{enumitem} %

\usepackage{pstricks,pst-node,pst-tree}
\usepackage{varwidth}
\usepackage[inference]{semantic}
\usepackage{setspace}

\psset{showbbox=false,treemode=D,linewidth=0.3pt,treesep=4ex,levelsep=1cm}

\usepackage{tikz}
\usetikzlibrary{arrows,shapes}
\usetikzlibrary{shapes.symbols}
\usetikzlibrary{shapes.multipart}
\usetikzlibrary{positioning}
\usetikzlibrary{arrows}
\usetikzlibrary{arrows.spaced}
\usetikzlibrary{matrix}

\tikzstyle{access} = [draw,semithick,fill=blue!20]

\newtheorem{claim}{Claim}

\colorlet{party}{Brown}
\colorlet{randomness}{Maroon}
\colorlet{protocol}{Black}
\colorlet{string}{BlueViolet}
\colorlet{entry}{NavyBlue}
\colorlet{public}{OliveGreen}
\colorlet{private}{Plum}

\usepackage[final]{pdfpages}

\newcommand{\stringlitt}[1]{\textcolor{string}{\text{``#1''}}}
\newcommand{\msgok}{\stringlitt{ok}}
\newcommand{\msgprove}{\stringlitt{prove}}
\newcommand{\msgproof}{\stringlitt{proof}}

\newcommand{\msgfinP}{\stringlitt{proverFinished}}
\newcommand{\msgfinS}{\stringlitt{serverFinished}}

\newcommand{\onrecv}{\textcolor{entry}{\bf On receiving}\xspace}
\newcommand{\oninit}{\textcolor{entry}{\bf On initialization}\xspace}

\renewcommand{\pccomment}[1]{\text{\textcolor{gray}{\scriptsize // #1}}}
\newcommand{\protocol}[3][\columnwidth]{
  \begin{boxedminipage}[t]{#1}
    \begin{center}
      \scriptsize{\textbf{#2}}
      \vspace{1em}
    \end{center}
    \procedure[mode=text, codesize=\footnotesize]{}{
      #3
    }
  \end{boxedminipage}
}

\newcommand{\protocolNoBox}[2]{
  \begin{center}
    \scriptsize{\textbf{#1}}
    \vspace{1em}
  \end{center}
  \procedure[mode=text, codesize=\footnotesize]{}{
    #2
  }
}

\newcommand{\aes}{\ensuremath{\mathsf{AES}}}
\newcommand{\cbc}{\ensuremath{\mathsf{CBC}}}
\newcommand{\gcm}{\ensuremath{\mathsf{GCM}}}

\newcommand{\gf}{\ensuremath{\mathsf{GF}}}

\newcommand{\boldhead}[1]{\vspace{2pt}\noindent\textbf{#1.}}
\newcommand{\VEC}[1]{\ensuremath{\bm{#1}}}

\renewcommand{\mac}{\ensuremath{\mathsf{MAC}}}
\newcommand{\mackey}{\ensuremath{\key^{\text{\mac}}}}
\newcommand{\enckey}{\ensuremath{\key^{\text{\enc}}}}

\renewcommand{\prover}{\textcolor{party}{\mathcal{P}}}
\renewcommand{\verifier}{\textcolor{party}{\mathcal{V}}}
\newcommand{\tlsserver}{\textcolor{party}{\mathcal{S}}}
\newcommand{\PUB}{\textcolor{public}{\theta_p}}
\newcommand{\PRI}{\textcolor{private}{\theta_s}}

\newcommand{\ZKP}[1]{\mathsf{ZK\text{-}PoK}\{#1\}}
\newcommand{\hmac}{\mathsf{HMAC}}

\newcommand{\REV}{\ensuremath{\mathsf{Reveal}}}
\newcommand{\RED}{\ensuremath{\mathsf{Redact}}}

\newcommand{\grammar}{\mathcal{G}}

\newcommand{\cons}{\mathsf{cons}}

\newcommand{\start}{\texttt{\underline{start}}\xspace}
\newcommand{\midd}{\texttt{\underline{middle}}\xspace}
\newcommand{\End}{\texttt{\underline{end}}\xspace}

\newcommand{\idealOracle}{\mathcal{F}_{\text{Oracle}}}
\newcommand{\PRED}{\mathsf{Stmt}}
\newcommand{\query}{\mathsf{Query}}

\newcommand{\PROT}{\mathsf{Prot}}

\newcommand{\mackeyP}{\mackey_{\prover}}
\newcommand{\mackeyV}{\mackey_{\verifier}}
\newcommand{\CERT}{\mathsf{cert}}
\newcommand{\MTA}{\mathsf{MtA}}
\newcommand{\idealtwopc}{\mathcal{F}_\text{2PC}}

\newcommand{\randc}{r_c}
\newcommand{\rands}{r_s}
\newcommand{\protadd}{\mathsf{ECtF}}
\newcommand{\secretP}{s_P}
\newcommand{\secretV}{s_V}
\newcommand{\pubkeyS}{Y_S}
\newcommand{\keyP}{\key_{\prover}}
\newcommand{\keyV}{\key_{\verifier}}

\newcommand{\CBC}{\mathsf{CBC}}

\newcommand{\idealZK}{\mathcal{F}_{\text{ZK}}}
\newcommand{\idealTwoPC}{\mathcal{F}_\text{2PC}}

\newcommand{\decoPROT}{\PROT_{\systemname}}

\newcommand{\gfgcm}{\FF_{2^{128}}}
\newcommand{\INC}{\mathsf{inc}}
\newcommand{\TAG}{\mathsf{Tag}}
\newcommand{\GHASHPOLY}{P_{\VEC{A} \| \VEC{C}  \| \ell_A \| \ell_C}}
\newcommand{\idealPPGCM}{\mathcal{F}_\text{PP}}
\newcommand{\idealAesSameMsg}{\mathcal{F}_\text{AES-EqM}}
\newcommand{\idealAesSameMsgOneOutput}{\mathcal{F}_\text{AES-EqM-Asym}}

\newcommand{\receiver}{\mathsf{receiver}}

\newcommand{\sid}{\ensuremath{\mathsf{sid}}}
\newcommand{\rec}{\mathsf{receiver}}

\usepackage[subtle,tracking=normal, wordspacing=normal,mathdisplays=tight]{savetrees} %

\ifcameraready
\usepackage[font=small,skip=4pt]{caption}
\fi

\newcommand{\vspacecamerareadyonly}[1]{\ifcameraready\vspace{#1}\fi}

\newcommand{\posthandshake}{query execution\xspace}
\newcommand{\Posthandshake}{Query execution\xspace} 
\ifshowcomment
\newcommand{\fanz}[1]{\textsf{\color{red}{[Fan: {#1}]}}}
\newcommand{\ari}[1]{\textsf{\color{blue}{[Ari: {#1}]}}}
\newcommand{\deepak}[1]{\textsf{\color{violet}{[Deepak: {#1}]}}}
\newcommand{\jasleen}[1]{\textsf{\color{cyan}{[Jasleen: {#1}]}}}
\newcommand{\steven}[1]{\textsf{\color{brown}{[Steven: {#1}]}}}
\else
\newcommand{\fanz}[1]{}
\newcommand{\ari}[1]{}
\newcommand{\deepak}[1]{}
\newcommand{\jasleen}[1]{}
\newcommand{\steven}[1]{}
\fi

\newcommand{\systemname}{\ensuremath{\mathsf{DECO}}\xspace}

\title{\systemname: Liberating Web Data Using Decentralized Oracles for TLS}

\begin{document}
\ifcameraready
\fancyhead{}
\fi

\begin{abstract}

Thanks to the widespread deployment of TLS, users can access private data over channels with end-to-end confidentiality and integrity. What they cannot do, however, is prove to third parties the {\em provenance} of such data, i.e., that it genuinely came from a particular website.
Existing approaches either introduce undesirable trust assumptions or require server-side modifications.

Users' private data is thus locked up at its point of origin. Users cannot export data in an integrity-protected way to other applications without help and permission from the current data holder.

We propose \systemname (short for \underline{dec}entralized \underline{o}racle) to address the above problems. \systemname allows users to prove that a piece of data accessed via TLS came from a particular website and optionally prove statements about such data in zero-knowledge, keeping the data itself secret. \systemname is the first such system that works without trusted hardware or server-side modifications.

\systemname can liberate private data from centralized web-service silos, making it accessible to a rich spectrum of applications.
To demonstrate the power of \systemname, we implement three applications that are hard to achieve without it: a private financial instrument using smart contracts, converting legacy credentials to anonymous credentials, and verifiable claims against price discrimination.
\end{abstract} 
\maketitle

\section{Introduction}

TLS is a powerful, widely deployed protocol that allows users to access web data over confidential, integrity-protected channels. But TLS has a serious limitation: it doesn't allow a user to prove to third parties that a piece of data she has accessed authentically came from a particular website. As a result, data use is often restricted to its point of origin,
curtailing data portability by users, a right acknowledged by recent regulations such as GDPR~\cite{gdpr}.

Specifically, when a user accesses data online via TLS, she cannot securely {\em export} it, without help (hence permission) from the current data holder.
Vast quantities of private data are thus intentionally or unintentionally locked up in the ``deep web''---the part of the web that isn't publicly accessible.

To understand the problem, suppose Alice wants to prove to Bob that she's over 18.
Currently, age verification services~\cite{age-veri} require users to upload IDs and detailed personal information, which raises privacy concerns.
But various websites, such as company payroll records or DMV websites, in principle store and serve verified birth dates.
Alice could send a screenshot of her birth date from such a site, but this is easily forged. And even if the screenshot could somehow be proven authentic, it would leak information---revealing her exact birth date, not just that she's over 18.

Proposed to prove provenance of online data to smart contracts,
{\em oracles} are a step towards exporting TLS-protected data to other systems with provenance and integrity assurances.
Existing schemes, however, have serious limitations. They either only work with deprecated TLS versions and offer no privacy from the oracle (e.g., TLSNotary~\cite{tlsnotary}) or rely on trusted hardware (e.g., Town Crier~\cite{zhang2016town}), against which various attacks have recently emerged, e.g.,~\cite{DBLP:conf/uss/BulckMWGKPSWYS18}.

Another class of oracle schemes assumes server-side cooperation, mandating that servers install TLS extensions (e.g.,~\cite{ritzdorf2017tls}) or change application-layer logic~(e.g.,~\cite{http-signatures,http-origin-signed-responses}).
Server-facilitated oracle schemes suffer from two fundamental problems.
First, they break legacy compatibility, causing a significant barrier to wide adoption.
Moreover, such solutions only provide \emph{conditional} exportability because the web servers have the sole discretion to determine which data can be exported, and can censor export attempts at will.
A mechanism that allows users to export {\em any} data they have access to would enable a whole host of currently unrealizable applications.

\subsection{\systemname}

To address the above problems, we propose
\systemname, a \underline{dec}entralized \underline{o}racle for TLS.
Unlike oracle schemes that require per-website support, \systemname is source-agnostic and supports {\em any} website running standard TLS. Unlike solutions that rely on websites' participation, \systemname requires no server-side cooperation. Thus a single instance of \systemname could enable {\em anyone} to become an oracle for {\em any} website.

\systemname  makes rich Internet data accessible with authenticity and privacy assurances to a wide range of applications, including
ones that cannot access the Internet such as smart contracts.
\systemname could fundamentally shift today's model of web data dissemination by providing private data delivery with an option for transfer to third parties or public release.
This technical capability highlights potential future legal and regulatory challenges, but also anticipates the creation and delivery of appealing new services. Importantly, \systemname does not require trusted hardware, unlike alternative approaches that could achieve a similar vision, e.g.,~\cite{delegatee,zhang2016town}.

At a high level, the prover commits to a piece of data $D$ and proves to the verifier that $D$ came from a TLS server $S$ and optionally a statement $\pi_D$ about $D$.
E.g., in the example of proving age, the statement $\pi_D$ could be the predicate ``$D=y/m/d$ is Alice's date of birth and the current date - $D$ is at least 18 years.''

Informally, \systemname achieves {\em authenticity}: The verifier is convinced only if the asserted statement about $D$ is true and $D$ is indeed obtained from website $S$. \systemname also provides {\em privacy} in that the verifier only learns the that the statement $\pi_D$ holds for some $D$ obtained from $S$.

\vspacecamerareadyonly{-1mm}
\subsection{Technical challenges}

Designing \systemname with the required security and practical performance, while using legacy-(TLS)-compatible primitives, introduces several important technical challenges. The main challenge stems from the fact that TLS generates symmetric encryption and authentication keys that are {\em shared} by the client (prover in \systemname) and web server. Thus, the client can {\em forge} arbitrary TLS session data, in the sense of signing the data with valid authentication keys.

To address this challenge, \systemname introduces a novel {\em three-party handshake} protocol among the prover,
verifier, and web server that creates an {\em unforgeable commitment} by the prover to the verifier on a piece of TLS session data $D$.
The verifier can check that $D$ is authentically from the TLS server.
From the prover's perspective, the three-party handshake preserves the security of TLS in presence of a malicious verifier.

\boldhead{Efficient selective opening}
After committing to $D$, the prover proves statements about the commitment.
Although arbitrary statements can be supported in theory, we optimize for what are likely to be the most popular applications---revealing only substrings of the response to the verifier. We call such statements {\em selective opening}. Fine-grained selective opening allows users to hide sensitive information and reduces the input length to the subsequent proofs.

A na\"ive solution would involve expensive verifiable decryption of TLS records using generic zero-knowledge proofs (ZKPs),
but we achieve an orders-of-magnitude efficiency improvement by exploiting the TLS record structure.
For example, a direct implementation of verifiable decryption of a TLS record would involve proving correct execution of a circuit of 1024 AES invocations in zero-knowledge, whereas by leveraging the MAC-then-encrypt structure of CBC-HMAC, we achieve the same with only 3 AES invocations.

\boldhead{Context integrity}
Selective opening allows the prover to only reveal a substring $D'$ of the server's response $D$. However, a substring may mean different things depending on when it appears and a malicious prover could cheat by quoting out of context.
Therefore we need to prove not just that $D'$ appears in $D$, but that it appears in the expected context, i.e., $D'$ has {\em context integrity} with respect to $D$.
(Note that this differs from ``contextual integrity'' in privacy theory~\cite{nissenbaum2009privacy}.)

Context-integrity attacks can be thwarted if the session content is structured and can be parsed. Fortunately most web data takes this form (e.g., in JSON or HTML).
A generic solution is to parse the entire session and prove that the revealed part belongs to the necessary branch of a parse tree.
But, under certain constraints that web data generally satisfies, parsing the entire session is not necessary.
We propose a novel {\em two-stage parsing scheme} where the prover pre-processes the session content, and only parses the outcome that is usually much smaller.
We draw from the definition of equivalence of programs, as used in programming language theory, to build a formal framework to reason about the security of two-stage parsing schemes. We provide several practical realizations for specific grammars. Our definitions and constructions generalize to other oracles too.
For example, it could prevent a generic version of the content-hidden attack mentioned in~\cite{ritzdorf2017tls}.

\subsection{Implementation and evaluation}

We designed and implemented \systemname as a complete end-to-end system.
To demonstrate the system's power,
we implemented three applications:
1) a confidentiality-preserving {\em financial instrument} using smart contracts;
2) converting legacy credentials to {\em anonymous credentials}; and
3) verifiable claims against {\em price discrimination}.

Our experiments with these applications show that \systemname is highly efficient. For example, for TLS 1.2 in the WAN setting, online time is 2.85s to perform the three-party handshake and 2.52s for 2PC query execution. It takes 3-13s to generate zero-knowledge proofs for the applications described above. More details are in~\cref{sec:eval}.%

\boldhead{Contributions}
In summary, our contributions are as follows:

\begin{itemize}[leftmargin=*]
\item We introduce {\bf \systemname}, a provably secure decentralized oracle scheme, along with an implementation and performance evaluation. \systemname is the first oracle scheme for modern TLS versions (both 1.2 and 1.3) that doesn't require trusted hardware or server-side modifications. We provide an overview of the protocol in~\cref{sec:overview} and specify the full protocol in~\cref{sec:deco protocols}.

\item {\bf Selective opening}: In~\cref{sec:seletive opening}, we introduce a broad class of statements for TLS records that can be proven efficiently in zero-knowledge. They allow users to open only substrings of a session-data commitment. The optimizations achieve substantial efficiency improvement over generic ZKPs.

\item {\bf Context-integrity attacks and mitigation}: We identify a new class of context-integrity attacks universal to privacy-preserving oracles (e.g.~\cite{ritzdorf2017tls}).
In~\cref{subsec:two stage parsing}, we introduce
our mitigation involving a novel, efficient two-stage parsing scheme, along with a formal security analysis, and several practical realizations.

\item {\bf Security definitions and proofs}:
Oracles are a key part of the smart contract ecosystem, but a coherent security definition has been lacking.
We formalize and strengthen existing oracle schemes and present a formal security definition using an ideal functionality in~\cref{sec:sec def}. We prove the functionality is securely realized by our protocols in~\cref{sec:securityproofs}.

\item {\bf Applications and evaluation}: In~\cref{sec:applications}, we present three representative applications that showcase \systemname's capabilities, and evaluate them in~\cref{sec:eval}.
\item {\bf Legal and compliance considerations}: \systemname can export data from websites without their explicit approval or even awareness. We discuss the resulting legal and compliance issues in~\cref{sec:legal}.
\end{itemize} %
\section{Background}
\label{sec:background}

\subsection{Transport Layer Security (TLS)}

We now provide necessary background on the TLS handshake and record protocols on which \systemname builds.

TLS is a family of protocols that provides privacy and data integrity between two communicating applications.
Roughly speaking, it consists of two protocols: a handshake protocol that sets up the session using asymmetric cryptography, establishing shared client and server keys for the next protocol, the record protocol, in which data is transmitted with confidentiality and integrity protection using symmetric cryptography.

\boldhead{Handshake}
In the handshake protocol, the server and client first agree on a set of cryptographic algorithms (also known as a cipher suite). They then authenticate each other (client authentication optional), and finally securely compute a shared secret to be used for the subsequent record protocol.

\systemname supports the recommended elliptic curve DH key exchange with ephemeral secrets (ECDHE~\cite{ecc-tls}).%

\boldhead{Record protocol} To transmit application-layer data (e.g., HTTP messages) in TLS, the record protocol first fragments the application data $\VEC{D}$ into fixed sized plaintext {\em records} $\VEC{D}=(D_1, \cdots,D_n)$. Each record is usually padded to a multiple of blocks (e.g., 128 bits).
The record protocol then optionally compresses the data, applies a MAC, encrypts, and transmits the result. Received data is decrypted, verified, decompressed, reassembled, and then delivered to higher-level protocols.
The specific cryptographic operations depend on the negotiated ciphersuite.
\systemname supports the AES cipher in two commonly used modes: CBC-HMAC and GCM.
We refer readers to~\cite{RFC5246} for how these primitives are used in TLS.

\boldhead{Differences between TLS 1.2 and 1.3}
Throughout the paper we focus on TLS 1.2 and discuss how to generalize our techniques to TLS 1.3 in \cref{para:tls1.3}. Here we briefly note the major differences between these two TLS versions.
TLS 1.3 removes the support for legacy non-AEAD ciphers. The handshake flow has also been restructured. All handshake messages after the ServerHello are now encrypted. Finally, a different key derivation function is used. For a complete description, see~\cite{RFC8446}.

\subsection{Multi-party computation}
Consider a group of $n$ parties ${\cal P}_1,\ldots,{\cal P}_n$, each of whom holds some secret $s_i$. Secure multi-party computation (MPC) allows them to jointly compute $f(s_i,\cdots,s_n)$ without leaking any information other than the output of $f$, i.e., ${\cal P}_i$ learns nothing about $s_{j\neq i}$.
Security for MPC protocols generally considers an adversary that corrupts $t$ players and attempts to learn the private information of an honest player. Two-party computation (2PC) refers to the special case of $n=2$ and $t=1$. %
We refer the reader to~\cite{lindell2005secure} for a full discussion of the model and formal security definitions.

There are two general approaches to 2PC protocols. Garbled-circuit protocols based on Yao~\cite{yao1982protocols} encode $f$ as a boolean circuit, an approach best-suited for bitwise operations (e.g., SHA-256). Other protocols leverage {\em threshold secret sharing} and are best suited for arithmetic operations. The functions we compute in this paper using 2PC, though, include both bitwise and arithmetic operations. We separate them into two components, and use the optimized garbled-circuit protocol from \cite{DBLP:conf/ccs/WangRK17} for the bitwise operations and the secret-sharing based $\MTA$ protocol from \cite{gennaro2018fast} for the arithmetic operations.

\section{Overview}
\label{sec:overview}

In this section we state the problem we try to solve with  \systemname and present a high-level overview of its architecture.

\subsection{Problem statement: Decentralized oracles}

Broadly, we investigate protocols for building ``oracles,'' i.e., entities that can prove provenance and properties of online data.
The goal is to allow a prover $\prover$ to prove to a verifier $\verifier$ that a piece of data came from a particular website $\tlsserver$ and optionally prove statements about such data in zero-knowledge, keeping the data itself secret.
Accessing the data may require private input (e.g., a password) from $\prover$ and such private information should be kept secret from $\verifier$ as well.

We focus on servers running TLS, the most widely deployed security protocol suite on the Internet. However, TLS alone does not prove data provenance. Although TLS uses public-key signatures for authentication, it uses symmetric-key primitives to protect the integrity and confidentiality of exchanged messages, using a shared session key established at the beginning of each session. Hence $\prover$, who knows this symmetric key, cannot prove statements about cryptographically authenticated TLS data to a third party.

A web server itself could assume the role of an oracle, e.g., by simply signing  data. However, server-facilitated oracles would not only incur a high adoption cost, but also put users at a disadvantage: the web server could impose arbitrary constraints on the oracle capability.
We are interested in a scheme where anyone can prove provenance of any data she can access, without needing to rely on a single, central point of control, such as the web server providing the data.

We tackle these challenges by introducing {\em decentralized oracles} that don't rely on trusted hardware or cooperation from web servers. The problem is much more challenging than for previous oracles, as it precludes solutions that require servers to modify their code or deploy new software, e.g.,~\cite{ritzdorf2017tls}, or use of prediction markets, e.g.,~\cite{adler2018astraea,peterson2015augur}, while at the same time going beyond these previous approaches by supporting proofs on arbitrary predicates over data.
Another approach, introduced in~\cite{zhang2016town}, is to use trusted execution environments (TEEs) such as Intel SGX. The downside is that recent attacks~\cite{DBLP:conf/uss/BulckMWGKPSWYS18} may deter some users from trusting TEEs.

\boldhead{Authenticated data feeds for smart contracts}
\label{para:adf}
An important application of oracle protocols is to construct authenticated data feeds (ADFs, as coined in~\cite{zhang2016town}), i.e., data with verifiable provenance and correctness, for smart contracts.
Protocols such as~\cite{zhang2016town} generate ADFs by signing TLS data using a key kept secret in a TEE. However, the security of this approach relies on that of TEEs. Using multiple TEEs could help achieve stronger integrity, but not privacy. If a single TEE is broken, TLS session content, including user credentials, can leak from the broken TEE.%

\systemname operates in a different model. Since smart contracts can't participate in 2PC protocols, they must rely on oracle nodes to participate as $\verifier$ on their behalf. Therefore we envision \systemname being deployed in a decentralized oracle network similar to~\cite{SCCL:2017}, where a set of independently operated oracles are available for smart contracts to use.
Note that oracles running \systemname are trusted only for integrity, not for privacy. Smart contracts can further hedge against integrity failures by querying multiple oracles and requiring, e.g., majority agreement, as already supported in~\cite{SCCL:2017}.
We emphasize that \systemname's privacy is preserved even all oracles  are compromised.
Thus \systemname enables users to provide ADFs derived from private data to smart contracts while hiding private data from oracles.

\subsection{Notation and definitions}
\label{sec:sec def}

We use $\prover$ to denote the prover, $\verifier$ the verifier and $\tlsserver$ the TLS server.
We use letters in boldface  (e.g., $\VEC{M}$) to denote vectors and $M_i$ to denote the $i$th element in $\VEC{M}$. %

\begin{figure}
\centering
\protocol{Functionality $\idealOracle$ between $\tlsserver, \prover$ and $\verifier$}{
\textbf{Input:}
The prover $\prover$ holds some private input $\PRI$. The verifier $\verifier$ holds a query template $\query$ and a statement $\PRED$. \\
\textbf{Functionality:}
\begin{itemize}[leftmargin=*]
    \item If at any point during the session, a message $(\sid, \rec, m)$ with $\rec \in \{\tlsserver, \prover, \verifier\}$ is received from $\adv$, forward $(\sid, m)$ to $\receiver$ and forward any responses to $\adv$.
    \item Upon receiving input $(\sid, \query, \PRED)$ from $\verifier$, send $(\sid, \query, \PRED)$ to $\prover$. Wait for $\prover$ to reply with $\msgok$ and $\PRI$.
    \item Compute $Q=\query(\PRI)$ and send $(\sid, Q)$ to $\tlsserver$ and record its response $(\sid, R)$. Send $(\sid, |Q|, |R|)$ to $\adv$.
    \item Send $(\sid, Q, R)$ to $\prover$ and (\sid, $\PRED(R), \tlsserver)$ to $\verifier$.
\end{itemize}
}
\caption{The oracle functionality.}
\label{fig:ideal oracleUC}
\end{figure}

We model the essential properties of an oracle using an ideal functionality $\idealOracle$ in~\cref{fig:ideal oracleUC}. To separate parallel runs of $\idealOracle$, all messages are tagged with a unique \emph{session id} denoted $\sid$. We refer readers to~\cite{uc} for details of ideal protocol execution.%

$\idealOracle$ accepts a secret parameter $\PRI$ (e.g., a password) from $\prover$,  a query template $\query$ and a statement $\PRED$ from $\verifier$.
A query template is a function that takes $\prover$'s secret $\PRI$ and returns a complete query, which contains public parameters specified by $\verifier$.
An example query template would be $\query(\PRI)$ = ``stock price of GOOG on Jan 1st, 2020 with API key = $\PRI$''.
The prover $\prover$ can later prove that the query sent to the server is well-formed, i.e., built from the template, without revealing the secret.
The statement $\PRED$ is a function that $\verifier$ wishes to evaluate on the server's response. Following the previous example, as the response $R$ is a number, the following statement would compare it with a threshold: $\PRED(R)$ = ``$R > \$1,000$''.

After $\prover$ acknowledges the query template and the statement (by sending $\msgok$ and $\PRI$), $\idealOracle$ retrieves a response $R$ from $\tlsserver$ using a query built from the template. We assume an honest server, so $R$ is the ground truth. $\idealOracle$ sends $\PRED(R)$ and the data source to $\verifier$. %

As stated in~\cref{def:oracle protocol}, we are interested in decentralized oracles that don't require any server-side modifications or cooperation, i.e., $\tlsserver$ follows the unmodified TLS protocol.  %

\begin{definition}
A decentralized oracle protocol for TLS is a three-party protocol $\PROT=(\PROT_{\tlsserver}, \PROT_{\prover}, \PROT_{\verifier})$
such that 1) $\PROT$ realizes $\idealOracle$ and 2) $\PROT_{\tlsserver}$ is the standard TLS, possibly along with an application-layer protocol.
\label{def:oracle protocol}
\end{definition}

\boldhead{Adversarial model and security properties} We consider a static, malicious network adversary $\adv$. Corrupted parties may deviate arbitrarily from the protocol and reveal their states to $\adv$. As a network adversary, $\adv$ learns the message length from $\idealOracle$ since TLS is not length-hiding.
We assume $\prover$ and $\verifier$  choose and agree on an appropriate query (e.g., it should be idempotent for most applications) and statement according to the application-layer protocol run by $\tlsserver$.

For a given query $Q$, denote the server's honest response by $\tlsserver(Q)$. We require that security holds when either $\prover$ or $\verifier$ is corrupted.
The functionality $\idealOracle$ reflects the following security guarantees:

\begin{itemize}[leftmargin=*]
\item \emph{Prover-integrity:} A malicious $\prover$ cannot forge content provenance, nor can she cause $\tlsserver$ to accept invalid queries or respond incorrectly to valid ones. Specifically, if the verifier inputs $(\query, \PRED)$ and outputs $(b,\tlsserver)$, then $\prover$ must have sent $Q=\query(\PRI)$ to $\tlsserver$ in a TLS session, receiving  response $R=\tlsserver(Q)$ such that $b=\PRED(R)$.
\item \emph{Verifier-integrity:} A malicious $\verifier$ cannot cause $\prover$ to receive incorrect responses. Specifically, if $\prover$ outputs $(Q,R)$ then $R$ must be the server's response to query $Q$ submitted by $\prover$, i.e., $R=\tlsserver(Q)$.
\item \emph{Privacy:} A malicious  $\verifier$ learns only public information $(\query, \tlsserver)$ and the evaluation of $\PRED(R)$.
\end{itemize}

\subsection{A strawman protocol}
\label{sec:strawman}

We focus on two widely used representative TLS cipher suites: CBC-HMAC and AES-GCM.
Our technique generalizes to other ciphers (e.g., Chacha20-Poly1305, etc.) as well.
Throughout this section we use CBC-HMAC to illustrate the ideas, with discussion of GCM deferred to later sections.

TLS uses separate keys for each direction of communication. Unless explicitly specified, we don't distinguish between the two and use $\enckey$ and $\mackey$ to denote session keys for both directions.

In presenting our design of \systemname, we start with a strawman protocol and incrementally build up to the full protocol.

\newcommand{\SESSION}{\VEC{\hat M}}

\boldhead{A strawman protocol} A strawman protocol that realizes $\idealOracle$ between $(\prover, \verifier)$ is as follows. $\prover$ queries the server $\tlsserver$ and records all messages sent to and received from the server in $\VEC{\hat Q}=(\hat Q_1,\dots, \hat Q_n)$ and $\VEC{\hat R}=(\hat R_1, \dots, \hat R_n)$, respectively. Let $\SESSION=(\VEC{\hat Q}, \VEC{\hat R})$ and $(\mackey, \enckey)$ be the session keys.

She then proves in zero-knowledge that 1) each $\hat R_i$ decrypts to $R_i \| \sigma_i$, a plaintext record and a MAC tag; 2) each MAC tag $\sigma_i$ for $R_i$ verifies against $\mackey$; and 3) the desired statement evaluates to $b$ on the response, i.e., $b=\PRED(\VEC{R})$.
Using the now standard notation introduced in~\cite{camenisch1997efficient}, $\prover$ computes
\begin{multline*}
p_r = \ZKP{\enckey, \VEC{R}: \forall i\in[n], \dec(\enckey, \hat R_i)=R_i \|\sigma_i \\
\land \mathsf{Verify}(\mackey, \sigma_i, R_i)=1 \land \PRED(\VEC{R})=b }.
\label{eq:zkp}
\end{multline*}

She also proves that $\VEC{Q}$ is well-formed as $\VEC{Q}=\query(\PRI)$ similarly in a proof $p_q$ and sends $(p_q, p_r, \mackey, \SESSION, b)$ to $\verifier$.

Given that $\SESSION$ is an authentic transcript of the TLS session, the prover-integrity property seems to hold. Intuitively, CBC-HMAC ciphertexts bind to the underlying plaintexts, thus $\SESSION$ can be treated as secure commitments~\cite{DBLP:conf/crypto/GrubbsLR17} to the session data. That is, a given $\SESSION$ can only be opened (i.e., decrypted and MAC checked) to a unique message.
The binding property prevents $\prover$ from opening $\SESSION$ to a different message other than the original session with the server.

Unfortunately, this intuition is flawed.
The strawman protocol fails completely because it {\em cannot} ensure the authenticity of $\SESSION$. The prover $\prover$ has the session keys, and thus she can include the encryption of arbitrary messages in $\SESSION$. %

Moreover, the zero-knowledge proofs that $\prover$ needs to construct involve decrypting and hashing the entire transcript, which can be prohibitively expensive. For the protocol to be practical, we need to significantly reduce the cost.

\subsection{Overview of \systemname}
\label{sec:overview of deco}

The critical failing of our strawman approach is that $\prover$ learns the session key before she commits to the session. One key idea in \systemname is to withhold the MAC key from $\prover$ until {\em after} she commits.
The TLS session between $\prover$ and $\tlsserver$ must still provide confidentiality and integrity. Moreover, the protocol must not degrade performance below the requirements of TLS (e.g., triggering a timeout).

As shown in~\cref{fig:overview}, \systemname is a three-phase protocol. The first phase is a novel {\bf three-party handshake} protocol in which the prover $\prover$, the verifier $\verifier$, and the TLS server $\tlsserver$ establish session keys that are {\em secret-shared between $\prover$ and $\verifier$}. After the handshake is a {\bf \posthandshake} phase during which $\prover$ accesses the server following the standard TLS protocol, but with help from $\verifier$. After $\prover$ commits to the query and response, $\verifier$ reveals her key share. Finally, $\prover$ proves statements about the response in a {\bf proof generation} phase.

\begin{figure}[t]
\centering
\resizebox{1.0\columnwidth}{!}{
\begin{tikzpicture}
\tikzstyle{line} = [thick,>=latex];
\tikzstyle{comment} = [text centered];
\tikzstyle{protocol} = [rounded corners,inner sep=2mm, draw=black!50, fill=black!10,font=\sffamily\Large];
\tikzstyle{party} = [rounded corners, inner sep=2mm, minimum height=.8cm,rectangle];

\node[party] (prover) at (5, 0) {Prover $\prover$};
\node[party] (verifier) at (10, 0) {Verifier $\verifier$};
\node[party] (server) at (0, 0) {Server $\tlsserver$};
\draw[dashed] (server.south) -- (0, -6);
\draw[dashed] (prover.south) -- (5, -6);
\draw[dashed] (verifier.south) -- (10, -6);
\scoped[yshift=-1.2cm]{
\node[comment,protocol,minimum width=11cm] (tph) at (5, 0) {Three-party handshake (\cref{sec:handshake})};
\node[fill=white,comment,anchor=south] (ks) at (0,-1.4) {Session keys $\key$};
\node[fill=white,anchor=south] (kp) at (5,-1.4) {$\key_{\prover}$};
\node[fill=white,anchor=south] (kv) at (10,-1.4) {$\key_{\verifier}$};
\draw[->,line] (tph.south -| ks) -- (ks);
\draw[->,line] (tph.south -| kp) -- (kp);
\draw[->,line] (tph.south -| kv) -- (kv);
}%
\scoped[yshift=-3.5cm]{
\node[comment,protocol, minimum width=6cm] (post-hs) at (7.5,0) {\Posthandshake (\cref{sec:post handshake})};
\draw[->, line] (kp) -- (post-hs.north -| kp);
\draw[->, line] (kv) -- (post-hs.north -| kv);
\draw[<->,line] (post-hs.west) -- node[midway,above,comment]{Send $Q=\query(\PRI)$} node[midway,below,comment] {Receive response $R$} (0, 0);
}%
\scoped[yshift=-4.5cm]{
\draw[->,line] (5, 0) -- node[midway,above,comment]{commit to $(Q,R)$} (10, 0);
\draw[<-,line] (5, -.6) -- node[midway,above,comment]{$\key_{\verifier}$} (10, -.6);
\node (verify) at (3, -.6) {verify $R$ using $\key_{\prover}$ and $\key_{\verifier}$};
\node[comment,protocol, minimum width=6cm] (pgen) at (7.5,-1.4) {\sf Proof generation (\cref{sec:proof gen})};
\node[anchor=east] (kpri) at (4.2,-1.2) {$\PRI$};
\node[anchor=east] (kp) at (4.2,-1.6) {$\key_{\prover}$};
\draw[->, line] (kpri.east) -- (pgen.west |- kpri.east);
\draw[->, line] (kp.east) -- (pgen.west |- kp.east);
}%

\end{tikzpicture}
}
\caption{An overview of the workflow in \systemname.
\ifcameraready\else{\normalfont The protocol has three phases: a {\bf three-party handshake} phase to establish session keys in a special format to achieve unforgeability, a {\bf \posthandshake} phase where $\prover$ queries the server for data using a query built from the template with her private parameters $\PRI$, and finally a {\bf proof generation} phase in which $\prover$ proves that the query is well-formed and the response satisfies the desired condition.}\fi}
\label{fig:overview}
\end{figure}

\subsubsection{Three-party handshake}

Essentially, $\prover$ and $\verifier$ {\em jointly} act as a TLS client. They negotiate a shared session key with $\tlsserver$ in a secret-shared form. We emphasize that this phase, like the rest of \systemname, is completely transparent to $\tlsserver$, requiring no server-side modifications.

For the CBC-HMAC cipher suite, at the end of the three-party handshake, $\prover$ and $\verifier$ receive $\mackeyP$ and $\mackeyV$ respectively, while $\tlsserver$ receives $\mackey=\mackeyP + \mackeyV$. As with the standard handshake, both $\prover$ and $\tlsserver$ get the encryption key $\enckey$.

Three-party handshake can make the aforementioned session-data commitment unforgeable as follows. At the end of the session, $\prover$ first commits to the session in $\SESSION$ as before, then $\verifier$ reveals her share $\mackeyV$.
From $\verifier$'s perspective, the three-party handshake protocol ensures that a fresh MAC key (for each direction) is used for every session, despite the influence of a potential malicious prover, and that the keys are unknown to $\prover$ until she commits.
Without knowledge of the MAC key, $\prover$ cannot forge or tamper with session data before committing to it. The unforgeability of the session-data commitment in \systemname thus reduces to the unforgeability of the MAC scheme used in TLS.

Other ciphersuites such as GCM can be supported similarly. In GCM, a single key (for each direction) is used for both encryption and MAC. The handshake protocol similarly secret-shares the key between $\prover$ and $\verifier$.
The handshake protocol are presented in~\cref{sec:handshake}.

\subsubsection{\Posthandshake}
Since the session keys are secret-shared, as noted, $\prover$ and $\verifier$ execute an interactive protocol to construct a TLS message encrypting the query. $\prover$ then sends the message to $\tlsserver$ as a standard TLS client. For CBC-HMAC, they compute the MAC tag of the query, while for GCM they perform authenticated encryption. Note that the query is private to $\prover$ and should not be leaked to $\verifier$.
Generic 2PC would be expensive for large queries, so we instead introduce custom 2PC protocols that are orders-of-magnitude more efficient than generic solutions, as presented in \cref{sec:post handshake}.

As explained previously, $\prover$ commits to the session data  $\SESSION$ before receiving $\verifier$'s key share, making the commitment unforgeable. Then $\prover$ can verify the integrity of the response, and prove statements about it, which we present now.

\vspacecamerareadyonly{-1mm}
\subsubsection{Proof generation}
\label{subsub:proof_gen}

With unforgeable commitments, if $\prover$ opens the commitment $\SESSION$ completely (i.e., reveals the encryption key) then $\verifier$ could easily verify the authenticity of $\SESSION$ by checking MACs on the decryption.

Revealing the encryption key for $\SESSION$, however, would breach privacy: it would reveal {\em all} session data exchanged between $\prover$ and $\tlsserver$.
In theory, $\prover$ could instead prove any statement $\PRED$ over $\SESSION$ in zero knowledge (i.e., without revealing the encryption key). Generic zero-knowledge proof techniques, though, would be prohibitively expensive for many natural choices of $\PRED$.

\systemname instead introduces two techniques to support efficient proofs for a broad, general class of statement, namely {\em selective opening} of a TLS session transcript. Selective opening involves either {\em revealing} a substring to $\verifier$ or {\em redacting}, i.e., excising, a substring, concealing it from $\verifier$. %

As an example, \cref{fig:example} shows a simplified JSON bank statement for Bob. Suppose Bob ($\prover$) wants to reveal his checking account balance to $\verifier$. Revealing the decryption key for his TLS session would be undesirable: it would {\em also} reveal the entire statement, including his transactions.
Instead, using techniques we introduce, Bob can efficiently reveal only the substring in lines 5-7. Alternatively, if he doesn't mind revealing his savings account balance, he might redact his transactions after line 7.

The two selective opening modes, revealing and redacting substrings, are useful privacy protection mechanisms. They can also serve as pre-processing for a subsequent zero-knowledge proof. For example, Bob might wish to prove that he has an account with a balance larger than \$1000, without revealing the actual balance. He would then prove in zero knowledge a predicate (``balance > \$1000'') over the substring that includes his checking account balance.

Selective opening {\em alone}, however, is not enough for many applications. This is because the {\em context} of a substring affects its meaning. Without what we call {\em context integrity}, $\prover$ could cheat and reveal a substring that falsely appears to prove a claim to $\verifier$. For example, Bob might not have a balance above \$1000. After viewing his bank statement, though, he might in the same TLS session post a message to customer service with the substring {\tt "balance": \$5000} and then view his pending messages (in a form of reflection attack). He could then reveal this substring to fool $\verifier$.

Various sanitization heuristics on prover-supplied inputs to $\verifier$, e.g., truncating session transcripts, could potentially prevent some such attacks, but, like other forms of web application input sanitization, are fragile and prone to attack~\cite{scholte2011quo}.

Instead, we introduce a rigorous technique by which session data are explicitly but confidentially parsed. We call this technique {\em zero-knowledge two-stage parsing}. The idea is that $\prover$ parses $\SESSION$ locally in a first stage and then proves to $\verifier$ a statement in zero knowledge about constraints on a resulting substring. For example, in our banking example, if bank-supplied key-value stores are always escaped with a distinguished character $\lambda$, then Bob could prove a correct balance by extracting via local parsing and revealing to $\verifier$ a substring {\tt "balance": \$5000} preceded by $\lambda$. We show for a very common class of web API grammars (unique keys) that this two-phase approach yields much more efficient proofs than more generic techniques.

\Cref{sec:proof gen} gives more details on proof generation in \systemname.

\begin{figure}
\begin{lstlisting}
    {"name": "Bob",
    "savings a/c": {
        "balance": $5000
    },
    "checking a/c": {
        "balance": $2000
    },
    "transactions": {...}}
\end{lstlisting}
\caption{Example bank statement to demonstrate selective opening and context-integrity attacks.}
\label{fig:example}
\end{figure}
\lstset{numbers=none} %
\section{The \systemname protocol}
\label{sec:deco protocols}

We now specify the full \systemname protocol, which consists of a three-party handshake in \cref{sec:handshake}, followed by 2PC protocols for \posthandshake in \cref{sec:post handshake}, and a proof generation phase. We prove its security in~\cref{sec:full protocol}.

\subsection{Three-party handshake}
\label{sec:handshake}

The goal of the three-party handshake (3P-HS) is to secret-share between the prover $\prover$ and verifier $\verifier$ the session keys used in a TLS session with server $\tlsserver$, in a way that is completely transparent to $\tlsserver$.
We first focus on CBC-HMAC for exposition, then adapt the protocol to support GCM.%

As with the standard TLS handshake, 3P-HS is two-step: first, $\prover$ and $\verifier$ compute additive shares of a secret \(Z\in EC(\FF_p)\) shared with the server through a TLS-compatible key exchange protocol.
ECDHE is the recommended and the focus here;
second, $\prover$ and $\verifier$ derive secret-shared session keys by securely evaluating the TLS-PRF~\cite{RFC5246} with their shares of $Z$ as inputs.
The full protocol is specified in~\cref{fig:3party-hs-protocol}.
Below we give text descriptions so formal specifications are not required for understanding.

\subsubsection{Step 1: key exchange}
Let $EC(\FF_p)$ denote the EC group used in ECDHE and $G$ its generator.

The prover $\prover$ initiates the handshake by sending a regular TLS handshake request and a random nonce $\randc$ to $\tlsserver$ (in the ClientHello message).
On receiving a certificate, the server nonce $\rands$, and a signed ephemeral DH public key $\pubkeyS=s_S\cdot G$ from $\tlsserver$ (in the ServerHello and ServerKeyExchange messages), $\prover$ checks the certificate and the signature and forwards them to $\verifier$.
After performing the same check, $\verifier$ samples a secret $\secretV$ and sends her part of the DH public key $Y_V=\secretV \cdot G$ to $\prover$, who then samples another secret $\secretP$ and sends the combined DH public key $Y_P=\secretP \cdot G + Y_V$ to $\tlsserver$.

Since the server $\tlsserver$ runs the standard TLS, $\tlsserver$ will compute a DH secret as $Z=s_S \cdot Y_P$.
$\prover$ (and $\verifier$) computes its share of $Z$ as $Z_P=s_P \cdot Y_S$ (and $Z_V=s_V \cdot Y_{S}$). Note that $Z=Z_P + Z_V$ where $+$ is the group operation of  $EC(\FF_p)$. Assuming the discrete logarithm problem is hard in the chosen group, $Z$ is unknown to either party.

\subsubsection{Step 2: key derivation}
Now that $\prover$ and $\verifier$ have established additive shares of $Z$ (in the form of {\em EC points}), they proceed to derive session keys by evaluating the TLS-PRF~\cite{RFC5246} keyed with the $x$ coordinate of $Z$.

A technical challenge here is to harmonize arithmetic operations (i.e., addition in $EC(\FF_p)$) with bitwise operations (i.e., TLS-PRF) in 2PC.
It is well-known that boolean circuits are not well-suited for arithmetic in large fields.
As a concrete estimate, an EC Point addition resulting in just the $x$ coordinate involves 4 subtractions, one modular inversion, and 2 modular multiplications.
An estimate of the AND complexity based on the highly optimized circuits of~\cite{demmler2015automated} results in over 900,000 AND gates just for the subtractions, multiplications, and modular reductions---not even including inversion, which would require running the Extended Euclidean algorithm inside a circuit.

Due to the prohibitive cost of adding EC points in a boolean circuit, $\prover$ and $\verifier$ convert the additive shares of an EC point in $EC(\FF_p)$ to additive shares of its $x$-coordinate in $\FF_p$, using the $\protadd$ protocol presented below.
Then the boolean circuit just involves adding two numbers in $\FF_p$, which can be done with only $\sim\!3|p|$ AND gates, that is $\sim\!768$ AND gates in our implementation where $p$ is 256-bit. %

\boldhead{$\protadd$: Converting shares in $EC(\FF_p)$ to shares in $\FF_p$}
The inputs to an $\protadd$ protocol are two EC points $P_1, P_2 \in EC(\FF_p)$, denoted $P_i=(x_i, y_i)$.
Suppose $(x_s,y_s)=P_1 \star P_2$ where $\star$ is the EC group operation, the output of the protocol is $\alpha, \beta \in \FF_p$ such that $\alpha + \beta =x_s$. Specifically, for the curve we consider,
$x_s=\lambda^2 - x_1 - x_2$ where $\lambda=\nicefrac{(y_2 -y_1)}{(x_2 - x_1 )}$.
Shares of the $y_s$ can be computed similarly but we omit that since TLS only uses the $x_s$.

$\protadd$ uses a Multiplicative-to-Additive ($\MTA$) share-conversion protocol as a building block. We use $\alpha, \beta := \MTA(a, b)$ to denote a run of $\MTA$ between Alice and Bob with inputs $a$ and $b$ respectively. At the end of the run, Alice and Bob receive $\alpha$ and $\beta$ such that $a \cdot b = \alpha + \beta$.
The protocol can be generalized to handle vector inputs without increasing the communication complexity. Namely for vectors $\VEC{a}, \VEC{b} \in \FF_p^n$, if $\alpha, \beta := \MTA(\VEC{a}, \VEC{b})$, then $\langle \VEC{a}, \VEC{b} \rangle = \alpha + \beta$.
See, e.g., ~\cite{gennaro2018fast} for a Paillier~\cite{DBLP:conf/eurocrypt/Paillier99}-based construction. 

Now we specify the protocol of $\protadd$.
$\protadd$ has two main ingredients.
Let $[a]$ denote a 2-out-of-2 sharing of $a$, i.e., $[a]=(a_1, a_2)$ such that party $i$ has $a_i$ for $i \in \set{1,2}$ while $a=a_1 + a_2$.
The first ingredient is share inversion: given $[a]$, compute $[a^{-1}]$. As shown in \cite{gennaro2018fast}, we can use the inversion protocol of Bar-Ilan and Beaver \cite{bar1989non} together with $\MTA$ as follows: party $i$ samples a random value $r_i$ and executes $\MTA$ to compute $\delta_1, \delta_2 :=\MTA((a_1, r_1), (r_2, a_2))$. Note that $\delta_1 + \delta_2 = a_1\cdot r_2 + a_2 \cdot r_1$. Party $i$ publishes $v_i = \delta_i + a_i \cdot r_i$ and thus both parties learn $v=v_1 + v_2$. Finally, party $i$ outputs $\beta_i = r_i \cdot v^{-1}$. The protocol computes a correct sharing of $a^{-1}$ because $\beta_1 + \beta_2 = a^{-1}$. Moreover, the protocol doesn't leak $a$ to any party assuming $\MTA$ is secure. In fact, party $i$'s view consists of $(a_1+a_2)(r_1+r_2)$, which is uniformly random since $r_i$ is uniformly random.

The second ingredient is share multiplication: compute $[ab]$ given $[a], [b]$. $[ab]$ can be computed using $\MTA$ as follows: parties execute $\MTA$ to compute $\alpha_1, \alpha_2$ such that $\alpha_1 + \alpha_2 = a_1\cdot b_2 + a_2\cdot b_1$. Then, party $i$ outputs $m_i = \alpha_i + a_i \cdot b_i$. The security and correctness of the protocol can be argued similarly as above.

Combining these two ingredients, \cref{fig:ecadd} in the Appendix presents the $\protadd$ protocol, with communication complexity 8 ciphertexts.

\boldhead{Secure evaluation of the TLS-PRF}
Having computed shares of the $x$-coordinate of $Z$, the so called premaster secret in TLS, in $\protadd$, $\prover$ and $\verifier$ evaluate the TLS-PRF in 2PC to derive session keys. Beginning with the SHA-256 circuit of~\cite{Campanelli:2017:ZCP:3133956.3134060}, we hand-optimized the TLS handshake circuit resulting in a circuit with total AND complexity of 779,213.

\boldhead{Adapting to support GCM}
For GCM, a single key (for each direction) is used for both encryption and MAC. Adapting the above protocol to support GCM in TLS 1.2 is straightforward. The first step would remain identical, while output of the second step needs to be truncated, as GCM keys are shorter.
\boldhead{Adapting to TLS 1.3}
\label{para:tls1.3}
The specification of TLS 1.3~\cite{RFC8446} has been recently published.
To support TLS 1.3, the 3P-HS protocol must be adapted to a new handshake flow and a different key derivation circuit.
Notably, all handshake messages after the ServerHello are now {\em encrypted}. A naïve strategy would be to decrypt them in 2PC, which would be costly as certificates are usually large.
However, thanks to the key independence property of TLS 1.3~\cite{tls13handshake}, we can construct a 3P-HS protocol of similar complexity to that for TLS 1.2, as outlined in~\cref{app:tls13details}. %
\subsection{\Posthandshake}
\label{sec:post handshake}

\newcommand{\macstoc}{2PC-MAC$_{StC}$}
\newcommand{\macctos}{2PC-MAC$_{CtS}$}

After the handshake, the prover $\prover$ sends her query $Q$ to the server $\tlsserver$ as a standard TLS client, but with help from the verifier $\verifier$. %
Specifically, since session keys are secret-shared, the two parties need to interact and execute a 2PC protocol to construct TLS records encrypting $Q$.
Although generic 2PC would in theory suffice, it would be expensive for large queries.
We instead introduce custom 2PC protocols that are orders-of-magnitude more efficient.

We first focus on one-round sessions where $\prover$ sends all queries to $\tlsserver$ before receiving any response.
Most applications of \systemname, e.g., proving provenance of content retrieved via HTTP, are one-round. Extending \systemname to multi-round sessions is discussed in~\cref{app:ext}.

\subsubsection{CBC-HMAC}
\label{sec:hmac trick}
Recall that $\prover$ and $\verifier$ hold shares of the MAC key, while $\prover$ holds the encryption key. To construct TLS records encrypting $Q$---potentially private to $\prover$, the two parties first run a 2PC protocol to compute the HMAC tag $\tau$ of $Q$, and then $\prover$ encrypts $Q \| \tau$ locally and sends the ciphertext to $\tlsserver$.

\newcommand{\ipad}{\mathsf{ipad}}
\newcommand{\opad}{\mathsf{opad}}
\newcommand{\ipadhash}{s_0}
\newcommand{\opadhash}{h_o}
\newcommand{\initstate}{\mathsf{IV}}

Let $\hash$ denote SHA-256.
Recall that the HMAC of message $m$ with key $\key$ is
\ifcameraready
$\hmac(\key, m) = \hash((\key \xor \opad) \,\|\, \underbrace{\hash((\key \xor \ipad) \,\|\, m)}_\text{inner hash}).$
\else
\[
\hmac_\hash(\key, m) = \hash((\key \xor \opad) \,\|\, \underbrace{\hash((\key \xor \ipad) \,\|\, m)}_\text{inner hash}).
\]
\fi

A direct 2PC implementation would be expensive for large queries, as it requires hashing the entire query in 2PC to compute the inner hash.
The key idea in our optimization is to make the computation of the inner hash local to $\prover$ (i.e., without 2PC).
If $\prover$ knew $\key\oplus \ipad$, she could compute the inner hash. We cannot, though, simply give $\key \oplus \ipad$ to $\prover$, as she could then learn $\key$ and forge MACs.

Our optimization exploits the Merkle–Damg{\aa}rd structure in SHA-256. Suppose $m_1$ and $m_2$ are two correctly sized blocks. Then $\hash(m_1 \| m_2)$ is computed as $f_\hash(f_\hash(\initstate, m_1), m_2)$ where $f_\hash$ denotes the one-way compression function of $\hash$, and $\initstate$ the initial vector.

After the three-party handshake, $\prover$ and $\verifier$ execute a simple 2PC protocol to compute $\ipadhash = f_\hash(\initstate, \mackey \xor \ipad)$, and reveal it to $\prover$. To compute the inner hash of a message $m$, $\prover$ just uses $s_0$ as the IV to compute a hash of $m$.
Revealing $\ipadhash$ does not reveal $\mackey$, as $f_\hash$ is assumed to be one-way.
To compute $\hmac(\key, m)$ then involves computing the outer hash in 2PC on the inner hash, a much shorter message. Thus, we manage to reduce the amount of 2PC computation to a few blocks regardless of query length, as opposed to up to 256 SHA-2 blocks in each record with generic 2PC.
The protocol is formally specified in~\cref{fig:2pc-hmac}.

\subsubsection{AES-GCM}
\label{sec:posths-gcm}

For GCM, $\prover$ and $\verifier$ perform authenticated encryption of $Q$. 2PC-AES is straightforward with optimized circuits (e.g.,~\cite{aescircuit}), but computing tags for large queries is expensive as it involves evaluating long polynomials in a large field {\em for each record}.
Our optimized protocol makes polynomial evaluation local via precompution. We refer readers to~\cref{app:gcm posths} for details. Since 2PC-GCM involves not only tag creation but also AES encryption, it incurs higher computational cost and latency than CBC-HMAC. 

In~\cref{sec:alternative protocol}, we present a highly efficient alternative protocol that avoids post-handshake 2PC protocols altogether, with additional trust assumptions.\fanz{new pls review!}

\subsection{Full protocol}
\label{sec:full protocol}

\begin{figure}
    \centering
\protocol{$\decoPROT$}{
$\PROT_{\tlsserver}$: follow the standard TLS protocol. \\[1mm]
$\PROT_{\prover}$ and $\PROT_{\verifier}$:
\begin{itemize}[leftmargin=*]
\item $\verifier$ sends $(\sid, \query$, $\PRED$) to $\prover$, where $\query$ is the query template and $\PRED$ the statement to be proven over the response to $\prover$.
\item $\prover$ examines them and chooses whether to proceed. If so, $\prover$ starts the handshake.
\item ({\bf 3P-HS}) $\prover, \verifier$ execute the three-party handshake protocol. $\prover$ gets the encryption key $\enckey$ and a share of the MAC key $\mackeyP$, while $\verifier$ gets the other share $\mackeyV$.
\item ({\bf Query}) $\prover$ computes a query using the template $Q=\query(\PRI)$. $\prover$ invokes 2PC-HMAC with $\verifier$ to compute a tag $\tau$. $\prover$ sends $(\sid, \hat Q =\enc(\enckey, Q\|\tau)$) to $\tlsserver$.
\item ({\bf Commit and verify}) After receiving a response $(\sid, \hat R)$ from $\tlsserver$, $\prover$ sends $(\sid, \hat Q,\hat R, \mackeyP)$ to $\verifier$ as a commitment to the session data. After receiving $(\sid, \mackeyV)$ from $\verifier$, $\prover$ computes $\mackey = \mackeyV + \mackeyP$, decrypts $R\|\tau=\dec(\enckey, \hat R)$, and verifies $\tau$ against $\mackey$.%
\item ({\bf Proof gen})  Let $b=\PRED(R)$, $x=(\enckey, \PRI, Q, R)$ and $w=(\hat Q, \hat R, \mackey, b)$. $\prover$ sends $(\sid, \msgprove,x,w)$ to $\idealZK$ and outputs $(Q,R)$. If $\verifier$ receives $(\sid, \msgproof, 1, (\hat Q, \hat R, \hat \key ^{\sf MAC}, b))$ from $\idealZK$, $\verifier$ checks if $\hat \key ^{\sf MAC} =\mackeyP+\mackeyV$. If so, $\verifier$ outputs $(\sid, b, \tlsserver)$.
\end{itemize}
}
    \caption{The \systemname protocol. We only show the CBC-HMAC variant for clarify, while the GCM variant is described in~\cref{sec:full protocol}.}
    \label{fig:fullprotocolUC}
\end{figure}

After querying the server and receiving a response, $\prover$ commits to the session by sending the ciphertexts to $\verifier$, and receives $\verifier$'s MAC key share. Then $\prover$ can verify the integrity of the response, and prove statements about it.
\Cref{fig:fullprotocolUC} specifies the full \systemname protocol for CBC-HMAC (the protocol for GCM is similar and described later).

For clarity, we abstract away the details of zero-knowledge proofs in an ideal functionality $\idealZK$ like that in~\cite{zkfunc}. On receiving $(\msgprove, x,w)$ from $\prover$, where $x$ and $w$ are private and public witnesses respectively, $\idealZK$ sends $w$ and the relationship $\pi(x,w)\in\bin$ (defined below) to $\verifier$. Specifically, for CBC-HMAC, $x,w,\pi$ are defined as follows:
$x=(\enckey, \PRI, Q, R)$ and $w=(\hat Q, \hat R, \mackey, b)$. The relationship $\pi(x,w)$ outputs $1$ if and only if (1) $\hat Q$ (and $\hat R$) is the CBC-HMAC ciphertext of $Q$ (and $R$) under key $\enckey, \mackey$; (2) $\query(\PRI)=Q$; and (3) $\PRED(R)=b$. Otherwise it outputs $0$.

Assuming functionalities for secure 2PC and ZKPs, it can be shown that $\decoPROT$ UC-securely realizes $\idealOracle$ for malicious adversaries, as stated in Theorem~\ref{thm:mainUC}. We provide a simulation-based proof (sketch) in~\cref{sec:securityproofs}.

\begin{theorem}[Security of $\decoPROT$]\label{thm:mainUC}
Assuming the discrete log problem is hard in the group used in the three-party handshake, and that $f$ (the compression function of SHA-256) is an random oracle,
$\PROT_\text{\systemname}$ UC-securely realizes $\idealOracle$ in the $(\idealTwoPC,\idealZK)$-hybrid world,  against a static malicious adversary with abort.
\end{theorem}

The protocol for GCM has a similar flow. We've specified the GCM variants of the three-party handshake and query construction protocols.
Unlike CBC-HMAC, GCM is not committing~\cite{DBLP:conf/crypto/GrubbsLR17}: for a given ciphertext $C$ encrypted with key $\key$, one knowing $\key$ can efficiently find $\key'\neq \key$ that decrypts $C$ to a different plaintext while passing the integrity check.
To prevent such attacks, we require $\prover$ to commit to her key share $\keyP$ before learning $\verifier$'s key share.
In the proof generation phase, in addition to proving statements about $Q$ and $R$, $\prover$ needs to prove that the session keys used to decrypt $\hat Q$ and $\hat R$ are valid against the commitment to $\keyP$. Proof of the security of the GCM variant is like that for CBC-HMAC.
\section{Proof generation}
\newcommand{\ciphertext}{\VEC{\hat M}}
\newcommand{\plaintext}{\VEC{M}}
\newcommand{\encrecord}{\hat \record} %
\newcommand{\record}{\mathsf{rec}}
\newcommand{\block}{B}

\label{sec:proof gen}

Recall that the prover $\prover$ commits to the ciphertext $\ciphertext$ of a TLS session and proves to $\verifier$ that the plaintext  $\plaintext$ satisfies certain properties.
Without loss of generality, we assume $\ciphertext$ and $\plaintext$ contain only one TLS record, and henceforth call them the {\em ciphertext record} and the {\em plaintext record}. Multi-record sessions can be handled by repeating the protocol for each record.

Proving only the provenance of $\plaintext$ is easy: just reveal the encryption keys. But this sacrifices privacy. Alternatively, $\prover$ could prove any statement about $\plaintext$ using general zero-knowledge techniques. But such proofs are often expensive.

In this section, we present two classes of statements optimized for what are likely to be the most popular applications: revealing only a substring of the response while proving its provenance (\cref{sec:seletive opening}), or further proving that the revealed substring appears in a context expected by $\verifier$ (\cref{subsec:two stage parsing}).

\subsection{Selective opening}
\label{sec:seletive opening}

We introduce {\em selective opening}, techniques that allow $\prover$ to efficiently {\em reveal} or {\em redact} substrings in the plaintext. %
Suppose the plaintext record is composed of chunks $\plaintext=(B_1,\cdots,B_n)$ (details of chunking are discussed shortly). Selective opening allows $\prover$ to prove that the $i$th chunk of  $\plaintext$ is $B_i$, without revealing the rest of $\plaintext$; we refer to this as $\REV$ mode. It can also prove that $\plaintext_{-i}$ is the same as $\plaintext$ but with the chunks removed. We call this $\RED$ mode.
Both modes are simple, but useful for practical privacy goals.
The granularity of selective opening depends on the cipher suite, which we now discuss.

\newcommand{\macblocks}{\hat{\mathbf{\sigma}}}

\subsubsection{CBC-HMAC}
Recall that for proof generation, $\prover$ holds both the encryption and MAC keys $\enckey$ and $\mackey$, while $\verifier$ only has the MAC key $\mackey$.
Our performance analysis assumes a ciphersuite with SHA-256 and AES-128, which matches our implementation, but the techniques are applicable to other parameters.
Recall that MAC-then-encrypt is used: a plaintext record $\plaintext$ contains up to 1024 AES blocks of data and 3 blocks of MAC tag $\sigma$, which we denote as  $\plaintext=(B_1,\dots,B_{1024}, \sigma)$ where $\sigma=(B_{1025}, B_{1026}, B_{1027})$. $\ciphertext$ is a CBC encryption of $\plaintext$, consisting of the same number of blocks: $\ciphertext = (\hat{B}_1,\dots,\hat{B}_{1024}, \macblocks)$ where $\macblocks = (\hat{B}_{1025}, \hat{B}_{1026}, \hat{B}_{1027})$.

\boldhead{Revealing a TLS record}
A na\"ive way to prove that $\ciphertext$ encrypts $\plaintext$ without revealing $\enckey$ is to prove correct encryption of each AES block in ZKP. However, this would require up to 1027 invocations of AES in ZKP, resulting in impractical performance.

Leveraging the MAC-then-encrypt structure, the same can be done using only 3 invocations of AES in ZKP. The idea is to prove that the last few blocks of $\ciphertext$ encrypt a tag $\sigma$ and reveal the plaintext directly. Specifically, $\prover$ computes $\pi_\sigma = \ZKP{\enckey: \macblocks =\cbc(\enckey, \sigma)}$ and sends $(\plaintext, \pi_\sigma)$ to $\verifier$. Then $\verifier$ verifies $\pi$ and checks the MAC tag over $\plaintext$ (note that $\verifier$ knows the MAC key.)
Its security relies on the collision-resistance of the underlying hash function in HMAC, i.e., $\prover$ cannot find $\plaintext'\neq\plaintext$ with the same tag $\sigma$.

\boldhead{Revealing a record with redacted blocks}
Suppose the $i$th block contains sensitive information that $\prover$ wants to redact.
A direct strategy is to prove that $\VEC{B}_{i-}=(B_1,\cdots,B_{i-1})$ and $\VEC{B}_{i+} = (B_{i+1},\cdots, B_n)$ form the prefix and suffix of the plaintext encrypted by $\ciphertext$, by computing $\pi_\sigma$ (see above) and $\ZKP{B_i: \sigma= \hmac(\mackey, \VEC{B}_{i-} \| B_i \| \VEC{B}_{i+})}.$ This is expensive though as it would involve  $3$ AES and $256$ SHA-256 compression in ZKP.

Leveraging the Merkle-Damgård structure of SHA-256 (c.f.~\cref{sec:hmac trick}), several optimization is possible.
Let $f$ denote the compression function of SHA-256,
and $s_{i-1}$ the state after applying $f$ on $\VEC{B}_{i-}$.
First, if both $s_{i-1}$ and $s_i$ can be revealed, e.g., when $B_i$ contains high-entropy data such as API keys, the above goal can be achieved using just 1 SHA-256 in ZKP. 
To do so, $\prover$ computes $\pi = \ZKP{ B_i: f(s_{i-1}, B_i)=s_i}$ and sends $(\pi, s_{i-1}, s_i, \VEC{B}_{i-}, \VEC{B}_{i+})$ to $\verifier$, who then 1) checks $s_{i-1}$ by recomputing it from $\VEC{B}_{i-}$; 2) verifies $\pi$; and 3) checks the MAC tag $\sigma$ by recomputing it from $s_i$ and $\VEC{B}_{i+}$.
Assuming $B_i$ is high entropy, revealing $s_{i-1}$ and $s_i$ doesn't leak $B_i$ since $f$ is one-way.

On the other hand, if both $s_{i-1}$ and $s_i$ cannot be revealed to $\verifier$ (e.g., when brute-force attacks against $B_i$ is feasible), we can still reduce the cost by having $\prover$ redact a prefix (or suffix) of the record containing the block $B_i$. The cost incurred then is $256-i$ SHA-2 hashes in ZKP. We relegate the details to~\cref{sec: cbc hmac tricks}.
Generally ZKP cost is proportional to record sizes so TLS fragmentation can also lower the cost by a constant factor.

\vspace{-2mm}
\subsubsection{GCM}
Unlike CBC-HMAC, revealing a block is very efficient in GCM.
First, $\prover$ reveals $\aes(\key, IV)$ and $\aes(\key, 0)$, with proofs of correctness in ZK, to allow $\verifier$ to verify the integrity of the ciphertext. Then, to reveal the $i$th block, $\prover$ just reveals the encryption of the $i$th counter $C_i = \aes(\key, \INC^i(IV))$ with a correctness proof. $\verifier$ can decrypt the $i$th block as $\hat B_i \oplus C_i$. $IV$ is the public initial vector for the session, and $\INC^i(IV)$ denotes incrementing $IV$ for $i$ times (the exact format of $\INC$ is immaterial.)
To reveal a TLS record, $\prover$ repeat the above protocol for each block. We defer details to~\cref{sec:proof gen gcm}.

\ifcameraready
\else
In summary, CBC-HMAC allows efficient selective revealing at the TLS record-level and redaction at block level in \systemname, while GCM allows efficient revealing at block level.
Selective opening can also serve as pre-processing to reduce the input length for a subsequent zero-knowledge proof, which we will illustrate in~\cref{sec:applications} with concrete applications.
\fi
\subsection{Context integrity by two-stage parsing}
\label{subsec:two stage parsing}

For many applications, the verifier $\verifier$ may need to verify that the revealed substring appears in the right context. We refer to this property as {\em context integrity}.
In this section we present techniques for $\verifier$ to specify contexts and for $\prover$ to prove context integrity efficiently.

For ease of exposition, our description below focuses on the revealing mode, i.e., $\prover$ reveals a substring of the server's response to $\verifier$. We discuss how redaction works in~\cref{sec:keyvalue}.

\subsubsection{Specification of contexts}
\newcommand{\partgrammar}{{\grammar'}}
\newcommand{\trans}{{\sf Trans}}
\newcommand{\opening}{{R_{\text{open}}}}
\newcommand{\parser}[1]{\mathsf{Parser}_{#1}}
\newcommand{\AUX}{{\mathsf{aux}}}
\newcommand{\CTX}{\mathcal{C}}
\newcommand{\ctx}{{\sf CTX}}

Our techniques for specifying contexts assume that the TLS-protected data sent to and from a given server $\tlsserver$ has a well-defined context-free grammar $\grammar$, known to both $\prover$ and $\verifier$.
In a slight abuse of notation, we let $\grammar$ denote both a grammar and the language it specifies. Thus, $R\in\grammar$ denotes a string $R$ in the language given by $\grammar$.
We assume that $\grammar$ is {\em unambiguous}, i.e., every $R\in \grammar$ has a unique associated parse-tree $T_R$.
JSON and HTML are examples of two widely used languages that satisfy these requirements, and are our focus here.

When $\prover$ then presents a substring $\opening$ of some response $R$ from $\tlsserver$, we say that $\opening$ has {\em context integrity} if $\opening$ is produced in a certain way expected by $\verifier$.
Specifically, $\verifier$ specifies a set $S$ of positions in which she might expect to see a valid substring $\opening$ in $R$.
In our definition,
$S$ is a set of paths
from the root in a parse-tree defined by $\grammar$ to internal nodes. Thus $s \in S$, which we call a {\em permissible path}, is a sequence of non-terminals.
Let $\rho_R$ denote the root of $T_R$ (the parse-tree of $R$ in $\grammar$).
We say that a string $\opening$ has context-integrity with respect to $(R, S)$ if $T_R$ has a subtree whose leaves \emph{yield} (i.e. concatenate to form) the string $\opening$, and that there is a path $s \in S$ from $\rho_R$ to the root of the said subtree. %

Formally, we define context integrity in terms of a predicate ${\sf \ctx_{\grammar}}$ in~\cref{def:context integ}. At a high level, our definition is reminiscent of the production-induced context in~\cite{DBLP:conf/ccs/SaxenaML11}.

\begin{definition}
Given a grammar $\grammar$ on TLS responses, $R\in\grammar$, a substring $\opening$ of $R$, a set $S$ of permissible paths, we define a {\em context function} $\ctx_\grammar$ as a boolean function such that $\ctx_\grammar: (S, R, \opening) \mapsto \true$ iff $\exists$ a sub-tree $T_{\opening}$ of $T_R$ with a path $s \in S$ from $\rho_{T_R}$ to $\rho_{T_{\opening}}$ and $T_\opening$ yields $\opening$.
$\opening$ is said to have {\em context integrity} with respect to $(R, S)$ if $\ctx_\grammar(S, R, \opening) = \true$.
\label{def:context integ}
\end{definition}

As an example, consider the JSON string $J$ in~\cref{fig:example}. JSON contains (roughly) the following rules:

{\small\ttfamily
\begin{tabular}{ll}
\textbf{Start} $\to$ object & \textbf{object} $\to$~\{~pairs~\} \\
\textbf{pair} $\to$ ~``key''~:~value & \textbf{pairs} $\to$ pair | pair, pairs\\
\textbf{key} $\to$ chars &  \textbf{value} $\to$ chars | object
\end{tabular}
}

In that example, $\verifier$ was interested in learning the derivation of the {\tt pair} $p_{\text{\tt balance}}$ with {\tt key ``balance''} in the {\tt object} given by the {\tt value} of the {\tt pair} $p_{\text{\tt checking}}$ with {\tt key ``checking a/c''}. Each of these non-terminals is the label for a node in the parse-tree $T_J$. The path from the root {\tt Start} of $T_J$ to $p_{\text{\tt checking}}$ requires traversing a sequence of nodes of the form {\tt Start $\to$ object $\to$ pairs$*$ $\to$ $p_{\text{\tt checking}}$}, where {\tt pairs$*$} denotes a sequence of zero or more {\tt pairs}. So $S$ is the set of such sequences and $\opening$ is the string {\tt ``checking a/c'': \{``balance'': \$2000\}}.  %

\subsubsection{Two-stage parsing}

Generally, proving $\opening$ has context integrity, i.e., $\ctx_\grammar(S, R, \opening)=\true$, without directly revealing $R$ would be expensive, since computing $\ctx_\grammar$ may require computing $T_R$ for a potentially long string $R$.
However, we observed that under certain assumptions that TLS-protected data generally satisfies, much of the overhead can be removed by having $\prover$  {\em pre-process} $R$ by applying a transformation $\trans$ agreed upon by $\prover$ and $\verifier$, and prove that $\opening$ has context integrity with respect to $R'$ (a usually much shorter string) and $S'$ (a set of permissible paths specified by $\verifier$ based on $S$ and $\trans$).

Based on this observation, we introduce a {\em two-stage parsing scheme} for efficiently computing $\opening$ and proving $\ctx_\grammar(S, R, \opening)=\true$.
Suppose $\prover$ and $\verifier$ agree upon $\grammar$, the grammar used by the web server, and a transformation $\trans$. Let $\partgrammar$ be the grammar of strings $\trans(R)$ for all $R\in \grammar$.
Based on $\trans$, $\verifier$ specifies permissible paths $S'$ and a constraint-checking function $\cons_{\grammar,\partgrammar}$.
In the first stage, $\prover$: (1) computes a substring $\opening$ of $R$ by parsing $R$ (such that $\ctx_\grammar(S, R, \opening)=\true$) (2) computes another string $R' = \trans(R)$.
In the second stage, $\prover$ proves to $\verifier$ in zero-knowledge that (1) $\cons_{\grammar, \partgrammar}(R, R')=\true$ and (2)  $\ctx_{\partgrammar}(S', R', \opening) = \true$.
Note that in addition to public parameters $\grammar, \partgrammar, S, S',\trans, \cons_{\grammar, \partgrammar}$, the verifier only sees a commitment to $R$, and finally, $\opening$.

This protocol makes the zero-knowledge computation significantly less expensive by deferring actual parsing to a non-verifiable computation. In other words, the computation of $\ctx_{\partgrammar}(S', R', \opening)$ and $\cons_{\grammar, \partgrammar}(R, R')$ can be much more efficient than that of $\ctx_\grammar(S,\allowbreak R, \opening)$.

We formalize the correctness condition for the two-stage parsing in an operational semantics rule in~\cref{def:correctness}.
Here, $\langle f, \sigma \rangle$ denotes applying a function $f$ on input $\sigma$, while $\frac{P}{C}$ denotes that if the premise $P$ is true, then the conclusion $C$ is true.

\begin{definition}
\label{def:correctness}
Given a grammar $\grammar$, a context function and permissible paths $\ctx_\grammar(S,\, \cdot\, , \, \cdot\, )$, a transformation $\trans$, a grammar $\partgrammar=\{R': R'=\trans(R), R\in\grammar\}$ with context function and permissible paths $\ctx_\partgrammar(S',\, \cdot\, ,\, \cdot\, )$ and a function $\cons_{\grammar,\partgrammar}$, we say $(\cons_{\grammar, \partgrammar}, S')$ are correct w.r.t. $S$, if for all $(R,R',\opening)$ such that $R \in \grammar$, booleans ${\sf b}$
the following rule holds:
\[
\inference{\langle \cons_{\grammar,\partgrammar}, (R,R') \rangle \Rightarrow \true \text{\,\,\,}
\langle \ctx_\partgrammar, (S', R', \opening) \rangle \Rightarrow {\sf b}}{\langle \ctx_{\grammar}, (S, R, \opening) \rangle \Rightarrow {\sf b}}.
\]
\end{definition}

Below, we focus on a grammar that most \systemname applications use, and present concrete constructions of two-stage parsing schemes.

\subsubsection{\systemname focus: Key-value grammars}\label{sec:keyvalue}

A broad class of data formats, such as JSON, have a notion of key-value pairs. Thus, they are our focus in the current version of \systemname.

A key-value grammar $\grammar$ produces key-value pairs according to the rule, ``{\tt pair $\to$ \start key \midd value \End}'', where \start, \midd and \End are delimitors.
For such grammars, an array of optimizations can greatly reduce the complexity for proving  context. We discuss a few such optimizations below, with formal specification relegated to~\cref{app:keyvalue}.

\boldhead{Revelation for a globally unique key}
For a key-value grammar $\grammar$, set of paths $S$, if for an $R\in \grammar$, a substring $\opening$ satisfying context-integrity requires that $\opening$ is parsed as a key-value pair with a globally unique key {\tt K} (formally defined in~\cref{sec:uniqueKeyDef}), $\opening$ simply needs to be a substring of $R$ and correctly be parsed as a {\tt pair}.
Specifically, $\trans(R)$  outputs a substring $R'$ of $R$ containing the desired key, i.e., a substring of the form ``{\tt \start K \midd value \End}'' and $\prover$ can output $\opening=R'$. $\partgrammar$ can be defined by the rule {\tt S$_\partgrammar$ $\to$ pair} where {\tt S$_\partgrammar$} is the start symbol in the production rules for $\partgrammar$. Then (1) $\cons_{\grammar, \partgrammar}(R, R')$ checks that $R'$ is a substring of $R$ and (2) for $S'=$\{{\tt S$_{\partgrammar}$}\}, $\ctx_\partgrammar(S', R', \opening)$ checks that (a) $R'\in \partgrammar$ and (b) $\opening=R'$.
Globally unique keys arise in~\cref{sec:ageProofApp} when selectively opening the response for {\tt age}. %

\boldhead{Redaction in key-value grammars}
Thus far, our description of two-stage parsing assumes the $\REV$ mode in which $\prover$ reveals a substring $\opening$ of $R$ to $\verifier$ and proves that $\opening$ has context integrity with respect to the set of permissible paths specified by $\verifier$. In the $\RED$ mode, the process is similar, but instead of revealing $\opening$ in the clear, $\prover$ generates a commitment to $\opening$ using techniques from~\cref{sec:seletive opening} and reveals $R$, with $\opening$ removed, for e.g. by replacing its position with a dummy character.

\newcommand{\redactedstring}{\ensuremath{R_\text{redact}}}

\section{Applications}
\label{sec:applications}

\systemname can be used for any oracle-based application. To showcase its versatility, we have implemented and evaluated three applications that leverage its various capabilities: 1) a confidential financial instrument realized by smart contracts; 2) converting legacy credentials to anonymous credentials; and 3) privacy-preserving price discrimination reporting.
Due to lack of space, we only present concrete implementation details for the first application, and refer readers to~\cref{app: app details} for others. Evaluation results are presented  in~\cref{sec:app eval}.
\subsection{Confidential financial instruments}
\label{sec:financialIntrumentsApp}

\newcommand{\contract}{\mathcal{SC}}
\newcommand{\id}{\mathsf{ID}}
\newcommand{\oracle}{\ensuremath{\mathcal{O}}\xspace}
\newcommand{\assetname}{\mathsf{N}}
\newcommand{\assetprice}{\mathsf{P}}
\newcommand{\assetday}{\mathsf{D}}
\newcommand{\commit}{\mathsf{com}}
\newcommand{\comm}{\mathcal{C}}
\newcommand{\randomness}{r}

Financial derivatives are among the most commonly cited smart contract applications~\cite{cftc_2018,openlawofficial_2018}, and exemplify the need for authenticated data feeds (e.g., stock prices).
For example, one popular financial instrument that is easy to implement in a smart contract is a {\em binary option}~\cite{binaryoption}. This is a contract between two parties betting on whether, at a designated future time, e.g., the close of day $\assetday$, the price $\assetprice^*$ of some asset $\assetname$ will equal or exceed a predetermined target price $\assetprice$, i.e., $\assetprice^* \geq \assetprice$. A smart contract implementing this binary option can call an oracle $\oracle$ to determine the outcome.

In principle, \oracle can conceal the underlying asset $\assetname$ and target price $\assetprice$ for a binary option on chain. It simply accepts the option details off chain, and reports only a bit specifying the outcome $\PRED := \assetprice^* \geq?~\assetprice$. This approach is introduced in~\cite{Mixicles:2019}, where it is referred to as a {\em Mixicle}.

A limitation of a basic Mixicle construction is that \oracle itself learns the details of the financial instrument. Prior to \systemname, only oracle services that use TEE (e.g.,~\cite{zhang2016town}) could conceal queries from \oracle.
We now show how \systemname can support execution of the binary option {\em without \oracle learning the details of the financial instrument, i.e., $\assetname$ or $\assetprice$}\footnote{The predicate direction $\geq?$ or $\leq?$ can be randomized. Concealing winner and loser identities and payment amounts is discussed in~\cite{Mixicles:2019}. Additional steps can be taken to conceal other metadata, e.g., the exact settlement time.}.

The idea is that the option winner plays the role of $\prover$, and obtains a signed result of $\PRED$ from $\oracle$, which plays the role of $\verifier$. We now describe the protocol and its implementation.

\boldhead{Protocol}
Let $\{\sk_\oracle, \pk_\oracle\}$ denote the oracles' key pair.
In our scheme, a binary option is specified by an asset name $\assetname$, threshold price $\assetprice$, and settlement date $\assetday$.
We denote the commitment of a message $M$ by $\comm_M = \commit(M, r_M)$ with a witness $r_M$.
\Cref{fig:confidentialbinaryoption} shows the workflow steps in a confidential binary option:

\textit{1) Setup:} Alice and Bob agree on the binary option $\{\assetname, \assetprice, \assetday\}$ and create a smart contract $\contract$ with identifier $\id_\contract$, The contract contains $\pk_\oracle$, addresses of the parties, and commitments to the option $\{\comm_\assetname, \comm_\assetprice, \comm_\assetday\}$ with  witnesses known to both parties.
They also agree on public parameters $\PUB$ (e.g., the URL to retrieve asset prices).

\textit{2) Settlement:} Suppose Alice wins the bet. To claim the payout, she uses \systemname to generate a ZK proof that the current asset price retrieved matches her position. Alice and $\oracle$ execute the \systemname protocol (with $\oracle$ acting as the verifier) to retrieve the asset price from $\PUB$ (the target URL). We assume the response contains $(\assetname^*,\assetprice^*, \assetday^*$).  In addition to the ZK proof in \systemname to prove origin $\PUB$, Alice proves
\ifcameraready
knowledge of $(\assetprice, \assetname^*, \assetprice^*, \assetday^*, r_\assetname, r_\assetprice, r_\assetday)$ such that $(\assetprice \leq \assetprice^*)\,\land\, \comm_\assetname = \commit(\assetname^*, r_\assetname) \,\land\, \comm_\assetprice = \commit(\assetprice, r_\assetprice) \,\land\, \comm_\assetday = \commit(\assetday^*, r_\assetday)$.
\else
the following statement:
\begin{multline*}
\ZKP{\assetprice, \assetname^*, \assetprice^*, \assetday^*, r_\assetname, r_\assetprice, r_\assetday: (\assetprice \leq \assetprice^*)\,\land\, \\ \comm_\assetname = \commit(\assetname^*, r_\assetname) \,\land\, \comm_\assetprice = \commit(\assetprice, r_\assetprice) \,\land\, \comm_\assetday = \commit(\assetday^*, r_\assetday)}.
\end{multline*}
\fi

Upon successful proof verification, the oracle returns a signed statement with the contract ID, $S=\sig(\sk_\oracle, \id_\contract)$.

\textit{3) Payout:} Alice provides the signed statement $S$ to the contract, which verifies the signature and pays the winning party.

Alice and Bob need to trust $\oracle$ for integrity, but not for privacy. They can further hedge against integrity failure by using multiple oracles, as explained in~\cref{para:adf}. Decentralizing trust over oracles is a standard and already deployed technique~\cite{SCCL:2017}. We emphasize that \systemname ensures privacy even if all the oracles are malicious.%

\begin{figure}
\centering
\resizebox{\columnwidth}{!}{
\begin{tikzpicture}
\tikzstyle{line} = [thick];
\tikzstyle{comment} = [text centered,text width=2.8cm, font=\footnotesize];
\tikzstyle{protocol} = [inner sep=2mm, draw=black!50, fill=black!10,font=\footnotesize];
\tikzstyle{party} = [rounded corners, inner sep=2mm, minimum height=.8cm,rectangle];

\node[party] (alice) at (0, 0) {Alice};
\node[party] (bob) at (6, 0) {Bob};
\node[party] (oracle) at (2, -2) {Oracle $\oracle$};
\node[party] (contract) at (4, -2) {Contract $\contract$};
\draw[dashed] (alice.south) -- (0, -6);
\draw[dashed] (bob.south) -- (6, -2);
\draw[dashed] (oracle.south) -- (2, -4.2);
\draw[dashed] (contract.south) -- (4, -6);
\scoped[yshift=-1cm]{
\draw[<->,line] (0, 0) -- node[above,comment]{1. Set up contract $\contract$, shared randomness $r_\assetname, r_\assetprice, r_\assetday$} (6, 0);
}%

\node[comment,protocol] (tph) at (4.4, -2.7) {$\id_\contract, \pk_\oracle, \{\comm_\assetname, \comm_\assetprice, \comm_\assetday\}$};

\scoped[yshift=-4cm]{
\draw[->,line] (0, 0) -- node[above,comment]{2. ZKP using \systemname} (2, 0);
}

\scoped[yshift=-4.2cm]{
\draw[<-,line] (0, 0) -- node[below,comment]{$S = \sig(\sk_\oracle, \id_\contract)$} (2, 0);
}

\scoped[yshift=-5.5cm]{
\draw[->,line] (0, 0) -- node[above,comment]{3. Send $S$} (4, 0);
}

\scoped[yshift=-5.7cm]{
\draw[<-,line] (0, 0) -- node[below,comment]{Receive payout} (4, 0);
}

\node[inner sep=2pt, draw, thick] (deco) at (9.3, -3.5)
    {
    \begin{minipage}{.65\columnwidth}
\begin{lstlisting}[frame=none]
> GET /query?function=GLOBAL_QUOTE&symbol=|\colorbox{red!20}{GOOGL}|
  Host: www.alphavantage.co
>
{
    "Global Quote": {
        |\colorbox{red!20}{"01. symbol": "GOOGL"}|,
        |\colorbox{red!20}{"05. price": "1157.7500"}|,
        |\colorbox{red!20}{"07. day": "2019-07-16"}|
    }
}
\end{lstlisting}
    \end{minipage}
   };

\draw[->, >=latex, blue!10!white, line width=15pt] (2.5, -4) -- (7, -4);
\end{tikzpicture}
}
\caption{Two parties Alice and Bob execute a confidential binary option.
Alice uses \systemname to access a stock price API and convince $\oracle$ she has won. Examples of request and response are shown to the right. Text in red is sensitive information to be redacted.
}
\label{fig:confidentialbinaryoption}
\label{fig:stockprice}
\end{figure}

\boldhead{Implementation details} \Cref{fig:stockprice} shows the request and response of a stock price API.
Let $\hat R$ and $R$ denote the response ciphertext  and the plaintext respectively. To settle an option, $\prover$ proves to $\verifier$ that $R$ contains evidence that he won the option, using the two-stage parsing scheme introduced in~\cref{subsec:two stage parsing}.
In the first stage, $\prover$ parses $R$ locally and identifies the smallest substring of $R$ that can convince $\verifier$. E.g., for stock prices, $R_\text{price}=\texttt{"05. price": "1157.7500"}$ suffices. In the second stage, $\prover$ proves knowledge of $(R_{\text{price}}, \assetprice, r_{\assetprice})$ in ZK such that 1) $R_\text{price}$ is a substring of the decryption of $\hat R$; 2) $R_\text{price}$ starts with \texttt{"05. price"}; 3) the subsequent characters form a floating point number $\assetprice^*$ and that $\assetprice^* \geq~\assetprice$; 4) $\commit(\assetprice, r_{\assetprice})=\comm_{\assetprice}$.

This two-stage parsing is secure assuming the keys are unique and the key \texttt{"05. price"} is followed by the price, making the grammar of this response a \emph{key-value grammar with unique keys}, as discussed in \cref{subsec:two stage parsing}. Similarly, $\prover$ proves that the stock name and date in $R$ match the commitments. With the CBC-HMAC ciphersuite, the zero-knowledge proof circuit involves redacting an entire record (408 bytes), computing commitments, and string processing.

\vspacecamerareadyonly{-2mm}
\subsection{Legacy credentials to anonymous credentials: Age proof}\label{sec:ageProofApp}

User credentials are often inaccessible outside a service provider's environment. Some providers offer third-party API access via OAuth tokens, but such tokens reveal user identifiers. \systemname allows users holding credentials in existing systems (what we call {\em legacy credentials}) to prove statements about them to third parties (verifiers) {\em anonymously}. Thus, \systemname is the first system that allows users to convert {\em any} web-based legacy credential into an anonymous credential without server-side support~\cite{ritzdorf2017tls} or trusted hardware~\cite{zhang2016town}.

We showcase an example where a student proves her/his age is over 18 using credentials (demographic details) stored on a University website. A student can provide this proof of age to any third party, such as a state issuing a driver's license or a hospital seeking consent for a medical test. We implement this example using the AES-GCM cipher suite and two-stage parsing (See~\cref{fig:age}) with optimizations based on unique keys as in~\cref{subsec:two stage parsing}.

\subsection{Price discrimination}
\label{sec:priceDisc}

Price discrimination refers to selling the same product or service at different prices to different buyers. Ubiquitous consumer tracking enables online shopping and booking websites to employ sophisticated price discrimination~\cite{useem_2017}, e.g., adjusting prices based on customer zip codes~\cite{howe_2017}. Price discrimination can lead to economic efficiency~\cite{odlyzko2003privacy}, and is thus widely permissible under existing laws.

In the U.S., however, the FTC forbids price discrimination if it results in competitive injury~\cite{ftc2017}, while new privacy-focused laws in Europe, such as the GDPR, are bringing renewed focus to the legality of the practice~\cite{borgesius2017online}. Consumers in any case generally dislike being subjected to price discrimination. Currently, however, there is no trustworthy way for users to report online price discrimination.

\systemname allows a buyer to make a verifiable claim about perceived price discrimination by proving the advertised price of a good is higher than a threshold, while hiding sensitive information such as name and address. We implement this example using the AES-GCM cipher suite for the TLS session and reveal 24 AES blocks containing necessary order details and the request URL (See~\cref{fig:shoppingorder}).

\section{Implementation and Evaluation}
\label{sec:eval}

In this section, we discuss implementation details and evaluation results for \systemname and our three applications.

\subsection{\systemname protocols}

We implemented the three-party handshake protocol (3P-HS) for TLS 1.2 and \posthandshake protocols (2PC-HMAC and 2PC-GCM) in about 4700 lines of C++ code.
We built a hand-optimized TLS-PRF circuit with total AND complexity of 779,213.
We also used variants of the AES circuit from~\cite{aescircuit}.
Our implementation uses Relic~\cite{relic} for the Paillier cryptosystem and the EMP toolkit~\cite{emp-toolkit} for the maliciously secure 2PC protocol of~\cite{DBLP:conf/ccs/WangRK17}.

We integrated the three-party handshake and 2PC-HMAC protocols with mbedTLS~\cite{mbedTLS}, a popular  TLS implementation, to build an end-to-end system.
2PC-GCM can be integrated to TLS similarly with more engineering effort. We evaluated the performance of 2PC-GCM separately. The performance impact of integration should be negligible. We did not implement 3P-HS for TLS 1.3, but we conjecture the performance should be comparable to that for TLS 1.2, since the circuit complexity is similar (c.f.~\cref{para:tls1.3}).

\begin{table}
\centering
\caption{Run time (in ms) of 3P-HS and \posthandshake protocols.}
\label{tab:2pc_costs}
\resizebox{\columnwidth}{!}{
\begin{tabular}{llll|ll}
\toprule
 & & \multicolumn{2}{c}{LAN} & \multicolumn{2}{c}{WAN}\\
 & & Online & Offline & Online & Offline \\
\midrule
3P-Handshake   & TLS 1.2 only & 368.5 (0.6) & 1668 (4) & 2850 (20) & 10290 (10) \\
2PC-HMAC       & TLS 1.2 only & 133.8 (0.5) & 164.9 (0.4) & 2520 (20) & 3191 (8) \\
2PC-GCM (256B) & 1.2 and 1.3  & 36.65 (0.02) & 392 (8) & 1208.5 (0.2) & 12010 (70) \\
2PC-GCM (512B) & 1.2 and 1.3  & 53.0 (0.5) & 610 (10) & 2345 (1) & 12520 (70) \\
2PC-GCM (1KB)  & 1.2 and 1.3  & 101.9 (0.5) & 830 (20) & 4567 (4) & 14300 (200) \\
2PC-GCM (2KB)  & 1.2 and 1.3  & 204.7 (0.9) & 1480 (30) & 9093.5 (0.9) & 18500 (200) \\
\bottomrule
\end{tabular}
}
\end{table}

\boldhead{Evaluation} We evaluated the performance of \systemname in both the LAN and WAN settings. Both the prover and verifier run on a \texttt{c5.2xlarge} AWS node with 8 vCPU cores and 16GB of RAM. We located the two nodes in the same region (but different availability zones) for the LAN setting, but in two distinct data centers (in Ohio and Oregon) in the WAN setting. The round-trip time between two nodes in the LAN and WAN is about 1ms and 67ms, respectively, and the bandwidth is about 1Gbps.

\Cref{tab:2pc_costs} summarizes the runtime of \systemname protocols during a TLS session.
50 samples were used to compute the mean and standard error of the mean (in parenthesis).
The MPC protocol we used relies on offline preprocessing to improve performance.
Since the offline phase is input- and target-independent, it can be done prior to the TLS session. Only the online phase is on the critical path.

As shown in~\cref{tab:2pc_costs}, \systemname protocols are very efficient in the LAN setting. It takes 0.37 seconds to finish the three-party handshake. For \posthandshake, 2PC-HMAC is efficient (0.13s per record) as it only involves one SHA-2 evaluation in 2PC, regardless of record size.
2PC-GCM is generally more expensive and the cost depends on the query length, as it involves 2PC-AES over the entire query.
We evaluated its performance with queries ranging from 256B to 2KB, the typical sizes seen in HTTP GET requests~\cite{spdy-whitepaper}.
In the LAN setting, the performance is efficient and comparable to 2PC-HMAC.

In the WAN setting, the runtime is dominated by the network latency because MPC involves many rounds of communication. Nonetheless, the performance is still acceptable, given that \systemname is likely to see only periodic use for most applications we consider.

\vspacecamerareadyonly{-2mm}
\subsection{Proof generation}
\label{sec:app eval}

We instantiated zero-knowledge proofs with a standard proof system~\cite{ben2014succinct} in libsnark~\cite{libsnark}.
We have devised efficiently provable statement templates, but users of \systemname need to adapt them to their specific applications.
SNARK compilers enable such adaptation in a high-level language, concealing low-level details from developers. We used xjsnark~\cite{kosba2018xjsnark} and its Java-like high-level language to build statement templates and libsnark compatible circuits.

Our rationale in choosing libsnark is its relatively mature tooling support.
The proofs generated by libsnark are constant-size and very efficient to verify, the downside being the per-circuit trusted setup.
With more effort, \systemname can be adapted to use, e.g., Bulletproofs~\cite{DBLP:conf/sp/BunzBBPWM18}, which requires no trusted setup but has large proofs and verification time.

\newcommand{\evaluate}[1]{\FPeval{\result}{round(#1,2)}\result}
\begin{table}[t]
\footnotesize
\caption{Costs of generating and verifying ZKPs in proof-generation phase of \systemname for applications in~\cref{sec:applications}.}
\label{tab:evaluation}
\centering
\begin{tabular}{rrrr}
\toprule
 & Binary Option & Age Proof & Price Discrimination \\
\midrule
prover time & 12.97 $\pm$ 0.04s  & 3.67 $\pm$ 0.02s & 12.68 $\pm$ 0.02s\\
verifier time & \evaluate{0.0035 + 0.0035 + 0.0038}s & \evaluate{0.0033 + 0.0039}s & \evaluate{0.0209 * 2 + 0.0039 * 2}s\\
proof size & 861B & 574B & 1722B \\
\#~constraints & 617k & 164k & 535k \\
memory & 1.78GB & 0.69GB & 0.92GB \\
\bottomrule
\end{tabular}
\end{table}

\boldhead{Evaluation} We measure five performance metrics for each example---prover time (the time to generate the proofs), verifier time (the time to verify proofs), proof size, number of arithmetic constraints in the circuit, and the peak memory usage during proof generation.

\Cref{tab:evaluation} summarizes the results.
50 samples were used to compute the mean and its standard error.
Through the use of efficient statement templates and two-stage parsing, \systemname achieves very practical prover performance.
Since libsnark optimizes for low verification overhead, the verifier time is negligible. The number of constraints (and prover time) is highest for the binary option application due to the extra string parsing routines. We use multiple proofs in each application to reduce peak memory usage. For the most complex application, the memory usage is 1.78GB. As libsnark proofs are of a constant size 287B, the proof sizes shown are multiples of that. %

\vspacecamerareadyonly{-2mm}
\subsection{End-to-end performance}

\systemname end-to-end performance depends on the available TLS ciphersuites, the size of private data, and the complexity of application-specific proofs.
Here we present the end-to-end performance of the most complex application of the three we implemented---the binary option.
It takes about~\evaluate{12.97 + 0.37 + 0.13 + 0.30}s to finish the protocol, which includes the time taken to generate unforgeable commitments (\evaluate{0.37 + 0.13}s), to run the first stage of two-stage parsing (0.30s), and to generate zero-knowledge proofs (12.97s). These numbers are computed in the LAN setting; in the WAN setting, MPC protocols are more time-consuming (\evaluate{2.85 + 2.52}s), pushing the end-to-end time up to~\evaluate{12.97 + 2.85 + 2.52 + 0.3}s.

In comparison, Town Crier uses TEEs to execute a similar application in about 0.6s~\cite[Table~I]{zhang2016town}, i.e., around 20x faster than \systemname, \emph{but with added trust assumptions}. Since \systemname is likely to be used only periodically for most applications, its overhead in achieving cryptographic-strength security assurances seems reasonable.
\section{Legal and Compliance Issues}
\label{sec:legal}

Although users can already retrieve their data from websites,
\systemname allows users to export the data {\em with integrity proofs} without their explicit approval or even awareness.
We now briefly discuss the resulting legal and compliance considerations.

Critically, however, \emph{\systemname users cannot unilaterally export data} to a third party with integrity assurance, but rely on oracles as verifiers for this purpose. While \systemname keeps user data private, oracles learn what websites and types of data a user accesses. Thus oracles can enforce appropriate data use, e.g., denying transactions that may result in copyright infringement.

Both users and oracles bear legal responsibility for the data they access. Recent case law on the Computer Fraud and Abuse Act (CFAA), however, shows a shift away from criminalization of web scraping~\cite{sellars2018twenty}, and federal courts have ruled that violating websites’ terms of service is not a criminal act \emph{per se}~\cite{EFF:2010,Khoury:2018}. Users and oracles that violate website terms of service, e.g., ``click wrap'' terms, instead risk \emph{civil} penalties~\cite{ABA:2019}. \systemname compliance with a given site’s terms of service is a site- and application-specific question.

Oracles have an incentive to establish themselves as trustworthy within smart-contract and other ecosystems. We expect that reputable oracles will provide users with menus of the particular attestations they issue and the target websites they permit, vetting these options to maximize security and minimize liability and perhaps informing or cooperating with target servers.

The legal, performance, and compliance implications of incorrect attestations based on incorrect (and potentially subverted) data are also important. Internet services today have complex, multi-site data dependencies, though, so these issues aren’t specific to \systemname. Oracle services already rely on multiple data sources to help ensure  correctness~\cite{SCCL:2017}. Oracle services in general could ultimately spawn infrastructure like that for certificates, including online checking and revocation capabilities~\cite{myers1999x} and different tiers of security~\cite{biddle2009browser}.

\vspacecamerareadyonly{-2mm}
\section{Related Work}
\label{sec:related}

\boldhead{Application-layer data-provenance}
Signing content at the application layer is a way to prove data provenance.
For example, \cite{http-origin-signed-responses,http-signatures} aim to retrofit signing capabilities into HTTP. 
Application-layer solutions, however, suffer from poor modularity and reusability, as they are application-specific.
They also require application-layer key management, violating the principle of layer separation in that cryptographic keys are no longer confined to the TLS layer.

Cinderella~\cite{delignat2016cinderella} uses verifiable computation to convert X.509 certificates into other credential types. Its main drawback is that few users possess certificates.
Open ID Connect~\cite{openid} providers can issue signed claims about users. However, adoption is still sparse and claims are limited to basic info such as names and email addresses.

\boldhead{Server-facilitated TLS-layer solutions}
Several proposed TLS-layer data-provenance proofs~\cite{tlssign,tlsevidence,ritzdorf2017tls} require server-side modifications. TLS-N~\cite{ritzdorf2017tls} is a TLS 1.3 extension that enables a server to sign the session using the existing PKI, and also supports chunk-level redaction for privacy. We refer readers to~\cite{ritzdorf2017tls} and references therein for a survey of TLS-layer solutions. Server-facilitated solutions suffer from high adoption cost, as they involve modification to security-critical server code. Moreover, they only benefit users when server administrators are able to and choose to cooperate. %

\boldhead{Smart contract oracles} Oracles~\cite{buterin2014next,SCCL:2017,zhang2016town} relay authenticated data from, e.g., websites, to smart contracts.
TLSNotary~\cite{tlsnotary}, used by Provable~\cite{Provable:2019}, allows a third party auditor to attest to a TLS connection between a server and a client, but relies on deprecated TLS versions (1.1 or lower).
Town Crier~\cite{zhang2016town} is an oracle service that uses TEEs (e.g., Intel SGX) for publicly verifiable evidence of TLS sessions and privacy-preserving computation on session data. While flexible and efficient, it relies on TEEs, which some users may reject given recently reported vulnerabilities, e.g.,~\cite{DBLP:conf/uss/BulckMWGKPSWYS18}.

\boldhead{Selective opening with context integrity}
Selective opening, i.e., decrypting part of a ciphertext to a third party while proving its integrity, has been studied previously.
Sanitizable signatures~\cite{ateniese2005sanitizable,brzuska2009security,steinfeld2001content, miyazaki2005digitally} allow a signed document to be selectively revealed. TLS-N~\cite{ritzdorf2017tls} allows ``chunk-level'' redacting of TLS records.
These works, however, consider a weaker adversarial model than \systemname. They fail to address the critical property of context integrity.
\systemname enforces proofs of context integrity in the rigorous sense of~\cref{subsec:two stage parsing}, using a novel two-stage parsing scheme that achieves efficiency by greatly reducing the length of the input to the zero-knowledge proof. %
\ifcameraready
\else
\section{Conclusion}
We have introduced \systemname, a privacy-preserving, decentralized oracle scheme for modern TLS versions that requires no trusted hardware or server-side modifications.
\systemname allows users to efficiently prove provenance and fine-grained statements about session content. We also identified context-integrity attacks that are universal to privacy-preserving oracles and provided efficient mitigation in a novel two-stage parsing scheme.
We formalized decentralized oracles in an ideal functionality, providing the first such rigorous security definition.
\systemname can liberate private data from centralized web-service silos, making it accessible to a rich spectrum of applications. We demonstrated \systemname's practicality through a fully functional implementation along with three example applications.
\fi

\vspace{-1mm}
\section*{Acknowledgements} This work was funded by NSF grants CNS-1514163, CNS-1564102, CNS-1704615, and CNS-1933655, and ARO grant W911NF16-1-0145.

\vspace{1mm}
\noindent{\em Personal financial interests:} Ari Juels is a technical advisor to Chainlink Smartcontract LLC and Soluna. 
\bibliographystyle{ACM-Reference-Format}
\bibliography{biblio}


\begin{thebibliography}{80}


\ifx \showCODEN    \undefined \def \showCODEN     #1{\unskip}     \fi
\ifx \showDOI      \undefined \def \showDOI       #1{#1}\fi
\ifx \showISBNx    \undefined \def \showISBNx     #1{\unskip}     \fi
\ifx \showISBNxiii \undefined \def \showISBNxiii  #1{\unskip}     \fi
\ifx \showISSN     \undefined \def \showISSN      #1{\unskip}     \fi
\ifx \showLCCN     \undefined \def \showLCCN      #1{\unskip}     \fi
\ifx \shownote     \undefined \def \shownote      #1{#1}          \fi
\ifx \showarticletitle \undefined \def \showarticletitle #1{#1}   \fi
\ifx \showURL      \undefined \def \showURL       {\relax}        \fi
\providecommand\bibfield[2]{#2}
\providecommand\bibinfo[2]{#2}
\providecommand\natexlab[1]{#1}
\providecommand\showeprint[2][]{arXiv:#2}

\bibitem[\protect\citeauthoryear{??}{age}{[n.d.]}]%
        {age-veri}
 \bibinfo{year}{[n.d.]}\natexlab{}.
\newblock \bibinfo{title}{Age Checker}.
\newblock \bibinfo{howpublished}{\url{https://agechecker.net}}.
\newblock


\bibitem[\protect\citeauthoryear{??}{Tho}{[n.d.]}]%
        {ThousandEyes}
 \bibinfo{year}{[n.d.]}\natexlab{}.
\newblock \bibinfo{title}{{Best BGP Route Network Monitoring Solution {$\vert$}
  ThousandEyes}}.
\newblock
\newblock
\urldef\tempurl%
\url{https://www.thousandeyes.com/solutions/bgp-and-route-monitoring}
\showURL{%
\tempurl}


\bibitem[\protect\citeauthoryear{??}{bgp}{[n.d.]a}]%
        {bgpmon}
 \bibinfo{year}{[n.d.]}\natexlab{a}.
\newblock \bibinfo{title}{{BGPmon {$\vert$} BGPmon}}.
\newblock
\newblock
\urldef\tempurl%
\url{https://bgpmon.net}
\showURL{%
\tempurl}


\bibitem[\protect\citeauthoryear{??}{bgp}{[n.d.]b}]%
        {bgpstream}
 \bibinfo{year}{[n.d.]}\natexlab{b}.
\newblock \bibinfo{title}{{BGPStream}}.
\newblock
\newblock
\urldef\tempurl%
\url{https://bgpstream.com}
\showURL{%
\tempurl}


\bibitem[\protect\citeauthoryear{??}{lib}{[n.d.]}]%
        {libsnark}
 \bibinfo{year}{[n.d.]}\natexlab{}.
\newblock \bibinfo{title}{libsnark}.
\newblock
\newblock
\newblock
\shownote{\url{https://github.com/scipr-lab/libsnark}.}


\bibitem[\protect\citeauthoryear{??}{ope}{[n.d.]}]%
        {openid}
 \bibinfo{year}{[n.d.]}\natexlab{}.
\newblock \bibinfo{title}{Open ID Connect}.
\newblock
\newblock
\urldef\tempurl%
\url{https://openid.net/connect}
\showURL{%
\tempurl}


\bibitem[\protect\citeauthoryear{??}{tls}{[n.d.]}]%
        {tlsnotary}
 \bibinfo{year}{[n.d.]}\natexlab{}.
\newblock \bibinfo{title}{{TLSN}otary}.
\newblock
\newblock
\newblock
\shownote{\url{https://tlsnotary.org/}.}


\bibitem[\protect\citeauthoryear{??}{gdp}{2014}]%
        {gdpr}
 \bibinfo{year}{2014}\natexlab{}.
\newblock \bibinfo{title}{Art. 20, {GDPR}, {R}ight to data portability}.
\newblock
\newblock
\newblock
\shownote{\url{https://gdpr-info.eu/art-20-gdpr/}.}


\bibitem[\protect\citeauthoryear{??}{bin}{2019}]%
        {binaryoption}
 \bibinfo{year}{2019}\natexlab{}.
\newblock \bibinfo{title}{Binary option}.
\newblock
\newblock
\newblock
\shownote{\url{https://en.wikipedia.org/wiki/Binary_option}.}


\bibitem[\protect\citeauthoryear{??}{Pro}{2019}]%
        {Provable:2019}
 \bibinfo{year}{2019}\natexlab{}.
\newblock \bibinfo{title}{Provable blockchain oracle}.
\newblock \bibinfo{howpublished}{\url{http://provable.xyz}}.
\newblock


\bibitem[\protect\citeauthoryear{??}{aes}{2019}]%
        {aescircuit}
 \bibinfo{year}{Aug 2019}\natexlab{}.
\newblock \bibinfo{title}{({Bristol Format}) Circuits of Basic Functions
  Suitable For {MPC}}.
\newblock
\newblock
\newblock
\shownote{\url{https://homes.esat.kuleuven.be/~nsmart/MPC/old-circuits.html}.}


\bibitem[\protect\citeauthoryear{Adler, Berryhill, Veneris, Poulos, Veira, and
  Kastania}{Adler et~al\mbox{.}}{2018}]%
        {adler2018astraea}
\bibfield{author}{\bibinfo{person}{John Adler}, \bibinfo{person}{Ryan
  Berryhill}, \bibinfo{person}{Andreas~G. Veneris}, \bibinfo{person}{Zissis
  Poulos}, \bibinfo{person}{Neil Veira}, {and} \bibinfo{person}{Anastasia
  Kastania}.} \bibinfo{year}{2018}\natexlab{}.
\newblock \showarticletitle{Astraea: {A} Decentralized Blockchain Oracle}. In
  \bibinfo{booktitle}{\emph{{IEEE} iThings/GreenCom/CPSCom/SmartData}}.
\newblock


\bibitem[\protect\citeauthoryear{Aranha and Gouv\^{e}a}{Aranha and
  Gouv\^{e}a}{[n.d.]}]%
        {relic}
\bibfield{author}{\bibinfo{person}{D.~F. Aranha} {and}
  \bibinfo{person}{C.~P.~L. Gouv\^{e}a}.} \bibinfo{year}{[n.d.]}\natexlab{}.
\newblock \bibinfo{title}{{RELIC is an Efficient LIbrary for Cryptography}}.
\newblock \bibinfo{howpublished}{\url{https://github.com/relic-toolkit/relic}}.
\newblock


\bibitem[\protect\citeauthoryear{{ARM}}{{ARM}}{2019}]%
        {mbedTLS}
\bibfield{author}{\bibinfo{person}{{ARM}}.} \bibinfo{year}{2019}\natexlab{}.
\newblock \bibinfo{title}{mbed{TLS}}.
\newblock \bibinfo{howpublished}{\url{https://github.com/ARMmbed/mbedtls}}.
\newblock


\bibitem[\protect\citeauthoryear{Association}{Association}{[n.d.]}]%
        {ABA:2019}
\bibfield{author}{\bibinfo{person}{{}American~Bar Association}.}
  \bibinfo{year}{[n.d.]}\natexlab{}.
\newblock
\newblock


\bibitem[\protect\citeauthoryear{Ateniese, Chou, De~Medeiros, and
  Tsudik}{Ateniese et~al\mbox{.}}{2005}]%
        {ateniese2005sanitizable}
\bibfield{author}{\bibinfo{person}{Giuseppe Ateniese},
  \bibinfo{person}{Daniel~H Chou}, \bibinfo{person}{Breno De~Medeiros}, {and}
  \bibinfo{person}{Gene Tsudik}.} \bibinfo{year}{2005}\natexlab{}.
\newblock \showarticletitle{Sanitizable signatures}. In
  \bibinfo{booktitle}{\emph{ESORICS}}.
\newblock


\bibitem[\protect\citeauthoryear{Bar-Ilan and Beaver}{Bar-Ilan and
  Beaver}{1989}]%
        {bar1989non}
\bibfield{author}{\bibinfo{person}{Judit Bar-Ilan} {and}
  \bibinfo{person}{Donald Beaver}.} \bibinfo{year}{1989}\natexlab{}.
\newblock \showarticletitle{Non-cryptographic fault-tolerant computing in
  constant number of rounds of interaction}. In \bibinfo{booktitle}{\emph{ACM
  PODC}}.
\newblock


\bibitem[\protect\citeauthoryear{Ben{-}Sasson, Chiesa, Tromer, and
  Virza}{Ben{-}Sasson et~al\mbox{.}}{2014}]%
        {ben2014succinct}
\bibfield{author}{\bibinfo{person}{Eli Ben{-}Sasson},
  \bibinfo{person}{Alessandro Chiesa}, \bibinfo{person}{Eran Tromer}, {and}
  \bibinfo{person}{Madars Virza}.} \bibinfo{year}{2014}\natexlab{}.
\newblock \showarticletitle{Succinct Non-Interactive Zero Knowledge for a von
  Neumann Architecture}. In \bibinfo{booktitle}{\emph{{USENIX} Security}}.
\newblock


\bibitem[\protect\citeauthoryear{Biddle, Van~Oorschot, Patrick, Sobey, and
  Whalen}{Biddle et~al\mbox{.}}{2009}]%
        {biddle2009browser}
\bibfield{author}{\bibinfo{person}{Robert Biddle}, \bibinfo{person}{Paul~C
  Van~Oorschot}, \bibinfo{person}{Andrew~S Patrick}, \bibinfo{person}{Jennifer
  Sobey}, {and} \bibinfo{person}{Tara Whalen}.}
  \bibinfo{year}{2009}\natexlab{}.
\newblock \showarticletitle{Browser interfaces and extended validation SSL
  certificates: an empirical study}. In \bibinfo{booktitle}{\emph{ACM workshop
  on Cloud computing security}}. \bibinfo{pages}{19--30}.
\newblock


\bibitem[\protect\citeauthoryear{Blake-Wilson, Bolyard, Gupta, Hawk, and
  Moeller}{Blake-Wilson et~al\mbox{.}}{2006}]%
        {ecc-tls}
\bibfield{author}{\bibinfo{person}{S. Blake-Wilson}, \bibinfo{person}{N.
  Bolyard}, \bibinfo{person}{V. Gupta}, \bibinfo{person}{C. Hawk}, {and}
  \bibinfo{person}{B. Moeller}.} \bibinfo{year}{2006}\natexlab{}.
\newblock \bibinfo{booktitle}{\emph{Elliptic Curve Cryptography ({ECC}) Cipher
  Suites for Transport Layer Security ({TLS})}}.
\newblock \bibinfo{type}{RFC} 4492.
\newblock


\bibitem[\protect\citeauthoryear{Borgesius and Poort}{Borgesius and
  Poort}{2017}]%
        {borgesius2017online}
\bibfield{author}{\bibinfo{person}{Frederik~Zuiderveen Borgesius} {and}
  \bibinfo{person}{Joost Poort}.} \bibinfo{year}{2017}\natexlab{}.
\newblock \showarticletitle{Online price discrimination and EU data privacy
  law}.
\newblock \bibinfo{journal}{\emph{Journal of consumer policy}}
  (\bibinfo{year}{2017}).
\newblock


\bibitem[\protect\citeauthoryear{Brown and Housle}{Brown and Housle}{2007}]%
        {tlsevidence}
\bibfield{author}{\bibinfo{person}{Mark Brown} {and} \bibinfo{person}{Russ
  Housle}.} \bibinfo{year}{2007}\natexlab{}.
\newblock \bibinfo{title}{Transport Layer Security ({TLS}) Evidence
  Extensions}.
\newblock
\newblock
\newblock
\shownote{\url{https://tools.ietf.org/html/draft-housley-evidence-extns-01}.}


\bibitem[\protect\citeauthoryear{Brzuska, Fischlin, Freudenreich, Lehmann,
  Page, Schelbert, Schr{\"o}der, and Volk}{Brzuska et~al\mbox{.}}{2009}]%
        {brzuska2009security}
\bibfield{author}{\bibinfo{person}{Christina Brzuska}, \bibinfo{person}{Marc
  Fischlin}, \bibinfo{person}{Tobias Freudenreich}, \bibinfo{person}{Anja
  Lehmann}, \bibinfo{person}{Marcus Page}, \bibinfo{person}{Jakob Schelbert},
  \bibinfo{person}{Dominique Schr{\"o}der}, {and} \bibinfo{person}{Florian
  Volk}.} \bibinfo{year}{2009}\natexlab{}.
\newblock \showarticletitle{Security of sanitizable signatures revisited}. In
  \bibinfo{booktitle}{\emph{PKC}}. Springer.
\newblock


\bibitem[\protect\citeauthoryear{Bulck, Minkin, Weisse, Genkin, Kasikci,
  Piessens, Silberstein, Wenisch, Yarom, and Strackx}{Bulck
  et~al\mbox{.}}{2018}]%
        {DBLP:conf/uss/BulckMWGKPSWYS18}
\bibfield{author}{\bibinfo{person}{Jo~Van Bulck}, \bibinfo{person}{Marina
  Minkin}, \bibinfo{person}{Ofir Weisse}, \bibinfo{person}{Daniel Genkin},
  \bibinfo{person}{Baris Kasikci}, \bibinfo{person}{Frank Piessens},
  \bibinfo{person}{Mark Silberstein}, \bibinfo{person}{Thomas~F. Wenisch},
  \bibinfo{person}{Yuval Yarom}, {and} \bibinfo{person}{Raoul Strackx}.}
  \bibinfo{year}{2018}\natexlab{}.
\newblock \showarticletitle{Foreshadow: Extracting the Keys to the Intel {SGX}
  Kingdom with Transient Out-of-Order Execution}. In
  \bibinfo{booktitle}{\emph{{USENIX} Security}}.
\newblock


\bibitem[\protect\citeauthoryear{B{\"{u}}nz, Bootle, Boneh, Poelstra, Wuille,
  and Maxwell}{B{\"{u}}nz et~al\mbox{.}}{2018}]%
        {DBLP:conf/sp/BunzBBPWM18}
\bibfield{author}{\bibinfo{person}{Benedikt B{\"{u}}nz},
  \bibinfo{person}{Jonathan Bootle}, \bibinfo{person}{Dan Boneh},
  \bibinfo{person}{Andrew Poelstra}, \bibinfo{person}{Pieter Wuille}, {and}
  \bibinfo{person}{Gregory Maxwell}.} \bibinfo{year}{2018}\natexlab{}.
\newblock \showarticletitle{Bulletproofs: Short Proofs for Confidential
  Transactions and More}. In \bibinfo{booktitle}{\emph{{IEEE} S\&P}}.
\newblock


\bibitem[\protect\citeauthoryear{Buterin et~al\mbox{.}}{Buterin
  et~al\mbox{.}}{2014}]%
        {buterin2014next}
\bibfield{author}{\bibinfo{person}{Vitalik Buterin} {et~al\mbox{.}}}
  \bibinfo{year}{2014}\natexlab{}.
\newblock \showarticletitle{A next-generation smart contract and decentralized
  application platform}.
\newblock \bibinfo{journal}{\emph{white paper}}  \bibinfo{volume}{3}
  (\bibinfo{year}{2014}), \bibinfo{pages}{37}.
\newblock


\bibitem[\protect\citeauthoryear{Butler, Farley, McDaniel, and Rexford}{Butler
  et~al\mbox{.}}{2010}]%
        {butler2010bgpsurvey}
\bibfield{author}{\bibinfo{person}{Kevin R.~B. Butler},
  \bibinfo{person}{Toni~R. Farley}, \bibinfo{person}{Patrick~D. McDaniel},
  {and} \bibinfo{person}{Jennifer Rexford}.} \bibinfo{year}{2010}\natexlab{}.
\newblock \showarticletitle{A Survey of {BGP} Security Issues and Solutions}.
\newblock \bibinfo{journal}{\emph{Proc. IEEE}} \bibinfo{volume}{98},
  \bibinfo{number}{1} (\bibinfo{year}{2010}), \bibinfo{pages}{100--122}.
\newblock


\bibitem[\protect\citeauthoryear{Camenisch and Stadler}{Camenisch and
  Stadler}{1997}]%
        {camenisch1997efficient}
\bibfield{author}{\bibinfo{person}{Jan Camenisch} {and} \bibinfo{person}{Markus
  Stadler}.} \bibinfo{year}{1997}\natexlab{}.
\newblock \showarticletitle{Efficient group signature schemes for large
  groups}. In \bibinfo{booktitle}{\emph{Annual International Cryptology
  Conference}}.
\newblock


\bibitem[\protect\citeauthoryear{Campanelli, Gennaro, Goldfeder, and
  Nizzardo}{Campanelli et~al\mbox{.}}{2017}]%
        {Campanelli:2017:ZCP:3133956.3134060}
\bibfield{author}{\bibinfo{person}{Matteo Campanelli}, \bibinfo{person}{Rosario
  Gennaro}, \bibinfo{person}{Steven Goldfeder}, {and} \bibinfo{person}{Luca
  Nizzardo}.} \bibinfo{year}{2017}\natexlab{}.
\newblock \showarticletitle{Zero-Knowledge Contingent Payments Revisited:
  Attacks and Payments for Services}. In \bibinfo{booktitle}{\emph{ACM CCS}}.
\newblock


\bibitem[\protect\citeauthoryear{Canetti}{Canetti}{2000}]%
        {uc}
\bibfield{author}{\bibinfo{person}{Ran Canetti}.}
  \bibinfo{year}{2000}\natexlab{}.
\newblock \bibinfo{title}{Universally Composable Security: A New Paradigm for
  Cryptographic Protocols}.
\newblock \bibinfo{howpublished}{Cryptology ePrint Archive, Report 2000/067}.
\newblock
\newblock
\shownote{\url{https://eprint.iacr.org/2000/067}.}


\bibitem[\protect\citeauthoryear{Cavage and Sporny}{Cavage and Sporny}{2019}]%
        {http-signatures}
\bibfield{author}{\bibinfo{person}{Mark Cavage} {and} \bibinfo{person}{Manu
  Sporny}.} \bibinfo{year}{2019}\natexlab{}.
\newblock \bibinfo{booktitle}{\emph{Signing {HTTP} Messages}}.
\newblock \bibinfo{type}{Internet-Draft} draft-cavage-http-signatures-11.
\newblock


\bibitem[\protect\citeauthoryear{CFTC}{CFTC}{2018}]%
        {cftc_2018}
\bibfield{author}{\bibinfo{person}{CFTC}.} \bibinfo{year}{2018}\natexlab{}.
\newblock \bibinfo{title}{A Primer on Smart Contrats}.
\newblock
  \bibinfo{howpublished}{\url{https://www.cftc.gov/sites/default/files/2018-11/LabCFTC_PrimerSmartContracts112718.pdf}}.
\newblock


\bibitem[\protect\citeauthoryear{Delignat-Lavaud, Fournet, Kohlweiss, and
  Parno}{Delignat-Lavaud et~al\mbox{.}}{2016}]%
        {delignat2016cinderella}
\bibfield{author}{\bibinfo{person}{Antoine Delignat-Lavaud},
  \bibinfo{person}{C{\'e}dric Fournet}, \bibinfo{person}{Markulf Kohlweiss},
  {and} \bibinfo{person}{Bryan Parno}.} \bibinfo{year}{2016}\natexlab{}.
\newblock \showarticletitle{Cinderella: Turning shabby X. 509 certificates into
  elegant anonymous credentials with the magic of verifiable computation}. In
  \bibinfo{booktitle}{\emph{IEEE S\&P}}.
\newblock


\bibitem[\protect\citeauthoryear{Demmler, Dessouky, Koushanfar, Sadeghi,
  Schneider, and Zeitouni}{Demmler et~al\mbox{.}}{2015}]%
        {demmler2015automated}
\bibfield{author}{\bibinfo{person}{Daniel Demmler}, \bibinfo{person}{Ghada
  Dessouky}, \bibinfo{person}{Farinaz Koushanfar}, \bibinfo{person}{Ahmad-Reza
  Sadeghi}, \bibinfo{person}{Thomas Schneider}, {and} \bibinfo{person}{Shaza
  Zeitouni}.} \bibinfo{year}{2015}\natexlab{}.
\newblock \showarticletitle{Automated synthesis of optimized circuits for
  secure computation}. In \bibinfo{booktitle}{\emph{ACM CCS}}.
\newblock


\bibitem[\protect\citeauthoryear{Deshpande, Thottan, Ho, and Sikdar}{Deshpande
  et~al\mbox{.}}{2009}]%
        {deshpande2009online}
\bibfield{author}{\bibinfo{person}{Shivani Deshpande}, \bibinfo{person}{Marina
  Thottan}, \bibinfo{person}{Tin~Kam Ho}, {and} \bibinfo{person}{Biplab
  Sikdar}.} \bibinfo{year}{2009}\natexlab{}.
\newblock \showarticletitle{An online mechanism for BGP instability detection
  and analysis}.
\newblock \bibinfo{journal}{\emph{IEEE transactions on Computers}}
  \bibinfo{volume}{58}, \bibinfo{number}{11} (\bibinfo{year}{2009}),
  \bibinfo{pages}{1470--1484}.
\newblock


\bibitem[\protect\citeauthoryear{Dierks and Rescorla}{Dierks and
  Rescorla}{2008}]%
        {RFC5246}
\bibfield{author}{\bibinfo{person}{T. Dierks} {and} \bibinfo{person}{E.
  Rescorla}.} \bibinfo{year}{2008}\natexlab{}.
\newblock \bibinfo{booktitle}{\emph{The Transport Layer Security ({TLS})
  Protocol Version 1.2}}.
\newblock \bibinfo{type}{RFC} 5246.
\newblock


\bibitem[\protect\citeauthoryear{Dowling, Fischlin, G{\"{u}}nther, and
  Stebila}{Dowling et~al\mbox{.}}{2015}]%
        {tls13handshake}
\bibfield{author}{\bibinfo{person}{Benjamin Dowling}, \bibinfo{person}{Marc
  Fischlin}, \bibinfo{person}{Felix G{\"{u}}nther}, {and}
  \bibinfo{person}{Douglas Stebila}.} \bibinfo{year}{2015}\natexlab{}.
\newblock \showarticletitle{A Cryptographic Analysis of the {TLS} 1.3 Handshake
  Protocol Candidates}. In \bibinfo{booktitle}{\emph{{ACM} Conference on
  Computer and Communications Security}}. \bibinfo{publisher}{{ACM}},
  \bibinfo{pages}{1197--1210}.
\newblock


\bibitem[\protect\citeauthoryear{Dworkin}{Dworkin}{2007}]%
        {dworkin2007sp}
\bibfield{author}{\bibinfo{person}{Morris~J Dworkin}.}
  \bibinfo{year}{2007}\natexlab{}.
\newblock \bibinfo{booktitle}{\emph{{SP} 800-38d. Recommendation for block
  cipher modes of operation: Galois/counter mode ({GCM}) and {GMAC}}}.
\newblock \bibinfo{type}{{T}echnical {R}eport}.
\newblock


\bibitem[\protect\citeauthoryear{Ellis, Juels, and Nazarov}{Ellis
  et~al\mbox{.}}{2017}]%
        {SCCL:2017}
\bibfield{author}{\bibinfo{person}{Steve Ellis}, \bibinfo{person}{Ari Juels},
  {and} \bibinfo{person}{Sergey Nazarov}.} \bibinfo{year}{4
  Sept.~2017}\natexlab{}.
\newblock \bibinfo{title}{ChainLink: A Decentralized Oracle Network}.
\newblock
  \bibinfo{howpublished}{\url{https://link.smartcontract.com/whitepaper}}.
\newblock


\bibitem[\protect\citeauthoryear{FTC}{FTC}{2017}]%
        {ftc2017}
\bibfield{author}{\bibinfo{person}{FTC}.} \bibinfo{year}{2017}\natexlab{}.
\newblock \bibinfo{title}{Price Discrimination: Robinson-Patman Violations}.
\newblock
  \bibinfo{howpublished}{\url{https://www.ftc.gov/tips-advice/competition-guidance/guide-antitrust-laws/price-discrimination-robinson-patman}}.
\newblock


\bibitem[\protect\citeauthoryear{Gennaro and Goldfeder}{Gennaro and
  Goldfeder}{2018}]%
        {gennaro2018fast}
\bibfield{author}{\bibinfo{person}{Rosario Gennaro} {and}
  \bibinfo{person}{Steven Goldfeder}.} \bibinfo{year}{2018}\natexlab{}.
\newblock \showarticletitle{Fast multiparty threshold {ECDSA} with fast
  trustless setup}. In \bibinfo{booktitle}{\emph{ACM CCS}}.
\newblock


\bibitem[\protect\citeauthoryear{Grubbs, Lu, and Ristenpart}{Grubbs
  et~al\mbox{.}}{2017}]%
        {DBLP:conf/crypto/GrubbsLR17}
\bibfield{author}{\bibinfo{person}{Paul Grubbs}, \bibinfo{person}{Jiahui Lu},
  {and} \bibinfo{person}{Thomas Ristenpart}.} \bibinfo{year}{2017}\natexlab{}.
\newblock \showarticletitle{Message Franking via Committing Authenticated
  Encryption}. In \bibinfo{booktitle}{\emph{{CRYPTO}}}.
\newblock


\bibitem[\protect\citeauthoryear{Grune}{Grune}{2010}]%
        {grune2007parsing}
\bibfield{author}{\bibinfo{person}{Dick Grune}.}
  \bibinfo{year}{2010}\natexlab{}.
\newblock \bibinfo{booktitle}{\emph{Parsing Techniques: A Practical Guide}
  (\bibinfo{edition}{2nd} ed.)}.
\newblock \bibinfo{publisher}{Springer}.
\newblock


\bibitem[\protect\citeauthoryear{Hajjeh and Badra}{Hajjeh and Badra}{2017}]%
        {tlssign}
\bibfield{author}{\bibinfo{person}{Ibrahim Hajjeh} {and}
  \bibinfo{person}{Mohamad Badra}.} \bibinfo{year}{2017}\natexlab{}.
\newblock \bibinfo{title}{{TLS} Sign}.
\newblock
\newblock
\newblock
\shownote{\url{https://tools.ietf.org/html/draft-hajjeh-tls-sign-04}.}


\bibitem[\protect\citeauthoryear{Hazay and Lindell}{Hazay and Lindell}{2010}]%
        {zkfunc}
\bibfield{author}{\bibinfo{person}{Carmit Hazay} {and} \bibinfo{person}{Yehuda
  Lindell}.} \bibinfo{year}{2010}\natexlab{}.
\newblock \bibinfo{title}{A Note on Zero-Knowledge Proofs of Knowledge and the
  {ZKPOK} Ideal Functionality}.
\newblock \bibinfo{howpublished}{Cryptology ePrint Archive, Report 2010/552}.
\newblock
\newblock
\shownote{\url{https://eprint.iacr.org/2010/552}.}


\bibitem[\protect\citeauthoryear{Hofmann}{Hofmann}{2010}]%
        {EFF:2010}
\bibfield{author}{\bibinfo{person}{Marcia Hofmann}.} \bibinfo{year}{21 July
  2010}\natexlab{}.
\newblock \bibinfo{title}{Court: Violating Terms of Service Is Not a Crime, But
  Bypassing Technical Barriers Might Be}.
\newblock \bibinfo{howpublished}{Electronic Frontier Foundation (EFF) News
  Update}.
\newblock


\bibitem[\protect\citeauthoryear{Howe}{Howe}{2017}]%
        {howe_2017}
\bibfield{author}{\bibinfo{person}{Neil Howe}.}
  \bibinfo{year}{2017}\natexlab{}.
\newblock \showarticletitle{A Special Price Just for You}.
\newblock \bibinfo{journal}{\emph{Forbes}} (\bibinfo{date}{Nov}
  \bibinfo{year}{2017}).
\newblock
\newblock
\shownote{\url{https://www.forbes.com/sites/neilhowe/2017/11/17/a-special-price-just-for-you/}.}


\bibitem[\protect\citeauthoryear{Juels, Breidenbach, Coventry, Nazarov, and
  Ellis}{Juels et~al\mbox{.}}{2019}]%
        {Mixicles:2019}
\bibfield{author}{\bibinfo{person}{Ari Juels}, \bibinfo{person}{Lorenz
  Breidenbach}, \bibinfo{person}{Alex Coventry}, \bibinfo{person}{Sergey
  Nazarov}, {and} \bibinfo{person}{Steve Ellis}.}
  \bibinfo{year}{2019}\natexlab{}.
\newblock \bibinfo{title}{Mixicles: Private Decentralized Finance Made Simple}.
\newblock
\newblock
\newblock
\shownote{Chainlink whitepaper.}


\bibitem[\protect\citeauthoryear{Khoury}{Khoury}{2018}]%
        {Khoury:2018}
\bibfield{author}{\bibinfo{person}{George Khoury}.} \bibinfo{year}{24
  Jan.~2018}\natexlab{}.
\newblock \bibinfo{title}{Violation of a Website's Terms of Service is Not
  Criminal}.
\newblock \bibinfo{howpublished}{Findlaw blog post}.
\newblock


\bibitem[\protect\citeauthoryear{Kosba, Papamanthou, and Shi}{Kosba
  et~al\mbox{.}}{2018}]%
        {kosba2018xjsnark}
\bibfield{author}{\bibinfo{person}{Ahmed~E. Kosba},
  \bibinfo{person}{Charalampos Papamanthou}, {and} \bibinfo{person}{Elaine
  Shi}.} \bibinfo{year}{2018}\natexlab{}.
\newblock \showarticletitle{xJsnark: {A} Framework for Efficient Verifiable
  Computation}. In \bibinfo{booktitle}{\emph{{IEEE} S\&P}}.
\newblock


\bibitem[\protect\citeauthoryear{Kr{\"{u}}gel, Mutz, Robertson, and
  Valeur}{Kr{\"{u}}gel et~al\mbox{.}}{2003}]%
        {kruegel2003topology}
\bibfield{author}{\bibinfo{person}{Christopher Kr{\"{u}}gel},
  \bibinfo{person}{Darren Mutz}, \bibinfo{person}{William~K. Robertson}, {and}
  \bibinfo{person}{Fredrik Valeur}.} \bibinfo{year}{2003}\natexlab{}.
\newblock \showarticletitle{Topology-Based Detection of Anomalous {BGP}
  Messages}. In \bibinfo{booktitle}{\emph{{RAID}}}
  \emph{(\bibinfo{series}{Lecture Notes in Computer Science},
  Vol.~\bibinfo{volume}{2820})}. \bibinfo{publisher}{Springer},
  \bibinfo{pages}{17--35}.
\newblock


\bibitem[\protect\citeauthoryear{Lindell}{Lindell}{2005}]%
        {lindell2005secure}
\bibfield{author}{\bibinfo{person}{Yehida Lindell}.}
  \bibinfo{year}{2005}\natexlab{}.
\newblock \showarticletitle{Secure multiparty computation for privacy
  preserving data mining}.
\newblock In \bibinfo{booktitle}{\emph{Encyclopedia of Data Warehousing and
  Mining}}. \bibinfo{publisher}{IGI Global}.
\newblock


\bibitem[\protect\citeauthoryear{Maram, Malvai, Zhang, Jean-Louis, Frolov,
  Kell, Lobban, Moy, Juels, and Miller}{Maram et~al\mbox{.}}{2020}]%
        {candid}
\bibfield{author}{\bibinfo{person}{Deepak Maram}, \bibinfo{person}{Harjasleen
  Malvai}, \bibinfo{person}{Fan Zhang}, \bibinfo{person}{Nerla Jean-Louis},
  \bibinfo{person}{Alexander Frolov}, \bibinfo{person}{Tyler Kell},
  \bibinfo{person}{Tyrone Lobban}, \bibinfo{person}{Christine Moy},
  \bibinfo{person}{Ari Juels}, {and} \bibinfo{person}{Andrew Miller}.}
  \bibinfo{year}{2020}\natexlab{}.
\newblock \bibinfo{title}{CanDID: Can-Do Decentralized Identity with Legacy
  Compatibility, Sybil-Resistance, and Accountability}.
\newblock \bibinfo{howpublished}{Cryptology ePrint Archive, Report 2020/934}.
\newblock
\newblock
\shownote{\url{https://eprint.iacr.org/2020/934}.}


\bibitem[\protect\citeauthoryear{Matetic, Schneider, Miller, Juels, and
  Capkun}{Matetic et~al\mbox{.}}{2018}]%
        {delegatee}
\bibfield{author}{\bibinfo{person}{Sinisa Matetic}, \bibinfo{person}{Moritz
  Schneider}, \bibinfo{person}{Andrew Miller}, \bibinfo{person}{Ari Juels},
  {and} \bibinfo{person}{Srdjan Capkun}.} \bibinfo{year}{2018}\natexlab{}.
\newblock \showarticletitle{DelegaTEE: Brokered Delegation Using Trusted
  Execution Environments}. In \bibinfo{booktitle}{\emph{USENIX Security}}.
\newblock


\bibitem[\protect\citeauthoryear{Miyazaki, Iwamura, Matsumoto, Sasaki,
  Yoshiura, Tezuka, and Imai}{Miyazaki et~al\mbox{.}}{2005}]%
        {miyazaki2005digitally}
\bibfield{author}{\bibinfo{person}{Kunihiko Miyazaki}, \bibinfo{person}{Mitsuru
  Iwamura}, \bibinfo{person}{Tsutomu Matsumoto}, \bibinfo{person}{Ry{\^{o}}ichi
  Sasaki}, \bibinfo{person}{Hiroshi Yoshiura}, \bibinfo{person}{Satoru Tezuka},
  {and} \bibinfo{person}{Hideki Imai}.} \bibinfo{year}{2005}\natexlab{}.
\newblock \showarticletitle{Digitally Signed Document Sanitizing Scheme with
  Disclosure Condition Control}.
\newblock \bibinfo{journal}{\emph{{IEICE} Transactions}}
  (\bibinfo{year}{2005}).
\newblock


\bibitem[\protect\citeauthoryear{Myers, Ankney, Malpani, Galperin, and
  Adams}{Myers et~al\mbox{.}}{1999}]%
        {myers1999x}
\bibfield{author}{\bibinfo{person}{Michael Myers}, \bibinfo{person}{Rich
  Ankney}, \bibinfo{person}{Ambarish Malpani}, \bibinfo{person}{Slava
  Galperin}, {and} \bibinfo{person}{Carlisle Adams}.}
  \bibinfo{year}{1999}\natexlab{}.
\newblock \bibinfo{booktitle}{\emph{X. 509 Internet public key infrastructure
  online certificate status protocol-OCSP}}.
\newblock \bibinfo{type}{{T}echnical {R}eport}. \bibinfo{institution}{RFC
  2560}.
\newblock


\bibitem[\protect\citeauthoryear{Nissenbaum}{Nissenbaum}{2009}]%
        {nissenbaum2009privacy}
\bibfield{author}{\bibinfo{person}{Helen Nissenbaum}.}
  \bibinfo{year}{2009}\natexlab{}.
\newblock \bibinfo{booktitle}{\emph{Privacy in context: Technology, policy, and
  the integrity of social life}}.
\newblock \bibinfo{publisher}{Stanford University Press}.
\newblock


\bibitem[\protect\citeauthoryear{Oded}{Oded}{2009}]%
        {simulation}
\bibfield{author}{\bibinfo{person}{Goldreich Oded}.}
  \bibinfo{year}{2009}\natexlab{}.
\newblock \bibinfo{booktitle}{\emph{Foundations of Cryptography: Volume 2,
  Basic Applications} (\bibinfo{edition}{1st} ed.)}.
\newblock \bibinfo{publisher}{Cambridge University Press}.
\newblock
\showISBNx{052111991X, 9780521119917}


\bibitem[\protect\citeauthoryear{Odlyzko}{Odlyzko}{2003}]%
        {odlyzko2003privacy}
\bibfield{author}{\bibinfo{person}{Andrew Odlyzko}.}
  \bibinfo{year}{2003}\natexlab{}.
\newblock \showarticletitle{Privacy, economics, and price discrimination on the
  Internet}. In \bibinfo{booktitle}{\emph{5th international conference on
  Electronic commerce}}.
\newblock


\bibitem[\protect\citeauthoryear{@OpenLawOfficial}{@OpenLawOfficial}{2018}]%
        {openlawofficial_2018}
\bibfield{author}{\bibinfo{person}{@OpenLawOfficial}.}
  \bibinfo{year}{2018}\natexlab{}.
\newblock \bibinfo{title}{The Future of Derivatives: An End-to-End, Legally
  Enforceable Option Contract Powered by Ethereum}.
\newblock
\newblock


\bibitem[\protect\citeauthoryear{Paillier}{Paillier}{1999}]%
        {DBLP:conf/eurocrypt/Paillier99}
\bibfield{author}{\bibinfo{person}{Pascal Paillier}.}
  \bibinfo{year}{1999}\natexlab{}.
\newblock \showarticletitle{Public-Key Cryptosystems Based on Composite Degree
  Residuosity Classes}. In \bibinfo{booktitle}{\emph{{EUROCRYPT}}}.
\newblock


\bibitem[\protect\citeauthoryear{Peterson and Krug}{Peterson and Krug}{2015}]%
        {peterson2015augur}
\bibfield{author}{\bibinfo{person}{Jack Peterson} {and} \bibinfo{person}{Joseph
  Krug}.} \bibinfo{year}{2015}\natexlab{}.
\newblock \showarticletitle{Augur: a decentralized, open-source platform for
  prediction markets}.
\newblock \bibinfo{journal}{\emph{arXiv:1501.01042}} (\bibinfo{year}{2015}).
\newblock


\bibitem[\protect\citeauthoryear{Projects}{Projects}{[n.d.]}]%
        {spdy-whitepaper}
\bibfield{author}{\bibinfo{person}{The~Chromium Projects}.}
  \bibinfo{year}{[n.d.]}\natexlab{}.
\newblock \bibinfo{title}{The {SPDY} whitepaper}.
\newblock
  \bibinfo{howpublished}{\url{https://dev.chromium.org/spdy/spdy-whitepaper}}.
\newblock


\bibitem[\protect\citeauthoryear{Rescorla}{Rescorla}{2018}]%
        {RFC8446}
\bibfield{author}{\bibinfo{person}{E. Rescorla}.}
  \bibinfo{year}{2018}\natexlab{}.
\newblock \bibinfo{booktitle}{\emph{The Transport Layer Security ({TLS})
  Protocol Version 1.3}}.
\newblock \bibinfo{type}{RFC} 8446.
\newblock


\bibitem[\protect\citeauthoryear{Ritzdorf, W{\"u}st, Gervais, Felley, and
  {\v{C}}apkun}{Ritzdorf et~al\mbox{.}}{2018}]%
        {ritzdorf2017tls}
\bibfield{author}{\bibinfo{person}{Hubert Ritzdorf}, \bibinfo{person}{Karl
  W{\"u}st}, \bibinfo{person}{Arthur Gervais}, \bibinfo{person}{Guillaume
  Felley}, {and} \bibinfo{person}{Srdjan {\v{C}}apkun}.}
  \bibinfo{year}{2018}\natexlab{}.
\newblock \showarticletitle{{TLS}-{N}: Non-repudiation over {TLS} Enabling
  Ubiquitous Content Signing.}. In \bibinfo{booktitle}{\emph{NDSS}}.
\newblock


\bibitem[\protect\citeauthoryear{Salowey, Choudhury, and McGrew}{Salowey
  et~al\mbox{.}}{2008}]%
        {gcm-in-tls}
\bibfield{author}{\bibinfo{person}{J. Salowey}, \bibinfo{person}{A. Choudhury},
  {and} \bibinfo{person}{D. McGrew}.} \bibinfo{year}{2008}\natexlab{}.
\newblock \bibinfo{booktitle}{\emph{{AES} {G}alois Counter Mode ({GCM}) Cipher
  Suites for {TLS}}}.
\newblock \bibinfo{type}{RFC} 5288.
\newblock


\bibitem[\protect\citeauthoryear{Saxena, Molnar, and Livshits}{Saxena
  et~al\mbox{.}}{2011}]%
        {DBLP:conf/ccs/SaxenaML11}
\bibfield{author}{\bibinfo{person}{Prateek Saxena}, \bibinfo{person}{David
  Molnar}, {and} \bibinfo{person}{Benjamin Livshits}.}
  \bibinfo{year}{2011}\natexlab{}.
\newblock \showarticletitle{{SCRIPTGARD:} automatic context-sensitive
  sanitization for large-scale legacy web applications}. In
  \bibinfo{booktitle}{\emph{ACM CCS}}.
\newblock


\bibitem[\protect\citeauthoryear{Scholte, Balzarotti, and Kirda}{Scholte
  et~al\mbox{.}}{2011}]%
        {scholte2011quo}
\bibfield{author}{\bibinfo{person}{Theodoor Scholte}, \bibinfo{person}{Davide
  Balzarotti}, {and} \bibinfo{person}{Engin Kirda}.}
  \bibinfo{year}{2011}\natexlab{}.
\newblock \showarticletitle{Quo Vadis? {A} Study of the Evolution of Input
  Validation Vulnerabilities in Web Applications}. In
  \bibinfo{booktitle}{\emph{Financial Cryptography}}.
\newblock


\bibitem[\protect\citeauthoryear{Sellars}{Sellars}{2018}]%
        {sellars2018twenty}
\bibfield{author}{\bibinfo{person}{Andrew Sellars}.}
  \bibinfo{year}{2018}\natexlab{}.
\newblock \showarticletitle{Twenty Years of Web Scraping and the Computer Fraud
  and Abuse Act}.
\newblock \bibinfo{journal}{\emph{BUJ Sci. \& Tech. L.}}  \bibinfo{volume}{24}
  (\bibinfo{year}{2018}), \bibinfo{pages}{372}.
\newblock


\bibitem[\protect\citeauthoryear{Steinfeld, Bull, and Zheng}{Steinfeld
  et~al\mbox{.}}{2001}]%
        {steinfeld2001content}
\bibfield{author}{\bibinfo{person}{Ron Steinfeld}, \bibinfo{person}{Laurence
  Bull}, {and} \bibinfo{person}{Yuliang Zheng}.}
  \bibinfo{year}{2001}\natexlab{}.
\newblock \showarticletitle{Content extraction signatures}. In
  \bibinfo{booktitle}{\emph{International Conference on Information Security
  and Cryptology}}.
\newblock


\bibitem[\protect\citeauthoryear{Sun, Edmundson, Vanbever, Li, Rexford, Chiang,
  and Mittal}{Sun et~al\mbox{.}}{2015}]%
        {bgp_raptor}
\bibfield{author}{\bibinfo{person}{Yixin Sun}, \bibinfo{person}{Anne
  Edmundson}, \bibinfo{person}{Laurent Vanbever}, \bibinfo{person}{Oscar Li},
  \bibinfo{person}{Jennifer Rexford}, \bibinfo{person}{Mung Chiang}, {and}
  \bibinfo{person}{Prateek Mittal}.} \bibinfo{year}{2015}\natexlab{}.
\newblock \showarticletitle{{RAPTOR:} Routing Attacks on Privacy in Tor}. In
  \bibinfo{booktitle}{\emph{{USENIX} Security}}.
\newblock


\bibitem[\protect\citeauthoryear{Useem}{Useem}{2017}]%
        {useem_2017}
\bibfield{author}{\bibinfo{person}{Jerry Useem}.}
  \bibinfo{year}{2017}\natexlab{}.
\newblock \showarticletitle{How Online Shopping Makes Suckers of Us All}.
\newblock \bibinfo{journal}{\emph{The Atlantic}} (\bibinfo{date}{Jul}
  \bibinfo{year}{2017}).
\newblock
\newblock
\shownote{\url{https://www.theatlantic.com/magazine/archive/2017/05/how-online-shopping-makes-suckers-of-us-all/521448/}.}


\bibitem[\protect\citeauthoryear{Wang, Asharov, Pass, Ristenpart, and
  Shelat}{Wang et~al\mbox{.}}{2019}]%
        {blindca}
\bibfield{author}{\bibinfo{person}{L. Wang}, \bibinfo{person}{G. Asharov},
  \bibinfo{person}{R. Pass}, \bibinfo{person}{T. Ristenpart}, {and}
  \bibinfo{person}{A. Shelat}.} \bibinfo{year}{2019}\natexlab{}.
\newblock \showarticletitle{Blind Certificate Authorities}. In
  \bibinfo{booktitle}{\emph{IEEE S\&P}}.
\newblock


\bibitem[\protect\citeauthoryear{Wang, Malozemoff, and Katz}{Wang
  et~al\mbox{.}}{2016}]%
        {emp-toolkit}
\bibfield{author}{\bibinfo{person}{Xiao Wang}, \bibinfo{person}{Alex~J.
  Malozemoff}, {and} \bibinfo{person}{Jonathan Katz}.}
  \bibinfo{year}{2016}\natexlab{}.
\newblock \bibinfo{title}{{EMP-toolkit: Efficient MultiParty computation
  toolkit}}.
\newblock \bibinfo{howpublished}{\url{https://github.com/emp-toolkit}}.
\newblock


\bibitem[\protect\citeauthoryear{Wang, Ranellucci, and Katz}{Wang
  et~al\mbox{.}}{2017}]%
        {DBLP:conf/ccs/WangRK17}
\bibfield{author}{\bibinfo{person}{Xiao Wang}, \bibinfo{person}{Samuel
  Ranellucci}, {and} \bibinfo{person}{Jonathan Katz}.}
  \bibinfo{year}{2017}\natexlab{}.
\newblock \showarticletitle{Authenticated Garbling and Efficient Maliciously
  Secure Two-Party Computation}. In \bibinfo{booktitle}{\emph{ACM CCS}}.
\newblock


\bibitem[\protect\citeauthoryear{Yao}{Yao}{1982}]%
        {yao1982protocols}
\bibfield{author}{\bibinfo{person}{Andrew Chi-Chih Yao}.}
  \bibinfo{year}{1982}\natexlab{}.
\newblock \showarticletitle{Protocols for secure computations}. In
  \bibinfo{booktitle}{\emph{FOCS}}.
\newblock


\bibitem[\protect\citeauthoryear{Yasskin}{Yasskin}{2019}]%
        {http-origin-signed-responses}
\bibfield{author}{\bibinfo{person}{Jeffrey Yasskin}.}
  \bibinfo{year}{2019}\natexlab{}.
\newblock \bibinfo{booktitle}{\emph{Signed {HTTP} Exchanges}}.
\newblock \bibinfo{type}{Internet-Draft}
  draft-yasskin-http-origin-signed-responses-05.
\newblock


\bibitem[\protect\citeauthoryear{Zhang, Cecchetti, Croman, Juels, and
  Shi}{Zhang et~al\mbox{.}}{2016}]%
        {zhang2016town}
\bibfield{author}{\bibinfo{person}{Fan Zhang}, \bibinfo{person}{Ethan
  Cecchetti}, \bibinfo{person}{Kyle Croman}, \bibinfo{person}{Ari Juels}, {and}
  \bibinfo{person}{Elaine Shi}.} \bibinfo{year}{2016}\natexlab{}.
\newblock \showarticletitle{{T}own {C}rier: An authenticated data feed for
  smart contracts}. In \bibinfo{booktitle}{\emph{ACM CCS}}.
\newblock


\bibitem[\protect\citeauthoryear{Zhang, Rexford, and Feigenbaum}{Zhang
  et~al\mbox{.}}{2005}]%
        {zhang2005learning}
\bibfield{author}{\bibinfo{person}{Jian Zhang}, \bibinfo{person}{Jennifer
  Rexford}, {and} \bibinfo{person}{Joan Feigenbaum}.}
  \bibinfo{year}{2005}\natexlab{}.
\newblock \showarticletitle{Learning-based anomaly detection in BGP updates}.
  In \bibinfo{booktitle}{\emph{Proceedings of the 2005 ACM SIGCOMM workshop on
  Mining network data}}. \bibinfo{pages}{219--220}.
\newblock


\bibitem[\protect\citeauthoryear{Zhang, Yen, Zhao, Massey, Wu, and Zhang}{Zhang
  et~al\mbox{.}}{2004}]%
        {zhang2004detection}
\bibfield{author}{\bibinfo{person}{Ke Zhang}, \bibinfo{person}{Amy Yen},
  \bibinfo{person}{Xiaoliang Zhao}, \bibinfo{person}{Dan Massey},
  \bibinfo{person}{S~Felix Wu}, {and} \bibinfo{person}{Lixia Zhang}.}
  \bibinfo{year}{2004}\natexlab{}.
\newblock \showarticletitle{On detection of anomalous routing dynamics in BGP}.
  In \bibinfo{booktitle}{\emph{International Conference on Research in
  Networking}}. Springer, \bibinfo{pages}{259--270}.
\newblock


\end{thebibliography}

\appendix

\section{Protocols details}
\label{sec:ectf}

\subsection{Formal specification}
\label{sec:formal protocols}

We gave a self-contained informal description of the three-party handshake protocol in~\cref{sec:handshake}.
The formal specification is given in~\cref{fig:3party-hs-protocol} along with its building block~$\protadd$ in~\cref{fig:ecadd}.
The post-handshake protocols for CBC-HMAC described in~\cref{sec:post handshake} is specified in~\cref{fig:2pc-hmac}.

\begin{figure}[t]
\begin{boxedminipage}{\columnwidth}
\protocolNoBox{The three-party handshake (3P-HS) protocol among $\prover$, $\verifier$ and $\tlsserver$}{
{\bf Public information:} Let $EC$ be the Elliptic Curve used in ECDHE over $\FF_p$ with order $p$, $G$ a parameter, and $Y_{\tlsserver}$ the server public key. \\
{\bf Output:} $\prover$ and $\verifier$ output $\mackeyP$ and $\mackeyV$ respectively, while the TLS server outputs $\mackey=\mackeyP + \mackeyV$. Besides,  both $\tlsserver$ and $\prover$ outputs $\enckey$.
\\[1mm][\hline]
\\[1mm]
\textbf{TLS server $\tlsserver$:} follow the standard TLS protocol. \\[1mm]
\textbf{Prover $\prover$:} \\
\oninit: $\prover$ samples $\randc \sample \bin^{256}$ and sends ClientHello($\randc$) to $\tlsserver$ to start a standard TLS handshake. \\
\onrecv ServerHello($\rands$), ServerKeyEx($Y,\sigma,\CERT$) from $\tlsserver$: \\
\begin{itemize}[leftmargin=*]
\item $\prover$ verifies that $\CERT$ is a valid certificate and that $\sigma$ is a valid signature over $(\randc, \rands,\pubkeyS)$ signed by a key contained in $\CERT$. $\prover$ sends $(\randc, \rands, \pubkeyS, \sigma, \CERT)$ to $\verifier$.
\item $\verifier$ checks $\CERT$ and $\sigma$ similarly. $\verifier$ then samples $\secretV \sample \FF_p$ and computes $Y_V =\secretV \cdot G$. Send $Y_V$ to $\prover$.
\item $\prover$ samples $\secretP \sample \FF_p$ and computes $Y_P=\secretP \cdot G$. Send ClientKeyEx($Y_P + Y_V$) to $\tlsserver$.
\item $\prover$ and $\verifier$ compute $Z_P = \secretP \cdot \pubkeyS$ and $Z_V = \secretV \cdot \pubkeyS$ respectively. They run $\protadd$ to compute a sharing of the $x$-coordinate of $Z=Z_P + Z_V$, denoted $z_P, z_V$.
\item $\prover$ and $\verifier$ send $z_P$ (and $z_V$) to $\idealtwopc^\text{hs}$ (specified below) to compute shares of session keys and the master secret. $\prover$ receives $(\enckey, \mackeyP, m_{\prover})$, while $\verifier$ receives $(\mackeyV, m_{\verifier})$.
\item $\prover$ computes a hash (denoted $h$) of the handshake messages sent and received thus far, and runs 2PC-PRF with $\verifier$ to compute $s=\prf(m_{\prover}\oplus m_{\verifier}, \text{``client finished''}, h)$ on the hash of the handshake messages and send a Finished($s$) to $\tlsserver$.
\end{itemize}
\onrecv other messages from $\tlsserver$: \\
\begin{itemize}[leftmargin=*]
    \item If it's Finished($s$), $\prover$ and $\verifier$ run a 2PC to check $s \stackrel{?}{=}\prf(m_{\prover}\oplus m_{\verifier}, \text{``server finished''}, h)$ and abort if not.
    \item Otherwise respond according to the standard TLS protocol.
\end{itemize}
}

\protocol{$\idealtwopc^\text{hs}$ with $\prover$ and $\verifier$}{
\textbf{Public Input:} nonce $\randc, \rands$ \\
\textbf{Private Input:} $z_P \in \FF_p$ from $\prover$; $z_V \in \FF_p^2$ from $\verifier$\\
\begin{itemize}
    \item $z:=z_P + z_v$
    \item $m:=\prf(z, \text{``master secret''}, \randc \| \rands)$ (truncate at $48$ bytes)
    \item $\mackey,\enckey:=\prf(m, \text{``key expansion''}, \rands \| \randc)$ \pccomment{key expansion}
    \item Sample $r_k, r_m \sample \FF_p$. Send $(\enckey, r_k, r_m)$ to $\prover$, and $(r_k \oplus \mackey, r_m \oplus m)$ to $\verifier$ privately.
\end{itemize}
}
\end{boxedminipage}
\caption{The protocol of three-party handshake.}
\label{fig:3party-hs-protocol}
\end{figure}

\begin{figure}
\centering
\protocol{$\protadd$ between $\prover$ and $\verifier$}
{
\textbf{Input}: $P_1=(x_1, y_1)\in EC(\FF_p)$ from $\prover$, $P_2=(x_2, y_2)\in EC(\FF_p)$ from $\verifier$. \\
\textbf{Output}: $\prover$ and $\verifier$ output $s_1$ and $s_2$ such that $s_1 + s_2=x$ where $(x,y)=P_1+P_2$ in $EC$.\\
\textbf{Protocol}:
\begin{itemize}
\item $\prover$ (and $\verifier$) sample $\rho_i \sample \ZZ_p$ for $i\in \set{1,2}$ respectively. $\prover$ and $\verifier$ run $\alpha_1,\alpha_2 :=\MTA((-x_1, \rho_1), (\rho_2, x_2))$.
\item $\prover$ computes $\delta_1 = -x_1 \rho_1 + \alpha_1$ and $\verifier$ computes $\delta_2 = x_2 \rho_2 + \alpha_2$.
\item $\prover$ (and $\verifier$) reveal $\delta_1$ (and $\delta_2$) to each other and compute $\delta=\delta_1 + \delta_2$.
\item $\prover$ (and $\verifier$) compute $\eta_i = \rho_i \cdot \delta^{-1}$ for $i\in \set{1,2}$ respectively.
\item $\prover$ and $\verifier$ run $\beta_1, \beta_2:=\MTA((-y_1, \eta_1), (\eta_2, y_2))$.
\item $\prover$ computes $\lambda_1 = -y_1 \cdot \eta_1 + \beta_1$ and $\verifier$ computes $\lambda_2 = y_2 \cdot \eta_2 + \beta_2$. They run $\gamma_1, \gamma_2 := \MTA(\lambda_1, \lambda_2)$.
\item $\prover$ (and $\verifier$) computes $s_i = 2\gamma_i + \lambda_i ^ 2 - x_i$ for $i\in\set{1,2}$ respectively.
\item $\prover$ outputs $s_1$ and $\verifier$ outputs $s_2$.
\end{itemize}
}
\caption{($\protadd$) A protocol for converting shares of EC points in $EC(\FF)$ to shares of coordinates in $\FF$.}
\label{fig:ecadd}
\end{figure}

\begin{figure}[t]
\centering
\protocol{2PC-HMAC between $\prover$ and $\verifier$}{%
\textbf{Input}: $\prover$ inputs $\mackeyP$, $m$ and $\verifier$ inputs $\mackeyV$. \\
\textbf{Output}: $\prover$ outputs $\hmac(\mackey, m)$ where $\mackey = \mackeyP \xor \mackeyV$. \\[1mm]
\textbf{One-time setup}: $\prover$ and $\verifier$ use 2PC to compute $\ipadhash = f(\initstate, \mackey \xor \ipad)$ and reveal $\ipadhash$ to $\prover$. \\
\textbf{To compute a tag for message $m$}:
\begin{itemize}
    \item $\prover$ computes inner hash $h_i = f(\ipadhash, m)$.
    \item $\prover$ inputs $\mackeyP$, $h_i$ and $\verifier$ inputs $\mackeyV$ to 2PC which reveals $\hash(\mackey \xor \opad \| h_i)$ to both parties.
\end{itemize}}
\caption{The 2PC-HMAC protocol. $f$ denotes the compression function of the hash function $\hash$ and $\initstate$ denotes the initial value.}\label{fig:2pc-hmac}
\end{figure}

\subsection{Selective opening (CBC-HMAC)}
\label{sec: cbc hmac tricks}

\boldhead{Redacting a suffix}
When a suffix $\VEC{B}_{i+}$ is to be redacted, $\prover$ computes $\pi = \ZKP{ \VEC{B}_{i+}, \enckey: f(s_{i}, \VEC{B}_{i+}) = ih \land H(\mackey \xor \opad || ih) = \sigma \land B_{1025}\| B_{1026} \| B_{1027} = \cbc(\enckey, \sigma) } $  and $s_{i}$ is the state after applying $f$ on $\VEC{B}_{i-} \| B_i$.  $\prover$ sends $(\pi, \VEC{B}_{i-} \| B_i)$ to $\verifier$. The verifier then 1) checks $s_{i-1}$ by applying $f$ on $\VEC{B}_{i-} \| B_i$, and 2) verifies $\pi$. Essentially, the security of this follows from pre-image resistance of $f$. Moreover, $\verifier$ doesn't learn the redacted suffix since $ih=f(s, \VEC{B}_{i+})$ is kept secret from $\verifier$. The total cost is $3$ AES and $256-i$ SHA-2 hashes in ZKP.

\boldhead{Redacting a prefix}
$\prover$ computes two ZKPs: 1) $\pi_1 = \ZKP{ \VEC{B}_{i-}, \allowbreak \mackey: H(\mackey \xor \ipad || \VEC{B}_{i-}) = s_{i-1}} $; 2) $\pi_2 = \ZKP{ \mackey, \enckey: H(\mackey \xor \opad || ih) = \sigma \land B_{1025}\| B_{1026} \| B_{1027} = \cbc(\enckey, \sigma) }$.
$\prover$ sends $(\pi_1, \pi_2, \allowbreak s_{i-1}, B_i \| \VEC{B}_{i+} )$ to $\verifier$. The verifier checks that 1) $s_{i-1}$ is correct using $\pi_1$ and then computes $f(s_{i-1}, B_i \| \VEC{B}_{i+})$ to obtain the inner hash $ih$, 2) $\pi_2$ is verified using the computed $ih$. The cost incurred is $3$ AES and $256-i$ SHA-2 hashes in ZKP.

Note that redacting a prefix/suffix only makes sense if the revealed portion does not contain any private user data. Otherwise, $\prover$ would have to find the smallest substring containing all the sensitive blocks and redact either the prefix/suffix similar to above.

\section{Protocols details for GCM}
\label{app:gcm}

\subsection{Preliminaries}

GCM is an authenticated encryption with additional data (AEAD) cipher.
To encrypt, the GCM cipher takes as inputs a tuple $(\key, IV, \VEC{M}, \VEC{A})$: a secret key, an initial vector, a plaintext of multiple AES blocks, and additional data to be included in the integrity protection; it outputs a ciphertext $\VEC{C}$ and a tag $T$.
Decryption reverses the process. The decryption cipher takes as input $(\key, IV, \VEC{C}, \VEC{A}, T)$ and first checks the integrity of the ciphertext by comparing a recomputed tag with $T$, then outputs the plaintext.

The ciphertext is computed in the counter mode: $C_i = \aes(\key,\allowbreak \INC^i(IV)) \xor M_i$ where $\INC^i$ denotes incrementing $IV$ for $i$ times (the exact format of $\INC$ is immaterial.)

The tag $\TAG(\key, IV, \VEC{C}, \VEC{A})$ is computed as follows.
Given a vector $\VEC{X} \in \gfgcm^m$, the associated GHASH polynomial $P_{\VEC{X}}: \gfgcm \to \gfgcm$ is defined as $P_{\VEC{X}}(h)=\sum_{i=1}^m X_i \cdot h^{m-i+1}$ with addition and multiplication done in $\gfgcm$.
Without loss of generality, suppose $\VEC{A}$ and $\VEC{C}$ are properly padded. Let $\ell_A$ and $\ell_C$ denote their length. A GCM tag is
\begin{equation}
\TAG(\key, IV, \VEC{C}, \VEC{A}):=\aes(\key, IV) \xor \GHASHPOLY(h)
\label{eqn:gcmtag}
\end{equation} where $h=\aes(\key, \VEC{0})$.

When GCM is used in TLS, each plaintext record $D$ is encrypted as follows. A unique nonce $n$ is chosen and the additional data $\kappa$ is computed as a concatenation of the sequence number, version, and length of $D$. GCM encryption is invoked to generate the payload record as $M = n \| \gcm(\key, n, D, \kappa)$. We refer readers to~\cite{dworkin2007sp} for a complete specification.

\subsection{\Posthandshake}
\label{app:gcm posths}

The 2PC protocols for verifying tags and decrypting records are specified in~\cref{fig:post-hs-gcm}.

\begin{figure}[t]
\begin{boxedminipage}{\columnwidth}
\protocolNoBox{Post-handshake protocols for GCM}{
{\bf Private input:} $\key_{\prover}$ and $\key_{\verifier}$ from $\prover$ and $\verifier$ respectively. $\key=\key_{\prover}+\key_{\verifier}$ is the encryption key. %
}
\protocolNoBox{Protocol for preprocessing}{
\oninit: $\prover$ (and $\verifier$) sends $\key_{\prover}$ and ($\key_{\verifier}$) to $\idealPPGCM$ and wait for output $\set{h_{\prover,i}}_i$ (and $\set{h_{\verifier,i}}_i$). \\
}%
\protocol{$\idealPPGCM$}{
After receiving $\key_1,\key_2$ from two parties, compute $h:=\aes(\key_1 + \key_2, \VEC{0})$. Sample $n$ random numbers $\set{r_i}_{i=1}^n$ and compute $\set{h^i}_{i=1}^n$ in $\gfgcm$. For $i\in [n]$, send $r_i$ to player $1$ and $r_i \oplus h^i$ to player $2$.
} %
\protocolNoBox{Protocol for decrypting TLS records}{
\textbf{Prover $\prover$:} \\
\onrecv a record $(IV, \VEC{C}, \VEC{A}, T)$ from $\tlsserver$: \\
\begin{itemize}
\item Let $\VEC{X}=\VEC{A} \| \VEC{C}  \| \ell_A \| \ell_C$.
\item Send $(\key_{\prover}, IV)$ to $\idealAesSameMsg$ and wait for output $c_{\prover}$.
\item Send $(IV,\VEC{X})$ to $\verifier$ and wait for the response $P$.
\item Compute $T' = P + c_{\prover} + \sum_i X_i \cdot h_{\prover,i}$ in $\gfgcm$.
\item Abort if $T' \neq T$. Otherwise, compute $\VEC{K}$ such that $K_i=\INC^i(IV)$ for $i\in[\ell_C]$. Send ($IV,\ell_C$, Decrypt) to $\verifier$.
\item Send $(\key_{\prover}, \VEC{K})$ to $\idealAesSameMsgOneOutput$ as party $1$ and wait for output $\VEC{K'}$.
\item Decrypt the message as $M_i = K'_i \xor C_i$.
\end{itemize} \\[1mm]
\textbf{Verifier $\verifier$:} \\
\onrecv ($IV, \VEC{X}$) from $\prover$: \\
\begin{itemize}
\item If $IV$ found in store, abort. Otherwise store $IV$ and proceed.
\item Send $(\key_{\verifier}, IV)$ to $\idealAesSameMsg$ and wait for output $c_{\verifier}$.
\item Compute $P=c_{\verifier} + \sum_i X_i \cdot h_{\verifier,i}$ in $\gfgcm$
\item Send $P$ to $\prover$.
\end{itemize}
\onrecv ($IV, n$, Decrypt) from $\prover$: \\
\begin{itemize}
    \item Compute $\VEC{K}$ such that $K_i=\INC^i(IV)$ for $i\in [n]$.
    \item Abort if any $K_i$ is found in store (as previously used IVs.)
    \item Send $(\key_{\verifier}, \VEC{K})$ to $\idealAesSameMsgOneOutput$ as party $2$.
\end{itemize}
}%
\protocol{$\idealAesSameMsg$}{
Wait for input $(\key_i, m_i)$ from party $i$ for $i\in \set{1,2}$. Abort if $m_1 \neq m_2$. Sample $r \sample \FF$. Compute $c=\aes(\key_1 \xor \key_2, m_1)$. Send $r$ to party $1$ and $c\xor r$ to party $2$.}
\protocol{$\idealAesSameMsgOneOutput$}{%
Wait for input $(\key_i, m_i)$ from party $i$ for $i\in \set{1,2}$. Abort if $m_1 \neq m_2$. Compute $c=\aes(\key_1 \xor \key_2, m_1)$. Send $c$ to party $1$ and $\bot$ to party $2$.
}
\end{boxedminipage}

\caption{The post-handshake protocols for AES-GCM.}
\label{fig:post-hs-gcm}
\end{figure}

\boldhead{Tag creation/verification}
Computing or verifying a GCM tag involves evaluating \cref{eqn:gcmtag} in 2PC.
A challenge is that \cref{eqn:gcmtag} involves both arithmetic computation (e.g., polynomial evaluation in $\gfgcm$) as well as binary computation (e.g., AES). Performing multiplication in a large field in a binary circuit is expensive, while computing AES (defined in $\gf(2^8)$) in $\gfgcm$ incurs high overhead. Even if the computation could somehow separated into two circuits, evaluating the polynomial alone---which takes approximately 1,000 multiplications in $\gfgcm$ for \emph{each} record---would be unduly expensive.

Our protocol removes the need for polynomial evaluation. The actual 2PC protocol involves only binary operations and thus can be done in a single circuit. Moreover, the per-record computation is reduced to only one invocation of 2PC-AES.

The idea is to compute shares of $\set{h^i}$ (in a 2PC protocol) in a preprocessing phase at the beginning of a session. The overhead of preprocessing is amortized over the session because the same $h$ used for all records that follow. With shares of $\set{h^i}$, $\prover$ and $\verifier$  can compute shares of a polynomial evaluation $\GHASHPOLY(h)$ locally. They also compute $\aes(\key, IV)$ in 2PC to get a share of $\TAG(\key, IV, \VEC{C}, \VEC{A})$. In total, only one invocation of 2PC-AES in needed to check the tag for each record.

It is critical that $\verifier$ never responds to the same $IV$ more than once; otherwise $\prover$ would learn $h$. Specifically, in each response, $\verifier$ reveals a blinded linear combination of her shares $\set{h_{\verifier, i}}$ in the form of $\mathcal{L}_{IV,X}=\aes(\key, IV) \oplus \sum_i X_i \cdot  h_{\verifier, i}$. It is important that the value is blinded by $\aes(\key, IV)$ because a single unblinded linear combination of $\set {h_{\verifier, i}}$ would allow $\prover$ to solve for $h$. Therefore, if $\verifier$ responds to the same $IV$ twice, the blinding can be removed by adding the two responses (in $\gfgcm$): $\mathcal{L}_{IV,X}\xor \mathcal{L}_{IV,X'}=\sum_i (X_i + X'_i)\cdot h_{\verifier, i}$. This follows from the nonce uniqueness requirement of GCM~\cite{gcm-in-tls}.

\boldhead{Encrypting/decrypting records}
Once tags are properly checked, decryption of records is straightforward. $\prover$ and $\verifier$ simply compute AES encryption of $\INC^i(IV)$ with 2PC-AES.
A subtlety to note is that $\verifier$ must check that the counters to be encrypted have \emph{not} been used as $IV$ previously. Otherwise $\prover$ would learn $h$ to $\prover$ in a manner like that outlined above.

\subsection{Proof Generation}
\label{sec:proof gen gcm}

\boldhead{Revealing a block}
$\prover$ wants to convince $\verifier$ that an AES block $B_i$ is the $i$th block in the encrypted record $\encrecord$. The proof strategy is as follows: 1) prove that AES block $B_i$ encrypts to the ciphertext block $\hat B_i$ and 2) prove that the tag is correct. Proving the correct encryption requires only $1$ AES in ZKP. Na\"ively done, proving the correct tag incurs evaluating the GHASH polynomial of degree $512$ and $2$ AES block encryptions in ZKP.

We manage to achieve a much more efficient proof by allowing $\prover$ to reveal two encrypted messages $\aes(\key, IV)$ and $\aes(\key, 0)$ to $\verifier$, thus allowing $\verifier$ to verify the tag (see~\cref{eqn:gcmtag}). $\prover$ only needs to prove the correctness of encryption in ZK and that the key used corresponds to the commitment, requiring $2$ AES and $1$ SHA-2 ($\prover$ commits to $\key_{\prover}$ by revealing a hash of the key). Thus, the total cost is $3$ AES and $1$ SHA-2 in ZKP.

\boldhead{Revealing a TLS record}
The proof techniques are a simple extension from the above case. $\prover$ reveals the entire record $\record$ and proves correct AES encryption of all the AES blocks, resulting in a total $514$ AES and $1$ SHA-2 in ZKP.

\boldhead{Revealing a TLS record except for a block}
Similar to the above case, $\prover$ proves encryption of all the blocks in the record except one, resulting in a total $513$ AES and $1$ SHA-2 in ZKP. %
\section{Protocol extensions}
\label{app:ext}

\subsection{Adapting to support TLS 1.3}
\label{app:tls13details}

To support TLS 1.3, the 3P-HS protocol must be adapted to a new handshake flow and a different key derivation circuit.
Notably, all handshake messages after the ServerHello are now {\em encrypted}. A naïve strategy would be to decrypt them in 2PC, which would be costly as certificates are usually large.
However, thanks to the key independence property of TLS 1.3~\cite{tls13handshake}, $\prover$ and $\verifier$ can securely reveal the handshake encryption keys without affecting the secrecy of final session keys~\cite{tls13handshake}.
Handshake integrity is preserved because the Finished message authenticates the handshake using yet another independent key.
(In fact~\cite[\S3.1]{tls13handshake} argues that the signatures already authenticate the handshake.)

Therefore the optimized 3P-HS work as follows. $\prover$ and $\verifier$ perform ECDHE the same as before. Then they derive handshake and application keys by executing 2PC-HKDF, and reveal the handshake keys to $\prover$, allowing $\prover$ to decrypt handshake messages locally (i.e., without 2PC). The 2PC circuit involves roughly 30 invocations of SHA-256, totaling to approximately 70k AND gates, comparable to that for TLS 1.2.
Finally, since CBC-HMAC is not supported by TLS 1.3, \systemname can only be used in GCM mode.

\subsection{Query construction is optional}

For applications that bind responses to queries, e.g., when a stock ticker is included with the quote, 2PC query construction protocols can be avoided altogether. Since TLS uses separate keys for each direction of communication, client-to-server keys can be revealed to $\prover$ after the handshake so that $\prover$ can query the server without interacting with $\verifier$.

\subsection{Supporting multi-round sessions}
\systemname can be extended to support multi-round sessions where $\prover$ sends further queries depending on previous responses. After each round, $\prover$ executes similar 2PC protocols as above to verify MAC tags of incoming responses, since MAC verification and creation is symmetric. However an additional commitment is required to prevent prevent $\prover$ from abusing MAC verification to forge tags.

In TLS, different MAC keys are used for server-to-client and client-to-server communication. To support multi-round sessions, $\prover$ and $\verifier$ run 2PC to verify tags for former, and create tags on fresh messages for latter. We've specified the protocols to create (and verify) MAC tags. Now we discuss additional security considerations for multi-round sessions.

When checking tags for server-to-client messages, we must ensure that $\prover$ cannot forge tags on messages that are not originally from the server. Suppose $\prover$ wishes to verify a tag $T$ on message $M$. The idea is to have $\prover$ first commit to $T$, then $\prover$ and $\verifier$ run a 2PC protocol to compute a tag $T'$ on message $M$. $\prover$ is asked to open the commitment to $\verifier$ and if $T\neq T'$, $\verifier$ aborts the protocol. Since $\prover$ doesn't know the MAC key, $\prover$ cannot compute and commit to a tag on a message that is not from the server.

When creating tags for client-to-server messages, $\verifier$ makes sure MAC tags are created on messages with increasing sequence numbers, as required by TLS.
This also prevents a malicious $\prover$ from creating two messages with the same sequence number, because there is no way for $\verifier$ to distinguish which one was sent to the server.

\subsection{An alternative \systemname protocol: Proxy mode}
\label{sec:alternative protocol}

As shown in~\cref{tab:2pc_costs},
the HMAC mode of \systemname is highly efficient and the runtime of creating and verifying HMAC tags in 2PC is independent of record size (cf. ~\cref{fig:2pc-hmac}).
The GCM mode is efficient for small requests with preprocessing, but can be expensive for large records. \ari{This seems too strong to me. We've shown that it can work reasonably well over a WAN, and with preprocessing, it definitely can.}\fanz{agree. revised.} We now present a highly efficient alternative that avoids post-handshake 2PC protocols altogether.

The idea is to have the verifier $\verifier$ act as a proxy between the prover $\prover$ and the TLS server $\tlsserver$, i.e., $\prover$ sends/receives messages to/from $\tlsserver$ through $\verifier$.
The modified flow of the \systemname protocol is as follows: after the three-party handshake, $\prover$ commits to her key share $\keyP$ then $\verifier$ reveals $\keyV$ to $\prover$. Therefore $\prover$ now has the entire session key $\key = \keyP+\keyV$. As $\prover$ uses $\key$ to continue the session with the server, $\verifier$ records the proxy traffic. After the session concludes, $\prover$ proves statements about the recorded session the same as before.

It's worth emphasizing that the three-party handshake is required for unforgeability. Unlike CBC-HMAC, GCM is not committing~\cite{DBLP:conf/crypto/GrubbsLR17}: for a given ciphertext and tag $(C,T)$ encrypted with key $\key$, one can find $\key'\neq \key$ that decrypts $C$ to a different plaintext while computing the same tag, as GCM MAC is not collision-resistant.
To prevent such attacks, the above protocol requires $\prover$ to commit to her key share before learning the session key.

\boldhead{Security properties and network assumptions}
The verifier-integrity and privacy properties are clear, as a malicious $\verifier$ cannot break the integrity and privacy of TLS (by assumption).

For prover integrity, though, we need to assume that the proxy can reliably connect to  $\tlsserver$ throughout the session.
First, we assume the proxy can ascertain that it indeed is connected with $\tlsserver$. Moreover, we assume messages sent between the proxy and $\tlsserver$ cannot be tampered with by $\prover$, who knows the session keys and thus could modify the session content.

Note that during the three-party handshake, $\verifier$ can ascertain the server's identity by checking the server's signature over a fresh nonce (in standard TLS). After the handshake, however, $\verifier$ has to rely on network-layer indicators, such as IP addresses.
In practice, $\verifier$ must therefore have correct, up-to-date DNS records, and that the network between $\verifier$ and the server (e.g., their ISP and the backbone network) must be properly secured against traffic injection, e.g., throught BGP attacks~\cite{bgp_raptor}. (Eavesdropping isn't problematic.)

\ari{In the following paragraph, I'd not couch things as if we've implemented them, and I'd just express this as an opportunity for future work. We can just note that: (1) BGP attacks are challenging to mount; (2) We can make them harder to mount by distributing verifier nodes geographically; (3) Various detection techniques have been proposed (cite them) that might be deployed by verifiers: (4) Often BGP attacks are documented after the fact (cite), and DECO can be enhanced to support credential revocation for affected sessions.}\fanz{changes here}
These assumptions have been embraced by other systems in a similar proxy setting (e.g.,~\cite{blindca}), as BGP attacks are challenging to mount in practice. We can further enhance our protocol against traffic interception by distributing verifiers nodes geographically.
Moreover, various detection techniques have been proposed~\cite{butler2010bgpsurvey,kruegel2003topology,zhang2005learning,deshpande2009online,zhang2004detection,bgpmon,ThousandEyes} that can be deployed by verifiers.
Often BGP attacks are documented after the fact (e.g., see~\cite{bgpstream}), therefore, when applicable, applications of \systemname can be enhanced to support revocation of affected sessions (for example, when \systemname is used to issue credentials in an identity system such as~\cite{candid}.)
We leave further exploration as future work.

This alternative protocol represents a different performance-security tradeoff. It's highly efficient because no intensive cryptography occurs after the handshake, but it requires additional assumptions about the network and therefore only withstands a weaker network adversary. %
\section{Security proofs}
\label{sec:securityproofs}
\label{sec:securityproofsUC}

\newcommand{\env}{\mathcal{Z}}

Recall Theorem \ref{thm:mainUC}. %
We now prove that the protocol in~\cref{fig:fullprotocolUC} securely realizes $\idealOracle$. Specifically, we show that for any real-world adversary $\adv$, we can construct an ideal world simulator $\simulator$, such that for all environments $\env$, the ideal execution with $\simulator$ is indistinguishable from the real execution with $\adv$. We refer readers to~\cite{simulation,uc} for simulation-based proof techniques.
\begin{proof}
Recall that we assume $\tlsserver$ is honest throughout the protocol. Hence, we only consider cases where $\adv$ maliciously corrupts either $\prover$ or $\verifier$. This means that we only need to construct ideal-world simulators for the views of $\prover$ and $\verifier$. %

\boldhead{Malicious $\prover$}
We wish to show the prover-integrity guarantee. Basically, if $\verifier$ receives $(b, \tlsserver)$, then $\prover$ must have input some $\PRI$ such that $\tlsserver(\query(\PRI))=R$ and $b=\PRED(R)$.

Given a real-world PPT adversary $\adv$, $\simulator$ proceeds as follows:

\begin{enumerate}[leftmargin=*]
\item $\simulator$ runs $\adv$, $\idealZK$ and $\idealTwoPC$ internally. $\simulator$ forwards any input $z$ from $\env$ to $\adv$ and records the traffic going to and from $\adv$.%
\item Upon request from $\adv$, $\simulator$ runs 3P-HS as $\verifier$ (using $\idealTwoPC$ as a sub-routine). During 3P-HS, when $\adv$ outputs a message $m$ intended for $\tlsserver$, $\simulator$ forwards it to $\idealOracle$ as $(\sid, \tlsserver, m)$ and forwards $(\sid, m)$ to $\adv$ if it receives any messages from $\idealOracle$. By the end, $\simulator$ learns $Y_P, s_V, \mackeyV$.
\item Upon request from $\adv$, $\simulator$ runs 2PC-HMAC as $\verifier$, using $\mackeyV$ as input. Again, $\simulator$ uses $\idealTwoPC$ as a sub-routine to run 2PC-HMAC and forwards messages to $\tlsserver$ as above and forwards the response from $\tlsserver$ to $\adv$. $\simulator$ records the messages between $\adv$ and $\tlsserver$ during this stage in $(\hat Q,\hat R)$. Note that these are ciphertext records.
\item When $\adv$ sends $(\sid, \hat Q, \hat R, \mackeyP)$, reply with $(\sid, \mackeyV)$.
\item\label{item:checkStep} Upon receiving $(\sid, \msgprove, x, w)$ (with $x=(\enckey, \PRI, Q, R)$ and $w=(\hat Q, \hat R, \mackey, b)$) from $\adv$, $\simulator$ checks that
\begin{align*}
    \hat Q &= \mathsf{CBC\_HMAC}(\enckey, \mackey, Q) \\
    \hat R &= \mathsf{CBC\_HMAC}(\enckey, \mackey, R) \\
    Q &= \query(\PRI)\,.
\end{align*}

\item\label{item:sendStep} If all of the above checks passed, $\simulator$ sends $\PRI$ to $\idealOracle$ and instructs $\idealOracle$ to send the output to $\verifier$. $\simulator$ outputs whatever $\adv$ outputs.
\end{enumerate}

Now we argue that the ideal execution with $\simulator$ is indistinguishable from the real execution with $\adv$.

{\bf Hybrid $\mathbf{H_1}$} is the real-world execution of $\decoPROT$.

{\bf Hybrid $\mathbf{H_2}$} is the same as $H_1$, except that $\simulator$ simulates $\adv$, $\idealZK$ and $\idealTwoPC$ internally. $\simulator$ records and forwards its private $\PRI$ input to $\adv$. For each step of $\decoPROT$, $\simulator$ forwards all messages between $\adv$ and $\verifier$ and $\adv$ and $\tlsserver$, as in the real execution. Since the simulation of ideal functionality is perfect, $H_1$ and $H_2$ are indistinguishable.

{\bf Hybrid $\mathbf{H_3}$} is the same as $H_2$, except that $\verifier$ sends input to $\idealOracle$, which sends it to $\simulator$ and $\simulator$ simulates $\verifier$ internally. Specifically, $\simulator$ samples $\hat s_V$ and uses $\hat s_V \cdot Y$ to derive a share of the MAC key $\hat K$, which it uses in the sequential 2PC-HMAC invocations. Upon receiving $(\sid, \hat Q, \hat R, \mackeyP)$, $\simulator$ sends $(\sid, \mackeyV)$ to $\adv$.  If $\simulator$ receives $(\sid, \msgprove, x, w)$, it internally forwards it to $\idealZK$, verifies its output as $\verifier$ and also, sends $\PRI$ to $\idealOracle$.
The indistinguishability between $H_2$ and $H_3$ is immediate because $\hat s_V$ is uniformly random.

{\bf Hybrid $\mathbf{H_4}$} is the same as $H_3$, except $\simulator$ adds the checks in Step \ref{item:checkStep}. The indistinguishability between $H_3$ and $H_4$ can be shown by checking that if any of the checks fails, $\verifier$ would abort the real-world execution as well. There are two reasons that $\simulator$ may abort: 1) $Q, R$ from $\adv$ is not originally from $\tlsserver$, or 2) $\enckey, \mackey$ from $\adv$ is not the same key as derived during the handshake. We now show that both conditions would trigger $\verifier$ to abort in $H_3$ as well except with negligible probability.
\begin{itemize}[leftmargin=*]
    \item Assuming DL is hard in the group used in the handshake, $\adv$ cannot learn $\hat s_V$. Furthermore, due to the security of 2PC, $\adv$ cannot learn the session MAC key $\mackey$. If $\adv$ maliciously selects ${\hat Y}_P$ correlated with ${\hat Y}_V$, it would have to find the discrete log of ${\hat{Y}_P - Y_V}$, denoted ${\hat s}_P$. Without such a ${\hat s}_P$, except with negligible probability, the output shares ${\hat K}^{\mathsf{MAC}}_{\verifier} $ and ${\hat K}^{\mathsf{MAC}}_{\prover}$ of 3P-HS would fail to verify a MAC from an honest server whose MAC key is derived using ${\hat Y}_P$ in 2PC-HMAC, later in the protocol.

    \item The unforgeability guarantee of HMAC ensures that without knowledge of $\mackey$, $\adv$ cannot forge tags that verifies against $\mackey$ (checked by $\verifier$ in the last step of $\decoPROT$).
    \item If $\adv$ sends a different $(\enckey, \mackey)$ pair than that derived during the handshake to $\simulator$ and the decryption and MAC check succeeds, then $\adv$ would have broken the receiver-binding property of CBC-HMAC~\cite{DBLP:conf/crypto/GrubbsLR17}.
\end{itemize}

It remains to show that $H_4$ is exactly the same the ideal execution. Due to Step \ref{item:checkStep} and \ref{item:sendStep}, $\idealOracle$ delivers $(\sid, \PRED(R), \tlsserver)$ to $\verifier$ only if $\exists \PRI$ from $\adv$ such that $R$ is the response from $\tlsserver$ to $\query(\PRI)$.

\boldhead{Malicious $\verifier$}
As the verifier is corrupt, we are interested in showing the verifier-integrity and privacy guarantees. $\simulator$ proceeds as follows:

\begin{enumerate}[leftmargin=*]
    \item $\simulator$ runs $\adv$, $\idealZK$ and $\idealTwoPC$ internally to simulate the real-world interaction with the prover $\prover$. Given input $z$ from the environment $\env$, $\simulator$ forwards it to $\adv$.
    \item Upon receipt of $\query$ and $\PRED$ from $\adv$, forward them to $\idealOracle$ and instruct it to send them to $\prover$.
    \item After $\prover$ sends $\PRI$ to $\idealOracle$, $\idealOracle$ sends the output $(\sid, Q, R)$ to $\prover$. $\simulator$ gets $(\sid, \PRED(R), \tlsserver)$ from $\idealOracle$ and learns the record sizes $|Q|$, $|R|$.
    \item Send $(\sid, \tlsserver, \mathsf{handshake})$ to $\idealOracle$, where $\mathsf{handshake}$ contains client handshake messages and receive certificate and signatures of $\tlsserver$ from $\idealOracle$. Note that at the end of the server handshake, $\prover$ receives and sends finished messages, which we denote $\msgfinS$ and $\msgfinP$. The finished messages include HMAC tags, which we denote $\tau_{\tlsserver}$ and $\tau_{\prover}$ (tags on $\tlsserver$ and $\prover$'s messages respectively).
    \item Upon request from $\adv$, $\simulator$ runs 3P-HS as $\prover$, using the server handshake messages received in the previous step, learning $s_P, Y_V$, $\enckey$, $\mackeyP$.
    \item $\simulator$ starts 2PC-HMAC as $\prover$ to compute a tag $\tau_q$ on a random $Q'\sample \{0, 1\}^{|Q|}$.
    \item $\simulator$ uses a random key $\hat k$ to compute a tag $\tau_r$ on a random $R' \sample \{0, 1\}^{|R|}$.
    \item\label{item:getStepV} Let $\hat Q = \CBC(\enckey, Q' \| \tau_q)$ and $\hat R = \CBC(\enckey, R' \| \tau_r)$. At the commit phase, $\simulator$ sends encrypted data $(\sid, \hat Q, \hat R, \mackeyP)$ to $\adv$ and receives $\mackeyV$ from $\adv$.
    \item\label{item:checkStepV1} $\simulator$ asserts that $\tau_{\tlsserver} = \mathsf{HMAC}(\mackey, \allowbreak \msgfinS)$ and that $\tau_{\prover} = \mathsf{HMAC}(\mackey, \msgfinP)$.
    \item\label{item:checkStepV2} $\simulator$ asserts that $\tau_q = \mathsf{HMAC}(\mackey, Q')$.
    \item To simulate the appropriate delay, $\simulator$ also runs a dummy computation $\mathsf{HMAC}(\mackey, R')$ in paralell with Step \ref{item:checkStepV1}.
    \item $\simulator$ sends $(\sid, \msgproof, 1, (\hat Q, \hat R, \mackey, \PRED(R)))$ to $\adv$ and outputs whatever $\adv$ outputs.
\end{enumerate}

We argue that the ideal execution with $\simulator$ is indistinguishable from the real execution with $\adv$ in a series hybrid worlds.

{\bf Hybrid $\mathbf{H_1}$} is the real-world execution of $\PROT_\text{\systemname}$.

{\bf Hybrid $\mathbf{H_2}$} is the same as $H_1$, except that $\simulator$ simulates $\idealZK$ and $\idealTwoPC$ internally. $\simulator$ also invokes $\idealOracle$ and gets  $(\sid, \PRED(R), \tlsserver)$, learns record sizes $|Q|$, $|R|$. Since the simulation of ideal functionality is perfect, $H_1$ and $H_2$ is indistinguishable.

{\bf Hybrid $\mathbf{H_3}$} is the same as $H_2$, except that $\simulator$ simulates $\prover$. Specifically, $\simulator$ samples $s_P$ and uses $s_P \cdot Y$ to derive a share of the MAC key $\mackeyP$. Then, $\simulator$ uses $\mackeyP$ and a random $Q' = \{0, 1\}^{|Q|}$ as inputs to 2PC-HMAC and receives the tag $\tau_q$. Then, $\simulator$ uses a random key $\hat k$, and a random $R' = \{0, 1\}^{|R|}$ to compute a dummy tag $\tau_r$. Afterwards, $\simulator$ commits, i.e., sends encryption of $Q'$ and $R'$ to $\adv$. $\simulator$ also adds the checks in Step \ref{item:checkStepV1} and \ref{item:checkStepV2}. To simulate the appropriate delay for checking a tag on $R'$, a plaintext of length $|R|$, $\simulator$ runs a dummy tag computation. Finally, $\simulator$ skips invoking $\idealZK$ and directly provides $\adv$ with the output obtained earlier from $\idealOracle$, i.e., $\PRED(R)$, alongwith $\mackey$, i.e. the tuple $(\sid, \msgproof, 1, (\hat Q, \hat R, \mackey, \PRED(R)))$. $\adv$ cannot distinguish between the real and ideal executions because:
\begin{enumerate}[leftmargin=*]
    \item Since input sizes are equal, the number of invocations of 2PC-HMAC is also equal.
    \item In each invocation of 2PC-HMAC and $\mathsf{HMAC}$, $\adv$ learns one SHA-2 hash of the input message which is like a random oracle.
    \item If the value of $\mackeyV$ provided by $\verifier$ is correct, in both the real and ideal world, all tags should verify and the protocol should proceed to the next step and the time to run the checks should be indistinguishable from the real world.
    \item $\adv$ can provide a malicious $\mackeyV$ in two ways:
    \begin{itemize}
        \item Malicious $\mackeyV$ is provided by $\verifier$ in Step \ref{item:getStepV}: $\tau_{\tlsserver}$ and $\tau_{\prover}$ will not verify in Step \ref{item:checkStepV1}. $\simulator$ will then abort with the same delay as in the real world.
        \item $\adv$ inputs a malicious $\mackeyV$ to the 2PC-HMAC: $\tau_q$ will fail to verify in \ref{item:checkStepV2} by the same argument as in the malicious $\prover$ case.
    \end{itemize}

    \item Since $|Q'| = |Q|$ and $|R'| = |R|$, their encryptions are also of equal size and indistinguishable.
    \item In the end, $\adv$ receives the same output as the real execution.
\end{enumerate}

\end{proof}
\section{Application Details}
\label{app: app details}

We provide the remaining application details omitted from~\cref{sec:applications} here.

\boldhead{Binary Option}\label{sec:binaryOpApp}
The user ($\prover$) also needs to reveal enough portion of the HTTP GET request to oracle ($\verifier$) in order to convince access to the correct API endpoint.
The GET request contains several parameters---some to be revealed like the API endpoint, and others with sensitive details like stock name and private API key.
$\prover$ redacts sensitive params using techniques from~\cref{sec:seletive opening} and reveals the rest to $\verifier$. The API key provides enough entropy preventing $\verifier$ from learning the sensitive params.
Without additional care though, a cheating $\prover$ can alter the semantics of the GET request and conceal the cheating by redacting extra parameters. To ensure this does not happen, $\prover$ needs to prove that the delimiter ``\&'' and separator ``='' do not appear in the redacted text. The security is argued below.

HTTP GET requests (and HTML) have a special restriction: the demarcation between a {\tt key} and a {\tt value} (i.e., \midd) and the start of a key-value pair (i.e., \start) are never substrings of a {\tt key} or a {\tt value}. This means that to redact more than a single contiguous {\tt key} or {\tt value}, $\prover$ must redact characters in \{\midd, \start\}. So we have $\cons_{\grammar, \grammar'}(R, R')$  check that: (1) $|R| = |R'|$; and (2) $\forall i\in |R'|$, either $R'[i] = D \land R[i] \notin \set{\text{\midd, \start}}$ or $R[i] = R'[i]$ ($D$ is a dummy character used to do in-place redaction). Checking $\ctx_\partgrammar$ is then unnecessary.

\boldhead{Age Proof}
\Cref{fig:age} shows the demographic details of a student stored on Univ. website such as the name, birth date, student ID among others. The prover parses 6-7 AES blocks that contain the birth date and proves her age is above 18 in ZK to the verifier. Like other examples, due to the unique HTML tags surrounding the birth date, this is also a key-value grammar with unique keys (see~\cref{subsec:two stage parsing}). Similar to application 1, this example requires additional string processing to parse the date and compute age.

\boldhead{Price discrimination}
\Cref{fig:shoppingorder} shows parts of an order invoice page on a shopping website (Amazon) with personal details such as the name and address of the buyer. The buyer wants to convince a third-party (verifier) about the charged price of a particular product on a particular date. In this example, we use AES-GCM ciphersuite and \REV{} mode. Only necessary details in the invoice like the item name, item price and order date are revealed, while hiding the rest. Number of AES blocks revealed from the response is 20 (thanks to a long product name). In addition, 4 AES blocks from the request are revealed to prove that the correct endpoint is accessed. Context integrity is guaranteed by revealing unique strings around, e.g., the string ``<tr>Order Total:'' near the item price appears only once in the entire response.
\deepak{Jasleen: it'd be good if you can check this argument once.}

\begin{figure}
\begin{lstlisting}[language=HTML]
<title>Demographic Data</title>
<span id='EMPLID'> 111111 </span>
<span id='NAME'> Alice </span>
|\colorbox{green!30}{<span id='BIRTHDATE'> 01/01/1990 </span>}|...
\end{lstlisting}
    \caption{The demographic details of a student displayed on a Univ. website. Highlighted text contains student age. \REV{} mode is used together with two-stage parsing.}
    \label{fig:age}
\end{figure}

\begin{figure}
\begin{lstlisting}[language=HTML]
<table>
|\colorbox{green!30}{<tr>Order Placed: November 23, 2018</tr>}|
|\colorbox{green!30}{<tr>Order Total: \$34.28</tr>}|
|\colorbox{green!30}{<tr>Items Ordered: Food Processor</tr>}|
</table>
...
<b> Shipping Address: </b>
<ul class="displayAddressUL">
<li class="FullName">Alice</li>
<li class="Address">Wonderland</li>
<li class="City">New York</li>
</ul>
\end{lstlisting}
\caption{The order invoice page on Amazon in HTML. \REV{} mode is used to reveal the necessary text, while sensitive text below is kept hidden.}
\label{fig:shoppingorder}
\end{figure}
\section{Key-Value Grammars and Two-Stage Parsing}
\label{app:keyvalue}
\subsection{Preliminaries and notation}
We denote context-free grammars as $\grammar=(V, \Sigma, P, S)$ where $V$ is a set of non-terminal symbols, $\Sigma$ a set of terminal symbols, $P: V\to (V\cup \Sigma)^*$ a set of productions or rules and $S\in V$ the start-symbol. We define production rules for CFGs in standard notation using `-' to denote a set minus and `..' to denote a range. For a string $w$, a parser determines if $w\in \grammar$ by constructing a parse tree for $w$. The parse tree represents a sequence of production rules which can then be used to extract semantics.

\subsection{Key-value grammars}
\label{sec:keyvalue grammar}
These are grammars with the notion of key-value pairs. These grammars are particularly interesting for \systemname since most API calls and responses are, in fact, key-value grammars.

\begin{definition}
$\grammar$ is said to be a key-value grammar if there exists a grammar $\mathcal{H}$, such that given any $s\in \grammar$, $s\in \mathcal{H}$, and $\mathcal{H}$ can be defined by the following rules:
\begin{spacing}{0.8}
\noindent
{\footnotesize
\texttt{\textbf{S} $\to$ object\\
\textbf{object} $\to$~noPairsString~open~pair~pairs~close \\
\textbf{pair} $\to$ \underline{start}~key~\underline{middle}~value~\underline{end}\\
\textbf{pairs} $\to$ pair pairs | ""\\
\textbf{key} $\to$ chars\\
\textbf{value} $\to$ chars | object\\
\textbf{chars} $\to$ char chars  | ""\\
\textbf{char} $\to$~Unicode~-~escaped~|~escape~escaped~|~addedChars\\
\textbf{special} $\to$ startSpecial | middleSpecial | endSpecial \\
\textbf{\underline{start}} $\to$ unescaped$_s$ startSpecial\\
\textbf{\underline{middle}} $\to$ unescaped$_m$ middleSpecial\\
\textbf{\underline{end}} $\to$ unescaped$_e$ endSpecial\\
\textbf{escaped} $\to$ special | escape | ...
}\par}
\end{spacing}
\label{def:kv}
\end{definition}

In~\cref{def:kv}, \texttt{S} is the start non-terminal (represents a sentence in $\mathcal{H}$), the non-terminals \texttt{open} and \texttt{close} demarcate the opening and closing of the set of key-value pairs and \texttt{\underline{start}, \underline{middle}, \underline{end}} are special strings demarcating the start of a key-value pair, separation between a key and a value and the end of the pair respectively.

In order to remove ambiguity in parsing special characters, i.e. characters which have special meaning in parsing a grammar, a special non-terminal, {\tt escape} is used. For example, in JSON, {\tt key}s are parsed when preceded by `whitespace double quotes' ({\tt ``}) and succeeded by double quotes. If a {\tt key} or {\tt value} expression itself must contain double quotes, they must be preceded by a backslash ({\tt \textbackslash}), i.e. escaped. In the above rules, the non-terminal {\tt unescaped} before special characters means that they can be parsed as special characters. So, moving forward, we can assume that the production of a key-value pair is unambigious. So, if a substring $R'$ of a string  $R$ in the key-value grammar $\grammar$  parses as a {\tt pair}, $R'$ must correspond to a {\tt pair} in the parse tree of $R$.

Note that in~\cref{def:kv}, \midd cannot derive an empty string, i.e. a non-empty string must mark \midd to allow parsing \texttt{key}s from \texttt{value}s. However, one of \start and \End can have an empty derivation, since they only demarcate the separation between \texttt{value} in one pair from \texttt{key} in the next. Finally, we note that in our discussion of two-stage parsing for key-value grammars, we only we consider permissible paths with the requirement that the selectively opened string, $\opening$ corresponds to a {\tt pair}.

\subsection{Two-stage parsing for a locally unique key}\label{sec:locallyUniqueKeyDef}
Many key-value grammars enforce key uniqueness within a scope. For example, in JSON, it can be assumed that keys are unique within a JSON {\tt object}, even though there might be duplicated keys across {\tt object}s.
The two-stage parsing for such grammars can be reduced to parsing a substring. Specifically, $\trans$ extracts from $R$ a continuous substring $R'$, such that the scope of a {\tt pair} can be correctly determined, even within $R'$. For instance, in JSON, if $\cons_{\grammar, \partgrammar}(R, R')$ returns $\true$ iff $R'$ is a prefix of $R$, then only parsing $R'$ as a JSON, up to generating the sub-tree yielding $\opening$ is sufficient for determining whether a string $\opening$ corresponds to the correct context in $R$. %

\subsection{Grammars with unique keys}
\label{sec:uniqueKeyDef}

Given a key-value grammar $\grammar$ we define a function which checks for uniqueness of keys, denoted $u_\grammar$. Given a string $s\in \grammar$ and another string $k$, $u_\grammar(s, k) = \text{\texttt{true}}$ iff there exists at most one substring of $s$ that can be parsed as \texttt{\underline{start}} $k$ \texttt{\underline{middle}}. Since $s\in \grammar$, this means, in any parse tree of $s$, there exists at most one branch with node \texttt{key} and derivation $k$. Let $\parser{\grammar}$ be a function that returns $\true$ if its input is in the grammar $\grammar$. We say a grammar $\grammar$ is a  \textit{key-value grammar with unique keys} if for all $s\in \grammar$ and all possible keys $k$, $u_\grammar(s, k)=\text{\texttt{true}}$, i.e. for all strings $R$, $C$:
\[
   \inference{\langle \parser{\mathcal{G}}, R\rangle \Rightarrow \true}{\langle u_\grammar, (R, C) \rangle \Rightarrow \true}.
\]

\newcommand{\genericLLParser}{\mathsf{ParsingAlgo}_\text{LL(1)}}
\newcommand{\fullParser}{\mathsf{ParsingAlgo}_\text{full}}
\newcommand{\partialParser}{\mathsf{ParsingAlgo}_\text{part}}
\newcommand{\ptable}{\mathsf{PT}}
\newcommand{\keyParser}{\mathsf{keyParser}}
\newcommand{\manageEscapes}{\mathsf{Escape}}
\newcommand{\partparser}[1]{\parser{#1}^{\text{part}}}
\newcommand{\genericparser}[1]{\parser{#1}^{\text{LL(1)}}}
\newcommand{\unique}{\mathcal{U}}
\newcommand{\pairGram}{\mathcal{P}}

\subsection{Concrete two-stage parsing for unique-key grammars}\label{sec:uniqueKeyProof}
Let $\unique$ be a unique-key grammar as given above. We assume that $\mathcal{U}$ is LL(1). This is the case for the grammars of interest in Section \ref{sec:applications}. See \cite{grune2007parsing} for a general LL(1) parsing algorithm.

We instantiate a context function, $\ctx_\unique$ for a set $T$, such that $T$ contains the permissible paths to a {\tt pair} for strings in $\unique$. We additionally allow $\ctx_\unique$ to take as input an auxiliary restriction, a {\tt key} ${\sf k}$ (the specified key in $\prover$'s output $\opening$). The tuple $(T, {\sf k})$ is denoted $S$ and $\ctx_\unique(S,\, \cdot\,,\, \cdot\,)$ as $\ctx_{\unique, S}$.

Let $\pairGram$ be a grammar given by the rule \texttt{\textbf{S$_\pairGram$} $\to$ pair}, where {\tt pair} is the non-terminal in the production rules for $\unique$ and {\tt S$_\pairGram$} is the start symbol in $\pairGram$. We define $\parser{\pairGram, {\sf k}}$ as a function that decides whether a string $s$ is in $\pairGram$ and if so, whether the {\tt key} in $s$ equals ${\sf k}$.
On input $R, \opening$, $\ctx_{\unique, S}$ checks that: (a) $\opening$ is a valid key-value pair with key ${\sf k}$ by running $\parser{\pairGram, {\sf k}}$ (b) $\opening$ parses as a key-value pair in $R$ by running an LL(1) parsing algorithm to  parse $R$.

To avoid expensive computation of $\ctx_{\unique, S}$ on a long string $R$, we introduce the transformation $\trans$, to extract the substring $R'$ of $R$, such that $R'=\opening$ as per the requirements.

For string $s, t$, we also define functions $substring(s, t)$, that returns $\true$ if $t$ is a substring of $s$ and $equal(s, t)$  which returns $\true$ if $s=t$. We define $\cons_{\unique, \pairGram}$ with the rule:
\[
    \inference{\langle substring(R, R') \rangle \Rightarrow \true {\, \, \, \,}\langle \parser{\pairGram, {\sf k}}, R'\rangle \Rightarrow \true}{\langle cons_{\unique, \pairGram}, (R, R')\rangle\Rightarrow \true}.
\]
and $S' = \{\text{\texttt{S}}_{\mathcal{P}}\}$. Meaning, $\ctx_{\pairGram}(S, R', \opening)=\true$ whenever $equal(R', \opening)$ and the rule
\begin{equation*}
\begin{gathered}
\inference{\langle equal, (R', \opening) \rangle \Rightarrow {\sf b}}{\langle \ctx_{\mathcal{P}}, (S', R', \opening) \rangle \Rightarrow {\sf b}} %
\end{gathered}
\end{equation*}
holds for all strings $R', \opening$.

\begin{claim}
$(\cons_{\unique, \pairGram}, S')$ are correct with respect to $S$.
\label{claim:unique key}
\end{claim}
\begin{proof}
We defer a formal proof and pseudocode for $\ctx_{\unique, S}$ to a full version, but the intuition is that if $R'$ is substring of $R$, a key-value pair $\opening$ is parsed by $\parser{\pairGram}$, then the same pair must have been a substring of $\unique$. Due to global uniqueness of keys in $\unique$,  there exists only one such pair $\opening$ and $\ctx_\unique(S, R, \opening)$ must be $\true$.
\end{proof} 
\end{document}